\documentclass[11pt]{amsart}
\usepackage{amsfonts,amssymb,amsmath}
\usepackage{amsthm}
\usepackage{thmtools}
\usepackage[top=0.9in, bottom=0.9in, left=0.9in, right=0.9in]{geometry}
\usepackage{natbib} 
\usepackage{setspace} 
\usepackage{enumerate}
\usepackage{fancyhdr}
\usepackage{subcaption}
\usepackage[utf8]{inputenc}
\usepackage{multirow,array}
\usepackage{graphicx}
\graphicspath{{figures/}} 
\usepackage{cancel}

\usepackage[T1]{fontenc}
\usepackage[utf8]{inputenc}
\usepackage{times}

\usepackage{import}
\usepackage{xifthen}
\usepackage{pdfpages}
\usepackage{transparent}

\usepackage{sgame,tikz}
\usetikzlibrary{cd}

\usepackage{enumitem}
\usepackage{url}
\usepackage{commath}
\usepackage{mathtools}
\usepackage[toc,title,titletoc,header]{appendix}

\usepackage{titletoc}
\newcommand\DoToC{%
\vskip1cm
\startcontents[sections]
\printcontents[sections]{l}{1}{\textbf{Table of Contents}\vskip3pt\hrule\vskip5pt\setcounter{tocdepth}{2}}
\vskip5pt\hrule\vskip5pt
}

\usepackage[colorlinks,backref=page]{hyperref} 
\usepackage{cleveref}

\usepackage{xcolor}
\usepackage{mathrsfs}
\usepackage{stmaryrd}
\definecolor{denim}{rgb}{0.08, 0.38, 0.74}
\hypersetup{
colorlinks,
linkcolor=denim,
citecolor=magenta,
urlcolor={blue!50!black}
}
\usepackage{float}
\usepackage{threeparttable}
\usepackage{adjustbox}
\usepackage{booktabs}
\usepackage{etoolbox} 
\usepackage{mathtools}
\usepackage[normalem]{ulem}

\allowdisplaybreaks

\renewcommand*{\backref}[1]{}
\renewcommand*{\backrefalt}[4]{%
\ifcase #1 %
\or        (Cited on page~#2.)%
\else      (Cited on pages~#2.)%
\fi}

\allowdisplaybreaks

\usepackage{tikz} \usetikzlibrary{cd}
\usepackage{pgfplots}
\usepgfplotslibrary{patchplots}
\usetikzlibrary{patterns, positioning, arrows}
\pgfplotsset{compat=1.15}

\usepackage{listings,xfp}
\usepackage{anyfontsize}

\makeatletter
{\small 
	\xdef\f@size@small{\f@size}
	\xdef\f@baselineskip@small{\f@baselineskip}
	\normalsize 
	\xdef\f@size@normalsize{\f@size}
	\xdef\f@baselineskip@normalsize{\f@baselineskip}
}
\newcommand{\smalltonormalsize}{%
	\fontsize
	{\fpeval{(\f@size@small+\f@size@normalsize)/2}}
	{\fpeval{(\f@baselineskip@small+\f@baselineskip@normalsize)/2}}%
	\selectfont
}
\makeatother
\numberwithin{equation}{section}

\usepackage[ruled,linesnumbered]{algorithm2e}
\usepackage[bottom]{footmisc}
\usepackage{bbold}

\newcommand{\bigOp}{\mathrm{O}_{\mathbb P}}
\newcommand{\smallOp}{\mathrm{o}_{\mathbb P}}

\newcommand{\one}{\mathbb{1}}

\newcommand{\Var}{\mathrm{Var}}

\setcitestyle{citesep={,},notesep={, }}
\makeatletter
\newcommand\MidSep{, }
\newcommand\LastSep{ and }
\newcommand\CiteList[1]{%
\let\last@elem\relax
\let\last@sep\relax
\@for\@list \equiv #1\do{%
	\ifx\last@elem\relax\else
	\ifx\last@sep\relax
	\def\last@sep{\LastSep}
	\else\MidSep  
	\fi
	\cite{\last@elem}%
	\fi
	\let\last@elem\@list
}
\ifx\last@elem\relax\else
\last@sep\cite{\last@elem}%
\fi
}
\newcommand{\customlabel}[2]{%
\protected@write \@auxout {}{\string \newlabel {#1}{{#2}{\thepage}{#2}{#1}{}} }%
\hypertarget{#1}{#2}
}
\newcommand{\newsectionstyle}{%
\renewcommand{\@secnumfont}{\bfseries}
\renewcommand\section{\@startsection{section}{2}%
	\z@{.5\linespacing\@plus.7\linespacing}{-.5em}%
	{\normalfont\scshape\bfseries}}%
	}
	\let\oldsection\section
	\let\old@secnumfont\@secnumfont
	\newcommand{\originalsectionstyle}{%
\let\@secnumfont\old@secnumfont
\let\section\oldsection
}
\makeatother

\newtheoremstyle{rhtheorem}
{5.5pt plus 2.0pt minus 4.0pt}
{10.0pt plus 2.0pt minus 4.0pt}
{\slshape}
{0pt}
{\bfseries}
{.}
{ }
{\thmname{#1}\thmnumber{ #2}\textnormal{\thmnote{ (#3)}}}

\usepackage{pdflscape}    
\usepackage{bbm}          

\numberwithin{equation}{section}
\theoremstyle{rhtheorem}

\newtheorem{theorem}{Theorem}[section]
\newtheorem{prop}[theorem]{Proposition}
\newtheorem{lemma}{Lemma}[section]

\newtheorem{assumption}{Assumption}[section]
\theoremstyle{definition}
\newtheorem{ex}{Example}[section]
\newtheorem{defn}{Definition}[section]
\newtheorem{remark}{Remark}[section]

\AtEndEnvironment{ex}{~\,$\diamondsuit$}
\AtEndEnvironment{remark}{~\,$\diamondsuit$}

\crefname{theorem}{Theorem}{Theorems}
\Crefname{theorem}{Theorem}{Theorems}
\crefname{defn}{Definition}{Definitions}
\Crefname{defn}{Definition}{Definitions}
\crefname{ex}{Example}{Examples}
\Crefname{ex}{Example}{Examples}
\crefname{cor}{Corollary}{Corollaries}
\Crefname{cor}{Corollary}{Corollaries}
\crefname{assumption}{Assumption}{Assumptions}
\Crefname{assumption}{Assumption}{Assumptions}
\crefname{lemma}{Lemma}{Lemmas}
\Crefname{lemma}{Lemma}{Lemmas}
\crefname{remark}{Remark}{Remarks}
\Crefname{remark}{Remark}{Remarks}
\crefname{prop}{Proposition}{Propositions}
\Crefname{prop}{Proposition}{Propositions}
\crefname{fact}{Fact}{Facts}
\Crefname{fact}{Fact}{Facts}

\newcommand{\E}{\mathbb{E}}
\newcommand{\R}{\mathbb{R}}
\newcommand{\pto}{\overset{\mathbb P}{\to}}
\newcommand{\dto}{\overset{d}{\longrightarrow}}

\newcommand{\iid}{\overset{\mathrm{i.i.d.}}{\sim}}

\newcommand{\Yit}{Y_{i,t}}
\newcommand{\Yitpo}[1]{Y_{i,t}(#1)}

\newcommand{\calG}{\mathcal{G}}
\newcommand{\calGtrg}{\calG_{\mathrm{trg}}}
\newcommand{\calS}{\mathcal{S}}
\newcommand{\calE}{\mathcal{E}}

\newcommand{\calX}{\mathcal{X}}

\newcommand{\Dit}{D_{i,t}}
\newcommand{\ATT}{\mathrm{ATT}}
\newcommand{\CATT}{\mathrm{CATT}}
\newcommand{\ES}{\mathrm{ES}}
\newcommand{\Stack}{\mathbb{S}}
\newcommand{\gps}[2]{p^{#1}(#2)}

\newcommand{\DeltaY}{\Delta Y}
\newcommand{\Prob}{\mathbb{P}}
\newcommand{\ESstack}{\widehat{\ES}_{\mathrm{stack}}}
\newcommand{\Vstack}{V_{\mathrm{stack}}}

\newcommand{\eg}{\textit{e.g.}}

\graphicspath{{figures/}{_Figures/}}

\title{Stacked Triple Differences}
\address{701 Tappan Avenue, Ross School of Business, University of Michigan, Ann Arbor, MI 48109}
\email{rexhsieh@umich.edu}
\author{Meng Hsuan Hsieh}
\date{\today}
\thanks{All errors are mine. I thank Andreas Hagemann, Jeongwon Jang, and Kaspar W\"uthrich for their extensive feedback. I thank Laura Kawano, Joel Slemrod, and all participants at the University of Michigan, Ross School of Business Brown Bag seminar for all their comments that greatly improved this paper.}
\date{\today}

\begin{document}

\maketitle

\begin{abstract}
	Triple differences (DDD) is a workhorse quasi-experimental design in applied economics. But, under staggered adoption, its conventional three-way fixed-effects (3WFE) implementation inherits the interpretation issues now well understood in the difference-in-differences literature.
	I introduce \textit{stacked DDD}. I extend the stacked difference-in-differences approach to the DDD setting by creating self-contained stacks, each consisting of four cells over an event window: treated and clean comparison cohorts, each with treatment-eligible and treatment-ineligible units. Appending these stacks yields a unified dataset for estimating treatment effects.
	I prove that, at each post-treatment event-time, a linear regression with fully saturated fixed-effects applied to the stacked dataset identifies a strictly positive, cell-size-weighted average of stack-level conditional average treatment effects, with stack weights proportional to stack-level cell sizes.
	Building on this characterization, I outline alternative weighting schemes that recover causal estimands with clear interpretations. Stacked DDD complements recent GMM and imputation-based frameworks by trading efficiency for regression-based transparency, pairwise (rather than global) parallel changes-in-trends, and direct control over both the comparison group for each treated unit and the aggregation weights. I provide two empirical illustrations where stacked DDD yields substantially different quantitative conclusions compared to existing procedures.

    \medskip

        
\end{abstract}

\section{Motivation and Introduction} \label{sec:intro}

Triple differences (DDD) is among the most widely used quasi-experimental research designs in applied economics. In settings where treatment requires satisfying two criteria---belonging to a group whose members become exposed to a policy, and being eligible within that group---DDD permits both group-specific and eligibility-specific departures from parallel changes-in-trends, and is therefore a more defensible identification strategy than difference-in-differences (DiD) whenever either margin alone is unlikely to follow parallel paths. This flexibility explains its prominence in public, labor, health, and environmental economics; see \citet{olden_triple_2022} for a systematic survey of applications. The corresponding formal econometric theory, however, has not received as much attention compared to DiD, which has been the subject of sustained reexamination \citep[including, but not limited to,][]{goodman-bacon_difference--differences_2021, callaway_difference--differences_2021-1, sant2020doubly, sun_estimating_2021, de_chaisemartin_two-way_2020, borusyak_revisiting_2023, dube_local_2025}. Only recently have \citet{strezhnev_decomposing_2023}, \citet{leventer_triple_2025}, \citet{caron_triple_2025}, and \citet{ortiz-santanna_triple_2025} examined the conventional practice of using DDD. They document that the three-way fixed effects (3WFE) specifications commonly used to implement DDD under staggered adoption target estimands contaminated by forbidden comparisons, mirroring the observation of \citet{goodman-bacon_difference--differences_2021} for DiD, and additional bias arises in the presence of covariates. Furthermore, they develop identification arguments under staggered adoption and heterogeneous treatment effects.

This paper develops the \emph{stacked triple-differences (stacked DDD)} framework, which is complementary to existing approaches in DDD. This is based on the stacking approach now well known in the DiD literature, used by \citet{cengiz_effect_2019, deshpande_who_2019, butters_how_2022, matsuzawa_minimum_2025} and others, and discussed formally by \citet{wing_stacked_2024}. The stacking approach works in the following manner. First, for each treatment-enabling group, design a clean comparison group that consists of units that are not treated within a defined event window. Second, create a dataset, called a stack, for each treated cohort and its clean comparison group. Finally, concatenate all stacks into a unified dataset, which I call the stacked dataset. By construction, in the DDD setting, there are both treatment-eligible and treatment-ineligible units within both treatment and control groups in each stack.

Estimating treatment effects on the stacked dataset then exploits three sources of identifying variation simultaneously. First, within each treated group, eligible units are subject to the policy while ineligible units are not, so differencing them removes any confounding trend common to the group regardless of eligibility status. Second, across the treatment and clean comparison groups, differencing outcomes removes eligibility-specific shocks that would be shared by eligible units in both cohorts. Third, difference post- and pre-treatment outcomes then isolates the treatment effects within each group-by-eligibility cell. By design, the stacking construction delivers a $2 \times 2 \times 2$ sub-experiment for every treatment cohort.\footnote{I will use ``stack'' and ``cohort'' interchangeably throughout, which are the nomenclature of the stacked DiD and staggered DiD literature, respectively.}
As such, units apparently lacking direct identifying power nonetheless play an essential role for identification in stacked DDD: ineligible units in the comparison cohort---neither treated nor eligible---anchor the baseline of the clean comparison group, and it is this anchor that permits the third difference to separate the treatment effect from an eligibility-specific trend. Without them, one cannot distinguish a genuine treatment effect from a group-specific shock to eligible units in the treated cohort.

From a practical perspective, stacked DDD have two advantages. One, each stack is constructed independently for each treatment-enabling group, so researchers are not forced to rely on a single control group. Instead, we can construct entirely different control groups for different treatment-enabling group. This is a useful feature, especially in settings where the identifying assumptions are more credible when applied to a specific set of comparison units (e.g. matched specifically on a vector of baseline covariates and outcomes).
Two, the implementation of stacked DDD is entirely regression-based, exactly like stacked DiD. In particular, the implementation of stacked DDD does not require imputation \citep{borusyak_revisiting_2023}, or optimal weighting or regression adjustment \citep{ortiz-santanna_triple_2025}.

A central question in the stacked DDD framework is what causal estimands does this approach identify. It is of interest to characterize the estimand that an event-study OLS regression targets when applied on the stacked dataset.
As a starting point, I first derive the estimand targeted by the standard DDD event-study OLS regression with 3WFE (unit, group-by-time, and eligibility-by-time fixed effects), and where the treatment timings are staggered. I show that its event-study coefficient at event-time $e$ is a linear combination of cohort- and event-time-specific conditional average treatment effects on the treated (ATTs): the own-event-time cohort weights sum to one but are not guaranteed to be non-negative for individual treatment cohorts, and the weights on {other} event-times $e' \neq e$ are generically nonzero. Only under homogeneous treatment effects across cohorts does this estimand collapse to the conventional ATT; aside from this edge case, the coefficient admits negative-weight contamination of the kind first documented for DiD by \citet{goodman-bacon_difference--differences_2021}.

In contrast, I show that the fully saturated OLS regression estimated on the stacked dataset targets a strictly positive, cell-size-weighted average of cohort-specific ATTs. Fully saturated, in the context of stacked DDD, means including stack-by-treatment status-by-time, and stack-by-eligibility-by-time fixed effects, in addition to event-time treatment indicators. The implied weights from fully saturated OLS regression are a function of cell sizes within each stack, therefore this estimator assigns large weights to stacks whose stack-level treatment effect estimates are likely more precise.

However, it is unclear if the weights implied by fully saturated OLS regression estimated on the stacked dataset lead to a causally interesting aggregation of cohort-level treatment effects. As such, building on the aforementioned characerization, I describe alternative aggregation schemes that can be applied to stacks to target estimands with more explicit population causal interpretations: cohort-size weights, which deliver a per-unit average treatment effect on the treated; and equal weights, which deliver a simple average of cohort-level ATTs. The stacked DDD framework, by design, addresses a concern raised by \citet{de_chaisemartin_two-way_2020} and \citet{de_chaisemartin_difference--differences_2024} that pooled estimators deliver estimands whose implicit weights are opaque and not necessarily non-negative.

The stacked design also weakens the identifying assumption: parallel changes-in-trends need hold only pairwise within each stack, between the treated group and its clean comparison group, and not globally across all pairs of treated and control groups. This is a weaker restriction than the global DDD parallel changes-in-trends assumed by pooled estimators. 
The stacked estimator is exactly recovered by a fully saturated OLS regression on the stacked dataset. I further show that the cluster-robust standard errors, {clustered at the level of treatment}, delivers valid inference in this setting. Specifically, I show that the cluster-robust standard errors automatically account for the cross-stack dependence induced by the same comparison units appearing in multiple sub-experiments. Therefore, I sidestep the need for deriving analytical standard error corrections \citep[e.g.][]{abadie_when_2022, mackinnon_cluster-robust_2023} when fully saturated OLS regression is used.
The main trade-off of stacked DDD is that while fully saturated OLS regression corresponds to a perfectly just-identified GMM estimator within each stack, this estimator is generally not efficient because, generally, some control units may be shared across multiple stacks.
Ultimately, I view stacked DDD as a complementary approach to the GMM procedure of \citet{ortiz-santanna_triple_2025} and the imputation-based approach of \citet{borusyak_revisiting_2023}: stacked DDD trades efficiency for regression-based transparency, pairwise rather than global parallel changes-in-trends, and explicit control over both the comparison units for each treatment-enabling group (i.e. within each stack) and the aggregation weights.

Two empirical applications illustrate the practical consequences from applying stacked DDD. First, applied to \citet{hansen_national_2023}'s study of genetically modified crop adoption, the stacked DDD estimator confirms the qualitative finding of positive yield effects but produces point estimates substantially smaller than the pooled 3WFE specification. Second, applied to \citet{shastry_vaccine_2025}'s evaluation of Gavi's vaccine program, the stacked DDD estimator reproduces the conclusion that Gavi-funded vaccines reduced child mortality from related causes, but attenuates the magnitude of the original study's estimates by roughly a third in the full sample and aligns with the authors' preferred vital-registry subset at a reduction of approximately a quarter of a death per 1,000 live births. In both applications, the stacked decomposition discloses which stack-level conditional average treatment effect on the treated drive the aggregated effect.

The remainder of the paper proceeds as follows. Section \ref{sec:framework} develops the notation and identifying assumptions. 
Section \ref{sec:hw_application} examines the current practice of 3WFE OLS regressions.
Section \ref{sec:regression} defines the stacked sub-experiment, the target estimands, and the within-stack and pooled regression specifications that recover them.
Section \ref{sec:identification} establishes identification under the DDD PCT restriction. Section \ref{sec:asymptotics} develops the asymptotic theory and inference. Section \ref{sec:empirical} reports the two illustrations of stacked DDD in practice. All proofs and extended discussions are collected in the appendices.

\section{Setup and Assumptions}\label{sec:framework}

\subsection{Notation}\label{sub:setup}

I observe panel data on $n$ units indexed by $i$ over a common set of $T$ time indexed by $t$. Each unit is observed in a subset $\mathcal{T}_i \subseteq \{1, \ldots, T\}$; the panel need not be balanced. Throughout, $T$ is fixed and all asymptotics are with respect to $n \to \infty$.

Each unit belongs to a treatment-enabling group $S_i \in \calS$, where
\begin{equation}\label{eq:calS}
	\calS \subseteq \{2, \ldots, T\} \cup \{\infty\}.
\end{equation}
The variable $S_i$ records the period at which unit $i$'s group first becomes exposed to the policy; a unit with $S_i = g$ for $g \in \{2, \ldots, T\}$ belongs to a group whose treatment becomes enabled in period $g$, and a unit with $S_i = \infty$ belongs to a group whose treatment is never enabled within the sample. Let $\calGtrg = \calS \setminus \{\infty\}$ denote the set of treated groups. Within each group, an eligibility indicator $Q_i \in \{0,1\}$ records whether the unit is itself eligible for treatment; I take $Q_i$ to be time-invariant throughout the panel, though the framework extends straightforwardly to time-varying eligibility. Treatment is received only by units in an activated enabling group who are themselves eligible, so the observed treatment status is
\begin{equation}\label{eq:treatment}
	\Dit = \one\{t \geq S_i\} Q_i ~,
\end{equation}
implying that units with $Q_i = 0$ are never treated irrespective of their group, and units with $S_i = \infty$ are never treated irrespective of their eligibility. 
Let $X_i \in \calX \subseteq \R^d$ denote a vector of pre-treatment, time-invariant covariates available for adjustment across groups and eligibility categories. Finally, I write $Y_{i,t}$ to denote the observed outcome of interest, and I will later use the notation $\DeltaY_{i,t} \equiv Y_{i,t} - Y_{i,g-1}$ to denote the long differenced outcome relative to the baseline period (to be defined clearly later).

Following \citet{robins_new_1986}, I index potential outcomes by treatment timing. Write $\Yitpo{g}$ for the potential outcome at time $t$ if unit $i$ is first treated at $g \in \calGtrg$, and $\Yitpo{\infty}$ for the potential outcome under no treatment.
By convention, I set $\Yitpo{g} \equiv \Yitpo{\infty}$ for all $t < g$, so that treatment-timing indexing is active only at and after first exposure.
The observed outcome satisfies
\begin{equation}\label{eq:observed_outcome}
	\Yit = \sum_{g \in \calGtrg} \one\{S_i = g\} Q_i \Yitpo{g} + \left(1 - \sum_{g \in \calGtrg} \one\{S_i = g\} Q_i\right) \Yitpo{\infty},
\end{equation}
so that units that actually got treated ($S_i = g$, $Q_i = 1$) realize $\Yitpo{g}$ while all remaining units---including ineligible units in treated groups and every unit in the never-treated group---realize $\Yitpo{\infty}$. This formulation embeds two maintained restrictions: treatment effects operate through the timing of onset, and ineligible units within treated groups are unaffected by the policy. The observed data for unit $i$ are therefore
\begin{equation}\label{eq:observed_data}
	W_i = \left(\{Y_{i,t}\}_{t \in \mathcal{T}_i},  S_i,  Q_i,  X_i\right),
\end{equation}
and $\{W_i\}_{i=1}^n$ is taken to be an i.i.d.\ sample from the population distribution of $W$.

DDD is ubiquitous in practice; I provide a few examples below.

\begin{ex}[ACA Medicaid expansion]\label{rmk:medicaid_setup}
	The Affordable Care Act (ACA) expanded Medicaid eligibility to adults with incomes below 138\% of the Federal Poverty Level, with states adopting the expansion at different times beginning in 2014 \citep{courtemanche_early_2017, kaestner_effects_2017}. In this application, $S_i$ records the year in which state $i$'s Medicaid expansion took effect, with $S_i = \infty$ for states that had not expanded by the end of the sample period. Eligibility is defined at the individual level: $Q_i = 1$ for adults with incomes below 138\% FPL (the newly eligible population) and $Q_i = 0$ for higher-income adults in the same state who are ineligible for Medicaid. Outcomes include health insurance coverage rates, emergency department visits, and health expenditures. The three differences are: (i) pre- versus post-expansion, (ii) expansion versus non-expansion states, and (iii) income-eligible versus income-ineligible populations.
\end{ex}

\begin{ex}[WTO accession and trade]\label{rmk:wto_setup}
	Countries have acceded to the WTO/GATT at different times since 1948, creating a staggered adoption setting for trade liberalization \citep{strezhnev_decomposing_2023}. In this application the unit $i$ is an \emph{ordered exporter--importer directed pair} $i = (e, m)$, with $e$ the exporting country and $m$ the importing country, so that the pair identity is fixed over the sample. $S_i$ records the year of country $i$'s WTO accession, with $S_i = \infty$ for non-member countries. Eligibility $Q_i \in \{0,1\}$ is an attribute of the directed pair: $Q_i = 1$ if the importer $m$ is also a WTO member throughout the sample window (so the exchange qualifies for Most Favored Nation treatment), and $Q_i = 0$ otherwise. Because both the exporter's WTO accession year $S_i$ and the importer's membership status are \emph{time-invariant} characteristics of the directed pair, $Q_i$ is time-invariant as required by \Cref{sub:setup}. The outcome $Y_{i,t}$ is the log of bilateral trade flows. The three differences are: (i) pre- versus post-accession, (ii) new-member versus non-member countries, and (iii) eligible (both-member) versus ineligible (one-member) trading pairs.
\end{ex}

\begin{ex}[EPA nonattainment designations]\label{rmk:epa_setup}
	Under the Clean Air Act, the U.S.\ Environmental Protection Agency (EPA) designates counties as ``nonattainment'' when they fail to meet national ambient air quality standards, with designations occurring at different times for different pollutants and counties \citep{greenstone_impacts_2002, walker_transitioning_2013}. In this application, $S_i$ records the year of county $i$'s nonattainment designation, with $S_i = \infty$ for counties that remained in attainment throughout the sample period. Eligibility is defined at the industry level: $Q_i = 1$ for polluting manufacturing industries subject to heightened regulatory scrutiny under nonattainment, and $Q_i = 0$ for non-polluting industries in the same county. Outcomes include manufacturing employment, plant openings and closings, and total factor productivity. The three differences are: (i) pre- versus post-designation, (ii) nonattainment versus attainment counties, and (iii) regulated versus unregulated industries.
\end{ex}

\subsection{Identification Assumptions}\label{subsec:assumptions}

I impose four assumptions in the ensuing identification analyses. The first three assumptions are standard in the difference-in-differences literature. The fourth is specific to the DDD setting.

\begin{assumption}[Random Sampling]\label{as:sampling}
	The observed data $\{W_i\}_{i=1}^n$ are independent and identically distributed:
	\begin{equation}\label{eq:sampling}
		\{W_i\}_{i=1}^n \iid F_W ~,
	\end{equation}
	where $W_i = (\{Y_{i,t}\}_{t=1}^T, S_i, Q_i, X_i) \in \R^T \times \calS \times \{0,1\} \times \calX$ and $F_W$ is the population distribution with $\E\|W_i\|^2 < \infty$.
\end{assumption}

\Cref{as:sampling} requires that the observed data consist of an i.i.d.\ draw from the joint distribution of outcomes across all time periods, the treatment-enabling group indicator $S_i$, the eligibility indicator $Q_i$, and pre-treatment covariates $X_i$. This assumption is standard in the panel data literature and is automatically satisfied when units are sampled randomly from the population of interest. It rules out spatial or network dependence across units but can be relaxed to allow for cluster-level dependence when inference is conducted at the cluster level.

\begin{assumption}[Overlap]\label{as:overlap}
	Every cell in the group-by-eligibility partition has positive probability:
	\begin{equation}\label{eq:overlap}
		\min_{\substack{s \in \calS,\\ q \in \{0,1\}}} \Prob(S_i = s,  Q_i = q) > 0~.
	\end{equation}
\end{assumption}

\Cref{as:overlap} ensures that every cell in the $S \times Q$ partition of units is populated. This is the DDD analogue of the standard overlap (or positivity) condition in the treatment effects literature, stated unconditionally. It requires that no enabling group or eligibility status has zero probability in the population. When covariates are incorporated (\Cref{app:covariates}), the overlap condition is strengthened to a conditional version requiring that the generalized propensity scores $\Prob(S_i = s, Q_i = q \mid X_i = x)$ are bounded away from zero uniformly over the covariate support.

\begin{assumption}[No Anticipation]\label{as:noanticipation}
	For all $g \in \calGtrg$ and $t < g$:
	\begin{equation}\label{eq:noanticipation}
		\Yitpo{g} = \Yitpo{\infty} \quad \text{almost surely.}
	\end{equation}
\end{assumption}

\Cref{as:noanticipation} requires that in periods before treatment is enabled for cohort $g$, the potential outcomes under treatment and under no treatment coincide. That is, units do not alter their behavior in anticipation of future treatment. This assumption is standard in the difference-in-differences literature \citep{callaway_difference--differences_2021-1} and is plausible in settings where the exact timing of treatment enablement is not known in advance or where institutional constraints prevent anticipatory behavior.

\begin{assumption}[DDD Parallel Changes-in-Trends (DDD-PCT)]\label{as:dddpct}
	For all $g \in \calGtrg$, all valid comparison groups $g_c > g$, and all $t \in \{2, \ldots, T\}$ with $t \leq g_c$:
	\begin{align}
		&\E[\Yitpo{\infty} - Y_{i,t-1}(\infty) \mid S_i = g, Q_i = 1] - \E[\Yitpo{\infty} - Y_{i,t-1}(\infty) \mid S_i = g, Q_i = 0] \notag \\
		&= \E[\Yitpo{\infty} - Y_{i,t-1}(\infty) \mid S_i = g_c, Q_i = 1] - \E[\Yitpo{\infty} - Y_{i,t-1}(\infty) \mid S_i = g_c, Q_i = 0]~. \label{eq:dddpct}
	\end{align}
\end{assumption}

\begin{assumption}[No Spillover]\label{as:nospillover}
	For every $g \in \calGtrg$ and every unit $i$ with $S_i = g$ and $Q_i = 0$, and for all $t \in \{1,\ldots,T\}$,
	\begin{equation}\label{eq:nospillover}
		Y_{i,t} = \Yitpo{\infty}~, \quad \text{almost surely.}
	\end{equation}
\end{assumption}

\begin{assumption}[Admissibility of Comparison Cohorts]\label{as:admissibility}
	For every treated cohort $g \in \calGtrg$, the set of admissible comparison cohorts
	\begin{equation}\label{eq:admissibility}
        \mathcal{C}(g) \equiv \bigl\{ g_c \in \calS : g_c > g + K ~\text{ and }~ \Prob(S_i = s, Q_i = q)>0 ~\text{ for all }~ s \in \{g, g_c\} \text{ and } q \in \{0,1\} \bigr\}
	\end{equation}
	is nonempty. 
\end{assumption}

\Cref{as:dddpct} is the core identifying assumption of the stacked DDD framework. It requires that the {difference} in untreated outcome trends between eligible ($Q_i = 1$) and ineligible ($Q_i = 0$) units is the same across the treated cohort $g$ and its comparison cohort $g_c$. This is a weaker condition than the parallel changes-in-trends assumptions required in standard DiD settings.
In particular, \Cref{as:dddpct} permits {group-specific trends}  $\E[\Yitpo{\infty} - Y_{i,t-1}(\infty) \mid S_i = g]$ to differ freely across enabling groups $g$, which is precisely the type of heterogeneity that invalidates parallel changes-in-trends in DiD settings. It also permits {eligibility-specific trends} to differ between eligible and ineligible units within the same treatment-enabling group. What is required is that this within-group differential trend be the {same} across groups.

\Cref{as:dddpct} as stated is already pairwise: it is indexed by a single treated cohort $g$ and a single comparison cohort $g_c$, and the identifying restriction is imposed one pair at a time. A pooled DDD regression compares eligible units against a single fixed comparison pool common to every treated cohort, which amounts to requiring \eqref{eq:dddpct} to hold simultaneously across every $(g, g_c)$ pair the pooled estimator mixes; that is the joint restriction the stacked design relaxes. The stacked estimator, by contrast, invokes \eqref{eq:dddpct} only for the $(g, g_c(g))$ pairs actually used to build each stack, and different stacks may use different $g_c$. 

\Cref{as:nospillover} is the SUTVA-type restriction that makes the within-group comparison of eligible and ineligible units a clean third difference. It fails under within-group spillovers (e.g. general equilibrium effects) that would alter the outcomes of ineligibles once eligibles are treated. When this assumption is suspect, the stacked DDD design must either be reframed around a comparison group that is unambiguously outside the spillover's reach, or supplemented with a spillover-robust correction such as partial-population designs.

Finally, \Cref{as:admissibility} states that for every treated cohort $g$, there is at least one cohort $g_c$ that (i) remains untreated throughout the event window $[g-L, g+K]$, and (ii) has positive mass in every $(S,Q)$-cell, and the combination of the two conditions guarantees that DDD within stack $g$ is well-defined. When there are never-treated groups, condition (i) holds automatically.

\begin{remark}[Conditional version]\label{rem:conditional_cpt}
	When pre-treatment covariates $X_i$ are available, \Cref{as:dddpct} can be strengthened to hold conditional on $X_i = x$ for almost all $x \in \calX$.
	This extension is discussed in \Cref{app:covariates}.
\end{remark}

\section{The Practice of Using Triple-Difference Event-Study Specification}\label{sec:hw_application}

\citet[][hereafter HW]{hansen_national_2023} study the impact of genetically modified (GM) crop adoption on agricultural yields using a triple-differences design with staggered adoption across countries and crops. Their event-study specification is a leading example of a pooled DDD regression.

\subsection{Mapping to the DDD Framework}\label{sub:hw_mapping}

In HW, the observational unit is a country--crop pair $(i,c)$, observed over years $t = 1, \ldots, T$. $Y_{ict}$ denotes an outcome of interest---for example, log agricultural yield. The event-study specification they use is
\begin{equation}\label{eq:hw_original}
	Y_{ict} = \delta_{it} + \gamma_{ci} + \lambda_{ct} + \sum_{\substack{j = -L \\ j \neq -1}}^{K} \alpha_j \one\{t - E_{ic} = j\} + \varepsilon_{ict}~,
\end{equation}
where $\delta_{it}$ are country-by-year fixed effects, $\gamma_{ci}$ are crop-by-country fixed effects, $\lambda_{ct}$ are crop-by-year fixed effects, $E_{ic}$ is the year when the first GM varieties of crop $c$ are harvested in country $i$, and the $\alpha_j$ are the event-study parameters of interest. The normalization is $\alpha_{-1} = 0$, anchoring the event study at the last pre-treatment period.

This specification maps into the framework of Section \ref{sec:framework} as follows. Relabel the country--crop pair as the ``unit,'' writing $i$ for a generic unit. Define the treatment-enabling group as the GM approval year, $S_i = E_{ic}$, and the eligibility indicator as $Q_i = \one\{\text{crop } c \text{ is a GM-eligible variety}\}$. The treatment indicator is $\Dit = \one\{t \geq S_i\} Q_i$, which turns on when the country has adopted GM for that crop {and} the crop is GM-eligible.

Under this relabeling, two features are worth noting.
First, the country-by-year fixed effects $\delta_{it}$ are finer than country-level fixed effects but are \emph{non-nested} with the group-by-time fixed effects $\delta_{S_i,t}$ along the crop dimension, because the group index $S_i = E_{ic}$ varies \emph{within} country (i.e. different crops in the same country may have different GM approval years). Country-by-year effects therefore absorb only the country-common piece of $\delta_{S_i,t}$; the residual group-by-time variation that differs across crops within a country is left unabsorbed.
Second, the crop-by-year fixed effects $\lambda_{ct}$ are strictly finer than (and therefore nest the variation exploited by) the binary eligibility-by-time effects $\eta_{Q_i,t}$, since crop identity is more granular than the binary eligible/ineligible classification.
Therefore, the HW specification in \eqref{eq:hw_original} therefore includes {all three} sets of two-way interactions that the correct DDD specification requires. However, despite this correct structure, the specification pools the event-time indicators $\one\{t - E_{ic} = j\}$ across all treatment cohorts, thereby imposing a common event-study path $\alpha_j$ across cohorts potentially with heterogeneous treatment effects.

In the notation of this paper, the HW specification \eqref{eq:hw_original} is the event-study version of the correct DDD regression \eqref{eq:correct_ddd}
\begin{equation}\label{eq:hw_paper_notation}
	Y_{i,t} = \alpha_i + \delta_{S_i, t} + \eta_{Q_i, t} + \sum_{\substack{j = -L \\ j \neq -1}}^{K} \alpha_j R_j(i,t) + \varepsilon_{i,t}~,
\end{equation}
where $R_j(i,t) = \one\{t - S_i = j\} Q_i$ is the DDD event-time indicator and I have absorbed the finer country-by-year and crop-by-year effects into their coarser counterparts for notational economy. The critical difference from the 3WFE event-study specification \eqref{eq:3wfe_eventstudy} is the inclusion of $\eta_{Q_i,t}$. I now characterize the estimand $\alpha_j$.

\subsection{Estimand of the Event-Study Specification \eqref{eq:hw_original}}\label{sub:hw_plim}

For any variable $Z_{i,t}$, let $\ddot{Z}_{i,t}$ denote its residual after projecting it onto the orthogonal complement of the space spanned by unit effects $\alpha_i$, group-by-time effects $\delta_{S_i,t}$, and eligibility-by-time effects $\eta_{Q_i,t}$. This is given by
\begin{equation}\label{eq:hw_demean}
	\ddot{Z}_{i,t} = Z_{i,t} - \overline{Z}_{i,\cdot} - \overline{Z}_{S_i,t} - \overline{Z}_{Q_i,t} + \overline{Z}_{S_i,\cdot} + \overline{Z}_{Q_i,\cdot} + \overline{Z}_{\cdot,t} - \overline{Z}_{\cdot\cdot}~,
\end{equation}
where the barred quantities are group means as defined in Section \ref{sub:missing_fe}. This is the inclusion-exclusion formula for three-way demeaning, to be stated later in \eqref{eq:fwl_correct}.

For each cohort $g \in \calGtrg$ and event-time $\ell$, define the cohort-specific event-time indicator $R_{g,\ell}(i,t) = \one\{S_i = g, Q_i = 1, t = g + \ell\}$ and the auxiliary regression
\begin{equation}\label{eq:hw_aux_regression}
	R_{g,\ell}(i,t) = \alpha_i + \delta_{S_i,t} + \eta_{Q_i,t} + \sum_{\substack{e = -L \\ e \neq -1}}^{K} \omega_{g,\ell}^{e,\star} R_e(i,t) + \upsilon_{i,t}~.
\end{equation}
The coefficient $\omega_{g,\ell}^{e,\star}$ measures how much of the variation in the cohort-specific indicator $R_{g,\ell}$ is captured by the aggregate event-time-$e$ indicator $R_e$ after partialling out all three sets of fixed effects.
These weights are estimable from the data without any assumptions on the outcome model. The following results are described in the following object.

Following \citet{sun_estimating_2021}, I define the \textit{cohort-average treatment on the treated} at event-time $e$ for cohort $g$ as
\begin{equation}\label{eq:catt}
	\CATT(g, e) \equiv \E\left[Y_{i,g+e}(g) - Y_{i,g+e}(\infty) \mid S_i = g, Q_i = 1\right]~.
\end{equation}
The CATT parameterization indexes treatment effects by (cohort, exposure duration) rather than (cohort, time), which is natural for studying how effects evolve with time since treatment onset, and the HW event-study specification \eqref{eq:hw_original} is parameterized in event-time.

\begin{prop}[Weights from auxiliary regression]\label{prop:hw_es_decomp}
	The coefficients $\omega_{g,\ell}^{e,\star}$ from the auxiliary regression \eqref{eq:hw_aux_regression} have the explicit form
	\begin{equation}\label{eq:hw_weights_explicit}
		\omega_{g,\ell}^{j,\star} = \mathbf{e}_j^{\top} \left(\sum_{t=1}^{T} \E \left[\ddot{\mathbf{R}}_{i,t} \ddot{\mathbf{R}}_{i,t}^{\top}\right]\right)^{-1} \E\left[\ddot{\mathbf{R}}_{i,g+\ell}  R_{g,\ell}(i, g+\ell)\right]~,
	\end{equation}
	where $\ddot{\mathbf{R}}_{i,t} = (\ddot{R}_e(i,t))_{e \neq -1}$ is the vector of three-way-demeaned event-time indicators, $\mathbf{e}_j$ is the unit vector selecting event-time $j$, and $\ddot{}$ denotes the three-way demeaning operator \eqref{eq:hw_demean}. These weights satisfy the following properties:
	\begin{enumerate}
		\item[{(i)}] own-period weights sum to one, $\sum_{g \in \calGtrg} \omega_{g,j}^{j,\star} = 1$~;
		\item[{(ii)}] other included periods sum to zero, $\sum_{g \in \calGtrg} \omega_{g,\ell}^{j,\star} = 0$ for each $\ell \neq j$, $\ell \neq -1$~;
		\item[{(iii)}] excluded period sums to negative one, $\sum_{g \in \calGtrg} \omega_{g,-1}^{j,\star} = -1$~;
		\item[{(iv)}] never-treated units receive zero weight, $\omega_{\infty,\ell}^{j,\star} = 0$ for all $j, \ell$~.
	\end{enumerate}
\end{prop}

\eqref{eq:hw_weights_explicit} follows from a direct application of the Frisch--Waugh--Lovell (FWL) theorem to \eqref{eq:hw_aux_regression}: $\omega_{g,\ell}^{j,\star}$ is the coefficient on the aggregate event-time-$j$ indicator $R_j$ when $R_{g,\ell}$ is regressed on the full vector of three-way-demeaned event-time indicators $\ddot{\mathbf{R}}_{i,t}$, so each weight is a linear functional of second moments of the $\ddot{R}_e$. There is no assumption on the outcome model, and property~{(iv)} is immediate: a never-treated unit has $R_\ell(i,t) \equiv 0$ for all $\ell$, so $\ddot{R}_\ell(i,t)$ aggregates to zero contribution in the numerator of \eqref{eq:hw_weights_explicit} at that cell. Additional remarks on the weights are discussed in \Cref{app:addl-remarks}.

\begin{remark}[{Variance-ratio representation of the weights}]  \label{rem:variance-ratio-rep}
	The matrix-form weight \eqref{eq:hw_weights_explicit} can be expressed as a ratio of variances using a second application of the FWL theorem. By FWL applied within the auxiliary regression \eqref{eq:hw_aux_regression}, the coefficient $\omega_{g,\ell}^{j,\star}$ on $R_j$ equals the bivariate regression coefficient obtained by first partialling out all other regressors (the three sets of fixed effects and all event-time indicators $R_{e'}$ for $e' \neq j$, $e' \neq -1$) from both the dependent variable $R_{g,\ell}$ and the regressor of interest $R_j$. Let $\widetilde{R}_j(i,t)$ denote the partial residual of $R_j(i,t)$ after projecting out $(\alpha_i, \delta_{S_i,t}, \eta_{Q_i,t})$ and all other event-time indicators $\{R_{e'}\}_{e' \neq j,   e' \neq -1}$. Then,
	\begin{equation}\label{eq:hw_variance_ratio}
		\omega_{g,\ell}^{j,\star} = \frac{\displaystyle\sum_{t=1}^{T} \E\left[\widetilde{R}_j(i,t) R_{g,\ell}(i,t)\right]}{\displaystyle\sum_{t=1}^{T} \E\left[\widetilde{R}_j(i,t)^2\right]}~.
	\end{equation}
	Since $R_{g,\ell}(i,t) = \one\{S_i = g, Q_i = 1, t = g + \ell\}$ is nonzero only at time $t = g + \ell$, the numerator collapses to a single time period
	\begin{equation}\label{eq:hw_numerator}
		\sigma_{j; g,\ell}^{\star} \equiv \sum_{t=1}^{T} \E\left[\widetilde{R}_j(i,t) R_{g,\ell}(i,t)\right] = \E\left[\widetilde{R}_j(i,  g+\ell) \one\{S_i = g, Q_i = 1\}\right] = p_{g,1} \widetilde{r}_{j; g,1}^{(g+\ell)}~,
	\end{equation}
	where $p_{g,1} = \Prob(S_i = g, Q_i = 1)$ is the population share of treated-eligible units in cohort $g$ and $\widetilde{r}_{j; g,1}^{(g+\ell)}$ is the value of the partial residual $\widetilde{R}_j$ evaluated at cell $(S_i = g, Q_i = 1)$ at time $t = g + \ell$. (The partial residual is constant within cells at a given time, since both $R_j$ and all projecting-out variables are cell-level indicators.) The denominator is the {partial residual variance}
	\begin{equation}\label{eq:hw_denominator}
		\sigma_j^{2,\star} \equiv \sum_{t=1}^{T} \E\left[\widetilde{R}_j(i,t)^2\right] = \sum_{t=1}^{T} \sum_{s \in \calS} \sum_{q \in \{0,1\}} p_{s,q} \left(\widetilde{r}_{j; s,q}^{(t)}\right)^2~,
	\end{equation}
	the total variation in the event-time-$j$ indicator that is {not} explained by the three sets of fixed effects or by any other event-time indicator. Therefore,
	\begin{equation}\label{eq:hw_weight_ratio}
		{\omega_{g,\ell}^{j,\star} = \frac{p_{g,1}\widetilde{r}_{j; g,1}^{(g+\ell)}}{\sigma_j^{2,\star}} }
	\end{equation}
	is the explicit form of the weight.
\end{remark}

\subsection{Estimand Under Additional Assumptions}\label{sub:hw_progressive}

I now add identifying assumptions progressively and analyze the estimands.

\begin{prop}[Under DDD-PCT only]\label{prop:hw_es_pt}
	Under {\Cref{as:dddpct}}, the population regression coefficient $\alpha_j$ is a linear combination of $\CATT(g,\ell)$ with the weights from Proposition~{\ref{prop:hw_es_decomp}}
	\begin{equation}\label{eq:hw_contamination}
		\alpha_j = \sum_{g \in \calGtrg}\sum_{\ell \neq -1} \omega_{g,\ell}^{j,\star} \CATT(g,\ell)~.
	\end{equation}
	Cross-period contamination {(}nonzero weights $\omega_{g,\ell}^{j,\star}$ for $\ell \neq j${)} and cross-cohort contamination {(}possibly negative own-period weights $\omega_{g,j}^{j,\star}$ for individual $g${)} are both present.
\end{prop}

\Cref{prop:hw_es_pt} says that when one estimates the HW event-study specification \eqref{eq:hw_paper_notation} on staggered-adoption data, the coefficient $\alpha_j$ in a linear combination of cohort-level CATTs running over {all} included event-times, not only over time $j$. This is property (i) from \Cref{prop:hw_es_decomp}. This complicates interpretation of the estimand absent additional assumptions.

\begin{prop}[Under DDD-PCT and no anticipation]\label{prop:hw_es_noanticip}
	Under {\Cref{as:noanticipation}} and {\Cref{as:dddpct}}, the population coefficient $\alpha_j$ satisfies
	\begin{equation}\label{eq:hw_noanticip}
		\alpha_j = \sum_{g \in \calGtrg}\sum_{\ell \geq 0} \omega_{g,\ell}^{j,\star} \CATT(g,\ell)~.
	\end{equation}
	$\CATT(g,\ell) = 0$ for $\ell < 0$, but for $j < 0$, $\alpha_j$ is generically nonzero because it depends on post-treatment $\CATT(g,\ell)$ for $\ell \geq 0$ through the weights $\omega_{g,\ell}^{j,\star}$.
\end{prop}

\Cref{prop:hw_es_noanticip} has a striking implication for applied work. A researcher using the HW specification who finds $\widehat{\alpha}_j \neq 0$ for $j < 0$ may incorrectly conclude that DDD PCT is violated, when in fact the pre-period coefficient reflects heterogeneous post-treatment effects contaminating the pre-treatment window through the implicit weights. Conversely, $\widehat{\alpha}_j \approx 0$ for $j < 0$ does not validate DDD-PCT.

\begin{prop}[Under DDD-PCT and treatment effect homogeneity]\label{prop:hw_es_homo}
	Under \Cref{as:noanticipation}, {\Cref{as:dddpct}} and the restriction that $\CATT(g,\ell) = \ATT_\ell$ for all $g$ (treatment effects depend on exposure duration but not on the cohort), the population coefficient $\alpha_j$ simplifies to
	\begin{equation}\label{eq:hw_homo}
		\alpha_j = \ATT_j ~.
	\end{equation}
\end{prop}

\Cref{prop:hw_es_homo} identifies the knife-edge conditions under which the HW event-study specification recovers an interpretable causal parameter: $\alpha_j$ from \eqref{eq:hw_paper_notation} requires DDD-PCT, no anticipation, and treatment effect homogeneity across treatment cohorts to admit a valid causal interpretation. Relating this to \Cref{prop:hw_es_decomp}, note that property~{(i)} is reassuring {only} under treatment effect homogeneity in the event-time dimension, i.e., $\CATT(g,j) = \ATT_j$ for all $g$; in that case property~{(ii)} guarantees that cross-period CATTs cancel and $\alpha_j = \ATT_j$ exactly. Under heterogeneous treatment effects, property~{(ii)} of \Cref{prop:hw_es_decomp} says that individual weights $\omega_{g,\ell}^{j,\star}$ for $\ell \neq j$ can be negative (and, as the variance-ratio representation in Remark \ref{rem:variance-ratio-rep} makes explicit, own-period weights $\omega_{g,j}^{j,\star}$ can also be negative for particular cohorts), and $\alpha_j$ becomes a {linear}---not convex---combination of CATTs running over both the cohort and event-time axes. The population regression coefficient therefore does not admit the interpretation of any weighted {average} treatment effect at event-time $j$.

\subsection{The Aggregated ATT}\label{sub:hw_att_agg}

Applied researchers often summarize the event study by computing an aggregated ATT, averaging the post-treatment coefficients with researcher-chosen weights. Define
\begin{equation}\label{eq:hw_agg_att}
	\widehat{\ATT}_{\mathrm{agg}} = \sum_{j=0}^{K} w_j \widehat{\alpha}_j~,
\end{equation}
where $w_j \geq 0$ are non-negative weights with $\sum_{j=0}^{K} w_j = 1$. Common choices are equal weights ($w_j = 1/(K+1)$) or exposure-duration-specific weights. The following proposition characterizes the probability limit of this estimand.

\begin{prop}[Aggregated ATT estimand]\label{prop:hw_agg_att}
	Under {\Cref{as:sampling}}--{\Cref{as:dddpct}}, the aggregated ATT \eqref{eq:hw_agg_att} satisfies
	\begin{equation}\label{eq:hw_agg_plim}
		\widehat{\ATT}_{\mathrm{agg}} \pto \sum_{g \in \calGtrg} \sum_{\ell \geq 0} \Omega_{g,\ell} \CATT(g,\ell)~,
	\end{equation}
	where the aggregated weights are
	\begin{equation}\label{eq:hw_agg_weights}
		\Omega_{g,\ell} = \sum_{j=0}^{K} w_j \omega_{g,\ell}^{j,\star}~.
	\end{equation}
	The aggregated weights $\Omega_{g,\ell}$ need not be non-negative, but they satisfy the normalization
	\begin{equation}\label{eq:hw_agg_weights_sum}
		\sum_{g \in \calGtrg} \sum_{\ell \geq 0} \Omega_{g,\ell} = 1~.
	\end{equation}
\end{prop}

\section{Estimands, and What Does Stacked OLS Identify?}\label{sec:regression}

Before developing the formal identification and estimation theory, I analyze the regression specifications that applied researchers commonly use for triple differences. This section characterizes exactly what each specification targets, under what conditions it recovers an interpretable causal parameter, and where the standard approach breaks down. I begin by deriving the two regression specifications that correctly target the causal estimands of interest, providing rigorous proofs for each claim. I then analyze commonly used specifications that fail to recover these estimands, diagnosing the precise source of each failure. Throughout this section I draw explicit parallels with the interaction-weighted framework of \citet{sun_estimating_2021} for difference-in-differences and the decomposition results of \citet{de_chaisemartin_two-way_2020}. The analysis proceeds without covariates; the covariate-adjusted framework is developed in Appendix \ref{app:covariates}.

\subsection{Stacked sub-experiments}

The central construction of the stacked DDD approach is that of the {stacked sub-experiment}. The following definition makes this precise.

\begin{defn}[Stacked Sub-Experiment]\label{def:stack}
	For each treatment cohort $g \in \calGtrg$ and comparison group $g_c \in \calS$ with $g_c > g$, the $g$-specific stack $\Stack_g = \Stack(g, g_c, L, K)$ consists of:
	\begin{enumerate}
		\item \textbf{Event window}  $\{g - L, \ldots, g - 1, g, \ldots, g + K\}$, where $L \geq 1$ pre-treatment and $K \geq 0$ post-treatment periods.
		\item \textbf{Treated units}  Units with $S_i = g$ and $Q_i = 1$ (the cohort receiving treatment);
		\item \textbf{Within-group controls}  Units with $S_i = g$ and $Q_i = 0$ (same enabling group, ineligible);
		\item \textbf{Clean comparison group, eligible}  Units with $S_i = g_c$ and $Q_i = 1$;
		\item \textbf{Clean comparison group, ineligible}  Units with $S_i = g_c$ and $Q_i = 0$;
	\end{enumerate}
	The clean comparison group satisfies $g_c > g + K$, so that no unit in the comparison group is treated during the time window of the stack.
\end{defn}

Each stack $\Stack_g$ is a self-contained $2 \times 2 \times 2$-like experiment. The four cells are defined by the interaction of two binary dimensions:
\begin{equation}\label{eq:four_cells}
	\underbracket{(S_i = g, Q_i = 1)}_{\text{treated}}, \quad
	\underbracket{(S_i = g, Q_i = 0)}_{\substack{\text{within-group}\\\text{control}}}, \quad
	\underbracket{(S_i = g_c,  Q_i = 1)}_{\substack{\text{comparison}\\\text{eligible}}}, \quad
	\underbracket{(S_i = g_c,  Q_i = 0)}_{\substack{\text{comparison}\\\text{ineligible}}}.
\end{equation}

\subsection{Estimands of Interest}\label{sub:estimands}

I now define the target parameters for the stacked DDD framework, building from cohort-specific treatment effects to aggregated event-study parameters.

\paragraph{Group-time average treatment effect.}
The fundamental building block is the group-time average treatment effect on the treated. For each treatment cohort $g \in \calGtrg$ and post-treatment period $t \geq g$, define
\begin{equation}\label{eq:att_gt}
	\ATT(g,t) = \E\left[\Yitpo{g} - \Yitpo{\infty} \mid S_i = g, Q_i = 1\right].
\end{equation}
This is the average causal effect of treatment for units in cohort $g$ who are eligible ($Q_i = 1$), measured at time $t$. The parameter $\ATT(g,t)$ is indexed by both the treatment cohort and time, and thus accommodates heterogeneity in treatment effects across cohorts and over time. This quantity is the DDD analogue of the group-time ATT defined by \citet{callaway_difference--differences_2021-1} in the standard DiD setting. Finally, note that $\CATT(g, e) = \ATT(g, g+e)$ under the identity $t = g + e$, so the two parameters refer to the same causal contrast.

Within a given stack $\Stack_g = \Stack(g, g_c, L, K)$, I define the \emph{stack-specific treatment effect}
\begin{equation}\label{eq:att_stack}
	\ATT_{\Stack_g, g_c}(g,t) ~,
\end{equation}
which is the DDD estimand formed from the four cells in \eqref{eq:four_cells} within the stack's time window. This is the difference-in-difference-in-differences contrast:
\begin{multline}\label{eq:ddd_contrast}
	\ATT_{\Stack_g, g_c}(g,t) = \E\left[\Delta \Yit \mid S_i = g, Q_i = 1\right] - \E\left[\Delta \Yit \mid S_i = g, Q_i = 0\right] \\
	- \Big\{\E\left[\Delta \Yit \mid S_i = g_c,  Q_i = 1\right] - \E\left[\Delta \Yit \mid S_i = g_c,  Q_i = 0\right]\Big\},
\end{multline}
where $\Delta \Yit = Y_{i,t} - Y_{i,g-1}$ is the outcome change relative to the last pre-treatment period.

Applied researchers typically summarize treatment effects along the event-time axis. For event-time $e = t - g$ (periods relative to treatment onset), define the population event-study parameter
\begin{equation}\label{eq:es}
	\ES(e) = \sum_{\substack{g \in \calGtrg:\\ g + e \in \{1,\ldots,T\}}} \omega_g(e) \ATT(g, g + e), \qquad e \geq 0,
\end{equation}
where $\omega_g(e) \geq 0$ are aggregation weights satisfying $\sum_g \omega_g(e) = 1$ for each $e$. Different weight choices yield different summary parameters.

The stacked event-study estimand aggregates stack-level treatment effect estimates across treated cohorts, which is defined as
\begin{equation}\label{eq:es_stack}
	\ES_{\mathrm{stack}}(e) = \sum_{\substack{g \in \calGtrg:\\ g + e \in \{g - L,\ldots,g+K\}}} \omega_g(e) \ATT_{\Stack_g, g_c(g)}(g, g + e),
\end{equation}
where $g_c(g)$ denotes the comparison group used for cohort $g$ and the summation ranges over cohorts for which event-time $e$ falls within the stack's window. The constraint $-L \leq e \leq K$ restricts attention to event-times within the symmetric window shared by all stacks.

\begin{remark}[Regression interpretation]\label{rem:regression_preview}
	The DDD estimand \eqref{eq:ddd_contrast} has an exact regression counterpart. Within each stack $\Stack_g$, the fully saturated OLS regression of $\DeltaY_{i,t}$ on the four cell indicators---intercept, group indicator $\one\{S_i = g\}$, eligibility indicator $\one\{Q_i = 1\}$, and their interaction $\one\{S_i = g, Q_i = 1\}$---recovers the sample triple difference as the coefficient on the interaction term. This equivalence means that applied researchers can implement the stacked DDD estimator by running a separate OLS regression within each stack and aggregating the resulting coefficients. The key requirement is that each stack-level regression must be fully saturated in the cell indicators; imposing common slopes across cells or pooling across stacks with a single treatment coefficient alters the target estimand, as I explore in Appendix \ref{app:covariates}.
\end{remark}

\subsection{The Saturated Within-Stack Regression}\label{sub:saturated}

Consider a single stack $\Stack_g$ with four cells defined by $(S_i, Q_i) \in \{(g,1), (g,0), (g_c,1), (g_c,0)\}$. Write $n_{s,q} = \sum_{i \in \Stack_g} \one\{S_i = s,  Q_i = q\}$ for the number of units in cell $(s,q)$, and $n_g = \sum_{s,q} n_{s,q}$ for the total number of units in the stack. For a fixed post-treatment period $t \geq g$, the fully saturated regression on long differences $\DeltaY_{i,t}$ is
\begin{equation}\label{eq:saturated}
	\DeltaY_{i,t} = \mu_{g,t} + \lambda_{g,t} \one\{S_i = g\} + \eta_{g,t} \one\{Q_i = 1\} + \tau_{g,t}^{\mathrm{sat}} \one\{S_i = g, Q_i = 1\} + \varepsilon_{i,t}~.
\end{equation}
The coefficients $(\mu_{g,t}, \lambda_{g,t}, \eta_{g,t}, \tau_{g,t}^{\mathrm{sat}})$ are subscripted by $(g,t)$ to emphasize that they are population parameters specific to the stack built around cohort $g$ evaluated at time $t$. The regression has four parameters for four cell means and is exactly identified. The following proposition establishes that the OLS interaction coefficient recovers the sample triple difference.

\begin{prop}[Saturated regression recovers triple difference]\label{prop:saturated_ols}
	In regression \eqref{eq:saturated}, the OLS coefficients are
	\begin{align}
		\widehat{\mu}_{g,t} &= \overline{\DeltaY}_{g_c,0,t}~, \label{eq:ols_mu} \\
		\widehat{\lambda}_{g,t} &= \overline{\DeltaY}_{g,0,t} - \overline{\DeltaY}_{g_c,0,t}~, \label{eq:ols_lambda} \\
		\widehat{\eta}_{g,t} &= \overline{\DeltaY}_{g_c,1,t} - \overline{\DeltaY}_{g_c,0,t}~, \label{eq:ols_eta} \\
		\widehat{\tau}_{g,t}^{\mathrm{sat}} &= \underbracket{\left(\overline{\DeltaY}_{g,1,t} - \overline{\DeltaY}_{g,0,t}\right)}_{\text{within-group DD}} - \underbracket{\left(\overline{\DeltaY}_{g_c,1,t} - \overline{\DeltaY}_{g_c,0,t}\right)}_{\text{across-group DD}}~, \label{eq:triple_diff_ols}
	\end{align}
	where $\overline{\DeltaY}_{s,q,t} = n_{s,q}^{-1}\sum_{i\colon S_i = s,  Q_i = q} \DeltaY_{i,t}$ is the cell-specific sample mean. Under \Cref{as:sampling}--\Cref{as:admissibility}, $\widehat{\tau}_{g,t}^{\mathrm{sat}} \pto \ATT(g,t)$.
\end{prop}

The saturated regression \eqref{eq:saturated} requires no functional form assumptions whatsoever. It is purely design-based; the four-cell structure of the stack, combined with the long-differencing that removes time-invariant heterogeneity, delivers the triple difference mechanically. Each coefficient has a transparent interpretation. The intercept $\mu_{g,t}$ is the mean outcome change for comparison-ineligible units, $\lambda_{g,t}$ is the group differential for ineligible units, $\eta_{g,t}$ is the eligibility premium in the comparison group, and $\tau_{g,t}^{\mathrm{sat}}$ is the excess eligibility premium in the treated group---the triple difference. This parallels the two-cell structure of the saturated DiD regression, where $\tau$ captures the excess treatment-group change relative to the control; here, the additional eligibility dimension requires four cells instead of two, but the logic is identical.

\subsection{Pooled Stacked Regressions and Their Targets}\label{sub:pooled}

The within-stack saturated regression \eqref{eq:saturated} can be pooled across stacks. I first analyze the fully unrestricted pooled regression, which correctly targets the cohort-time treatment effects, and then examine restricted specifications that impose various degrees of effect homogeneity.

{The unrestricted pooled regression.} Pool all stacks into a single dataset, where each observation $(i, t, g)$ belongs to stack $\Stack_g$ with $S_i \in \{g, g_c(g)\}$ and $t \in \{g - L, \ldots, g + K\}$. The fully unrestricted regression with stack fixed effects and stack-by-cell fixed effects is
\begin{equation}\label{eq:full_hetero}
	\DeltaY_{i,t,g} = \lambda_{s,t,g} + \eta_{q,t,g} + \sum_{g' \in \calGtrg} \sum_{e=-L}^{K} \tau_{g',e} \one\{g = g',  S_i = g',  Q_i = 1, t = g' + e\} + \varepsilon_{i,t,g}~,
\end{equation}
where the subscript $g$ on the left indexes the stack from which the observation is drawn, $\lambda_{s,t,g}$ are stack-by-group-by-time fixed effects, $\eta_{q,t,g}$ are stack-by-eligibility-by-time fixed effects, and $\tau_{g',e}$ is a cohort-by-event-time coefficient.b

\begin{prop}[Unrestricted pooled regression equivalence]\label{prop:pooled_equiv}
	Under \Cref{as:sampling}--\Cref{as:admissibility}, The fully unrestricted regression \eqref{eq:full_hetero} is numerically identical to running the saturated regression \eqref{eq:saturated} separately within each stack for each time period. Specifically, $\widehat{\tau}_{g',e} = \widehat{\tau}_{g',g'+e}^{\mathrm{sat}}$ for all $g' \in \calGtrg$ and $e \in \{-L, \ldots, K\}$.
\end{prop}

\Cref{prop:pooled_equiv} establishes that the unrestricted pooled regression \eqref{eq:full_hetero} is a ``master regression'' that simultaneously delivers all within-stack triple differences and all within-stack pre-trend diagnostics. The pre-treatment coefficients provide a regression-based implementation of the pre-trend test---systematic departures of $\widehat{\tau}_{g',e}$ from zero for $e < 0$ signal violations of DDD-PCT.

The researcher forms the event-study parameter via post-estimation aggregation
\begin{equation}\label{eq:post_agg}
	\ESstack(e) = \sum_{g \in \calGtrg(e)} \widehat{\omega}_g(e) \widehat{\tau}_{g,e}
\end{equation}
with explicit, researcher-chosen weights $\widehat{\omega}_g(e) \geq 0$ summing to one. The practical question becomes how to transparently aggregate the cohort-level estimates.
The upshot here is that estimating a {fully-saturated} event-study linear regression on the stacked dataset targets a weighted average of cohort-specific effects with \textit{strictly} positive weights. A stacked event-study regression is fully saturated when it includes a full set of stack-by-group-by-time and stack-by-eligibility-by-time fixed effects, absorbing all cell-specific time shocks. In long differences, this specification is
\begin{equation}\label{eq:stacked_es_reg}
	\DeltaY_{i,t,g} = \lambda_{s,t,g} + \eta_{q,t,g} + \sum_{\substack{e=-L \\ e \neq -1}}^{K} \tau_e \one\{S_i = g, Q_i = 1, t = g + e\} + \varepsilon_{i,t,g}~,
\end{equation}
where $\lambda_{s,t,g}$ is a fixed effect for group $s \in \{g, g_c\}$ at time $t$ in stack $g$, and $\eta_{q,t,g}$ is a fixed effect for eligibility type $q \in \{0, 1\}$ at time $t$ in stack $g$. Because $\lambda_{s,t,g}$ and $\eta_{q,t,g}$ consume exactly three degrees of freedom for the four cells in stack $g$ at time $t$, the addition of the treatment indicator perfectly saturates the four cell means at every period.

The following proposition establishes that this fully-saturated stacked event-study regression is sensible: at each post-treatment event-time $e$, its coefficient is a cell-size-weighted average of the cohort-specific treatment effects with strictly positive weights summing to one. This parallels the result of \citet{sun_estimating_2021} for the fully-saturated stacked DiD event-study regression, extended here to the triple-differences setting with the additional eligibility dimension.

\begin{prop}[Estimands of the fully-saturated event-study regression]\label{prop:stacked_es_estimands}
	In the fully-saturated stacked event-study regression \eqref{eq:stacked_es_reg}, the OLS coefficient $\widehat{\tau}_e$ for each event-time $e \neq -1$ is numerically identical to a weighted average of the within-stack saturated coefficients $\widehat{\tau}_{g,g+e}^{\mathrm{sat}}$
	\begin{equation}\label{eq:stacked_es_weights}
		\widehat{\tau}_e = \sum_{g \in \calGtrg(e)} w_g^{\mathrm{FWL}}(e) \widehat{\tau}_{g,g+e}^{\mathrm{sat}}~,
	\end{equation}
	where the weights $w_g^{\mathrm{FWL}}(e)$ are strictly positive {under \Cref{as:shares}} and sum to one across $g \in \calGtrg(e)$. The weight for stack $g$ is proportional to the FWL residual variance of the treatment indicator within that stack at time $t = g+e$. Under \Cref{as:sampling}--\Cref{as:admissibility} and \Cref{as:shares},
	\begin{enumerate}
		\item[{(i)}] {No cross-cohort contamination.} The weights $w_g^{\mathrm{FWL}}(e) > 0$.
		\item[{(ii)}] {No cross-period contamination.} For any {fixed} event-time $e$, $\widehat{\tau}_e \pto \sum_{g} w_g^{\mathrm{FWL}}(e) \CATT(g,e)$. Treatment effects from other event-times $e' \neq e$ receive exactly zero weight.
	\end{enumerate}
\end{prop}

\Cref{prop:stacked_es_estimands} highlights a advantage of the stacked DDD framework. While researchers are encouraged to compute the unrestricted cohort-specific estimates and explicitly aggregate them using transparent weights (such as cohort-size weights, as in Section \ref{sub:pooled}), running the fully-saturated stacked event-study regression provides a safe, simple alternative. It guarantees a sensible weighted average of cohort-level treatment effects, in the sense that it avoids the negative weights phenomenon entirely.

\begin{remark}[Implementing the stacked DDD as a linear regression]\label{rem:implementation_regression}
	For practitioners, the stacked DDD event-study estimator can be implemented as a single linear regression on the stacked dataset. After constructing the stacked dataset (one observation per unit $\times$ time period $\times$ stack), the regression specification \eqref{eq:stacked_es_reg} includes the following fixed effects and treatment indicators.
	\begin{enumerate}
		\item[{(FE1)}] \textbf{Stack $\times$ treatment status $\times$ time} $\lambda_{s,t,g}$;
		\item[{(FE2)}] \textbf{Stack $\times$ eligibility $\times$ time fixed effects} $\eta_{q,t,g}$;
		\item[{(TRT)}] \textbf{Event-time treatment indicators}.
	\end{enumerate}
	Within each stack $g$ at each time $t$, the four cells $(g,1), (g,0), (g_c,1), (g_c,0)$ have three degrees of freedom absorbed by FE1 and FE2 (the group indicator $\one\{S_i = g\}$, the eligibility indicator $\one\{Q_i = 1\}$, and a constant are jointly determined by $\lambda_{s,t,g}$ and $\eta_{q,t,g}$). The treatment indicator adds one more degree of freedom, thereby perfectly saturating the four cell means. This saturation is what ensures the regression recovers the triple difference, and guarantees a valid causal interpretation on the event-study coefficients (on the event-time treatment indicators).
\end{remark}

\subsection{Population estimands under different aggregation schemes}

The stack DDD event-study estimator \[\ESstack(e) = \sum_{g \in \calGtrg(e)} \widehat{\omega}_g(e)  \widehat{\ATT}_{\Stack_g}(g, g+e)\] depends on the aggregation weights $\widehat{\omega}_g(e) \geq 0$, $\sum_g \widehat{\omega}_g(e) = 1$. Different weight choices target different population parameters, and the stacked DDD framework makes these targets explicit. I now formally characterize what each weighting scheme identifies when the within-stack estimators are consistent for $\ATT(g, g+e) = \CATT(g,e)$ (\Cref{thm:identification}).

\subsubsection*{Cohort-size weights.} Set $\widehat \omega_g^{\mathrm{cohort}}(e) = n_{g,1} / \sum_{g' \in \calGtrg(e)} n_{g',1}$, where $n_{g,1}$ is the number of treated-eligible units in stack $\Stack_g$ and $\calGtrg(e) = \{g \in \calGtrg \mid g + e \leq T\}$ is the set of cohorts observed at event-time $e$. The population analogue is $\omega_g^{\mathrm{cohort}}(e) = \Prob(S_i = g \mid S_i \in \calGtrg(e),  Q_i = 1)$, the share of cohort $g$ among treated-eligible units observed at event-time $e$. The resulting estimand is the per-capita average treatment effect
\begin{equation}\label{eq:cohort_estimand}
    \ES^{\mathrm{cohort}}(e) \equiv \sum_{g \in \calGtrg(e)} \omega_g^{\mathrm{cohort}}(e) \CATT(g,e) = \E\left[\ATT(S_i,  S_i + e) \mid S_i + e \leq T,  Q_i = 1\right]~,
\end{equation}
which averages the cohort-specific effects weighted by the number of treated-eligible individuals in each cohort. This parameter has a direct welfare interpretation as the average effect experienced by a randomly drawn treated-eligible unit at event-time $e$, and coincides with the average effect of switching defined and discussed in \Cref{sub:existing_ddd}. It is the natural estimand when the policy question concerns aggregate impact---how much did the treated population benefit, on average? As \citet{de_chaisemartin_difference--differences_2024} emphasize, this per-capita parameter is the most policy-relevant estimand in many applications, since it reflects the actual distribution of treatment effects across the affected population.

\subsubsection*{Equal weights.} Set $\omega_g^{\mathrm{eq}}(e) = 1/|\calGtrg(e)|$, assigning each cohort equal influence regardless of its sample size. The target is the simple average of cohort-specific effects
\begin{equation}\label{eq:equal_estimand}
	\ES^{\mathrm{eq}}(e) \equiv \frac{1}{|\calGtrg(e)|}\sum_{g \in \calGtrg(e)} \CATT(g,e)~,
\end{equation}
which treats each policy adoption event---each cohort's experience of treatment---as an equally informative observation about the causal mechanism. This estimand answers the question ``what is the average treatment effect across the distinct adoption episodes observed at event-time $e$?'', weighting each cohort's evidence equally. It is appropriate when the research interest lies in the policy mechanism rather than the aggregate impact, or when one cohort is much larger than the others and the researcher does not want a single large cohort to dominate the estimate. Equal weights do not target a population-weighted parameter and should therefore be reported alongside cohort-size weights rather than as a primary specification.

\subsubsection*{Precision weights.} Set $\omega_g^{\mathrm{prec}}(e) = ({n_g/\sigma_g^2(e)})/({\sum_{g'\in\calGtrg(e)} n_{g'}/\sigma_{g'}^2(e)})$, where $\widehat{\sigma}_g^2(e)$ is a consistent estimator of the asymptotic variance of $\widehat{\ATT}_{\Stack_g}(g, g+e)$. This inverse-variance weighting minimizes the asymptotic variance of $\ESstack(e)$ when the within-stack estimators are asymptotically independent (an approximation when stacks share comparison units). The population target of precision weights is
\begin{equation}\label{eq:precision_estimand}
	\ES^{\mathrm{prec}}(e) \equiv \frac{\sum_{g \in \calGtrg(e)} \sigma_g^{-2}(e)  \CATT(g,e)}{\sum_{g \in \calGtrg(e)} \sigma_g^{-2}(e)}~,
\end{equation}
where $\sigma_g^2(e)$ is the asymptotic variance of the within-stack estimator. Unlike cohort-size and equal weights, the precision-weighted estimand depends on the data generating process through the variance terms $\sigma_g^2(e)$ and therefore does not have a fixed causal interpretation that is invariant to the sampling design. In particular, two different samples from the same population will in general target different weighted averages of $\CATT(g,e)$. Precision weights are the statistically optimal choice for minimizing estimation error, but the resulting estimand is harder to interpret substantively. Researchers using precision weights should report the realized weight shares $\widehat{\omega}_g^{\mathrm{prec}}(e)$ alongside point estimates.

\subsubsection*{General welfare weights.} The three schemes above are special cases of a general welfare-weighted aggregation. Suppose the researcher assigns welfare weights $v_g > 0$ reflecting the relative importance of cohort $g$'s treatment effect. The welfare-weighted event-study parameter is
\begin{equation}\label{eq:welfare_estimand}
	\ES^{v}(e) = \frac{\sum_{g \in \calGtrg(e)} v_g \CATT(g,e)}{\sum_{g \in \calGtrg(e)} v_g}~.
\end{equation}
The stacked DDD framework accommodates any choice of $v_g$ through the aggregation step, a flexibility that enhances interpretability in the case of stacked DDD. Making the welfare weights explicit is a key advantage of the stacked approach, as it separates the statistical problem (consistent estimation of $\CATT(g,e)$) from the aggregation problem (which cohorts matter more for the policy question at hand). This transparency echoes the recommendation of \citet{de_chaisemartin_two-way_2020} and \citet{de_chaisemartin_difference--differences_2024} that researchers should always report the weights underlying their treatment effect aggregation.

\subsection{Connection to Existing Estimators}\label{sub:existing_ddd}

\subsubsection*{The interaction-weighted DDD estimator.}
The stacked DDD is the natural triple-differences analogue of \citeauthor{sun_estimating_2021}'s (\citeyear{sun_estimating_2021}) interaction-weighted (IW) estimator for DiD. This connection reveals both the estimand targeted by the stacked estimator and the source of contamination in conventional 3WFE event-study specifications.
The event-study parameter is
\begin{equation}\label{eq:es_catt}
	\ES_{\mathrm{stack}}(e) = \sum_{g \in \calGtrg(e)} \omega_g(e) \CATT(g, e)~,
\end{equation}
where $\calGtrg(e) = \{g \in \calGtrg : g + e \leq T\}$ is the set of cohorts observed for at least $e$ post-treatment periods.

I define the IW-DDD estimator as
\begin{equation}\label{eq:iw_ddd}
	\widehat{\CATT}_{\mathrm{IW}}(e) = \sum_{g \in \calGtrg(e)} \widehat{\omega}_g(e) \widehat{\CATT}(g, e)~,
\end{equation}
where $\widehat{\CATT}(g,e)$ is a consistent estimator of $\CATT(g,e)$ and $\widehat{\omega}_g(e)$ are non-negative weights summing to one. When $\widehat{\CATT}(g,e)$ is the within-stack estimator and the weights coincide with those used in the stacked aggregation, the IW-DDD estimator is numerically identical to $\ESstack(e)$. Formally, by \Cref{prop:saturated_ols}, the within-stack saturated coefficient $\widehat\tau^{\mathrm{sat}}_{g,g+e}$ serves as the plug-in $\widehat\CATT(g,e)$ in \eqref{eq:iw_ddd}, and substituting into the IW-DDD definition gives
\begin{equation}\label{eq:iw_representation}
	\widehat\CATT_{\mathrm{IW}}(e) = \sum_{g\in\calGtrg(e)} \widehat\omega_g(e)\,\widehat\tau^{\mathrm{sat}}_{g,g+e} = \ESstack(e) ~.
\end{equation}

\subsubsection*{The average effect of switching}
Following \citet{de_chaisemartin_difference--differences_2024}, define the average effect of switching at event-time $e$ as
\begin{equation}\label{eq:avsw_iw}
	\mathrm{AS}(e) = \frac{\sum_{g \in \calGtrg(e)} \Prob(S_i = g, Q_i = 1) \CATT(g,e)}{\sum_{g \in \calGtrg(e)} \Prob(S_i = g, Q_i = 1)}~,
\end{equation}
the population-share-weighted average of cohort-specific treatment effects at exposure $e$. This coincides with the cohort-size-weighted event-study parameter $\ES^{\mathrm{cohort}}(e)$ from \eqref{eq:cohort_estimand}. Under the stacked DDD framework, $\mathrm{AS}(e)$ is estimated by the cohort-size-weighted stacked estimator.

\section{Identification}\label{sec:identification}

This section presents the identification strategy for stacked triple differences. I first establish nonparametric identification of the cohort-time average treatment effects within each stack, then show that the stacked construction mechanically eliminates the forbidden comparisons identified by \citet{strezhnev_decomposing_2023}, and finally develop testable implications of the identifying assumptions.

\subsection{Identification within Stacks}\label{subsec:id_within_stacks}

I now establish identification of $\ATT(g,t)$ within each stack. The main result shows that the triple difference of cell-mean outcome changes recovers the causal effect.

\begin{theorem}[Identification in Stacks]\label{thm:identification}
	Under \Cref{as:sampling}--\Cref{as:admissibility}, for each stack $\Stack_g$ with comparison group $g_c$, and for all post-treatment periods $t \geq g$~,
	\begin{align}
		\ATT(g,t) &= \E[\DeltaY_{i,t} \mid S_i = g, Q_i = 1] - \E[\DeltaY_{i,t} \mid S_i = g, Q_i = 0] \notag \\
		&\quad - \left(\E[\DeltaY_{i,t} \mid S_i = g_c, Q_i = 1] - \E[\DeltaY_{i,t} \mid S_i = g_c, Q_i = 0]\right)~. \label{eq:ra_id}
	\end{align}
\end{theorem}

\begin{remark}[Pairwise versus global parallel changes-in-trends]\label{rmk:pairwise_pt}
	A key advantage of the stacked DDD framework is that \Cref{as:dddpct} need hold only pairwise between each treatment cohort $g$ and its specific comparison group $g_c$, rather than globally across all cohort pairs simultaneously. This is a weaker requirement than the typical ``global'' parallel changes-in-trends assumption. In the stacked DDD framework, one explicitly chooses which comparison group $g_c$ to pair with each treated cohort $g$, and the identifying assumption is tailored to each pair individually. If parallel changes-in-trends fails for one particular comparison group---say, because a concurrent policy shock affects the eligibility-specific trend in that group---only that stack's estimates are affected; the remaining stacks retain their validity.
\end{remark}

\begin{remark}[Role of each cell in the stack]\label{rmk:cell_roles}
	Each of the four cells in the stack $\Stack_g$ plays a distinct identifying role. The {treated-eligible} cell $(S_i = g, Q_i = 1)$ provides observed treated outcomes---these units are the target population for $\ATT(g,t)$. The {within-group ineligible} cell $(S_i = g, Q_i = 0)$ measures the group-specific trend for cohort $g$ among ineligible units, enabling the researcher to separate the treatment effect from group-specific confounds. The {comparison-eligible} cell $(S_i = g_c,  Q_i = 1)$ measures the eligibility-specific trend absent treatment. Finally, the {comparison-ineligible} cell $(S_i = g_c,  Q_i = 0)$ anchors the comparison group's baseline, enabling the third difference that isolates the causal effect from eligibility-specific trends.
\end{remark}

\begin{remark}[Identification without homogeneity]\label{rmk:id_heterogeneity}
	\Cref{thm:identification} identifies $\ATT(g,t)$ under arbitrary within-cohort treatment effect heterogeneity. The individual-level effects $Y_{i,t}(g) - Y_{i,t}(\infty)$ may vary freely across units sharing the same $(g, t)$, and the identification formula \eqref{eq:ra_id} recovers the average of these heterogeneous effects over the treated-eligible population without imposing any restriction on their joint distribution. In particular, no constant-effects assumption is needed.
\end{remark}

Several features of \Cref{thm:identification} merit emphasis. First, identification uses {only} data within the stack $\Stack_g$---that is, units with $S_i \in \{g, g_c\}$. No information from other treatment cohorts is needed, and no cross-stack restrictions are imposed.
Second, the identification formula in \eqref{eq:ra_id} is a {regression adjustment} representation. It recovers the counterfactual by modeling conditional outcome changes in each of the three comparison cells---$(g, 0)$, $(g_c, 1)$, and $(g_c, 0)$---and combining them to form the triple difference.
Third, the use of the long difference $\DeltaY_{i,t}$ avoids the compositional issues that arise when the comparison group's treatment status changes over time.
Finally, the identification formula \eqref{eq:ra_id} is the unconditional triple difference of cell means. When pre-treatment covariates are available, one can leverage covariate-adjusted identification results developed in Appendix \ref{app:covariates}.

\section{Asymptotic Theory: Estimation and Inference}\label{sec:asymptotics}

This section establishes the asymptotic properties of the stacked DDD estimator. The key challenge is that comparison units may be shared across stacks (e.g., never-treated units appear in every stack), inducing cross-stack dependence that must be accounted for in variance estimation and inference.

\subsection{Estimation}

I begin with within-stack estimation of $\ATT(g,t)$, then describe aggregation across stacks to form the event-study parameter $\ESstack(e)$, discuss the treatment of multiple comparison groups, and outline the practical implementation algorithm.

Throughout, fix a treatment cohort $g \in \calGtrg$ with associated comparison group $g_c$ and stack $\Stack_g$. Let $n_g$ denote the total number of units in the stack. Write $n_{s,q} = \sum_{i \in \Stack_g} \one\{S_i = s,  Q_i = q\}$ for the number of units in cell $(s,q)$, and $n_{g,1}$ for the number of treated-eligible units.

\subsubsection{Within-Stack Estimation}\label{sec:within-stack}

The identification formula (\Cref{thm:identification}) expresses $\ATT(g,t)$ as a triple difference of cell-mean outcome changes. The natural sample analogue is the sample triple difference
\begin{equation}\label{eq:att-sample}
	\widehat{\ATT}_{\Stack_g, g_c}(g,t) = \big(\overline{\DeltaY}_{g,1,t} - \overline{\DeltaY}_{g,0,t}\big) - \big(\overline{\DeltaY}_{g_c,1,t} - \overline{\DeltaY}_{g_c,0,t}\big)~,
\end{equation}
where $\overline{\DeltaY}_{s,q,t} = n_{s,q}^{-1}\sum_{i\colon S_i = s,  Q_i = q} \DeltaY_{i,t}$ is the cell-specific sample mean of long differences $\DeltaY_{i,t} = Y_{i,t} - Y_{i,g-1}$.

As shown in Section \ref{sub:saturated}, this estimator is numerically identical to the OLS coefficient on the interaction term in the fully saturated regression \eqref{eq:saturated} within the stack. It requires no functional form assumptions and no nuisance function estimation---it is a purely design-based estimator that exploits only the four-cell structure within each stack. For the following results, I maintain the following assumption.

\begin{assumption}[Non-vanishing Stack Shares]\label{as:shares}
	Let $n$ denote the total number of unique units in the global sample. Let $\Stack_g$ denote the set of indices $i$ belonging to stack $g$, and let $n_g = \sum_{1 \leq i \leq n} \one\{i \in \Stack_g\}$ be the total number of units in stack $g$. Let $n_{s,q}$ be the number of units in cohort $s \in \{g, g_c\}$ with eligibility $q \in \{0,1\}$. As $n \to \infty$,
	\begin{enumerate}
		\item[{(i)}] \emph{Cohort shares.} For each $g \in \calG = \calGtrg\cup\{\infty\}$, $n_g/n \pto \lambda_g \equiv \Prob(i \in \Stack_g) > 0$.
		\item[{(ii)}] \emph{Within-cohort cell shares.} For each $g \in \calG$ and each cell $(s,q) \in \{0,1\}^2$ with $\Prob(S_i = s, Q_i = q \mid G_i = g)>0$, the empirical share satisfies $n^{(g)}_{s,q}/n_g \pto \pi^{(g)}_{s,q} \equiv \Prob(S_i = s, Q_i = q \mid G_i = g) > 0$.
		\item[{(iii)}] \emph{Stack-level shares.} For each treated cohort $g \in \calGtrg$ with admissible comparison $g_c \in \mathcal{C}(g)$ and each cell $(s,q)$ in stack $g$, the within-stack share $n_{s,q}/n_g \pto \pi_{s,q}^{(g,g_c)} > 0$, where $n_g = \sum_{s,q} n_{s,q}$ aggregates over the four cells of stack $g$.
	\end{enumerate}
\end{assumption}

\Cref{as:shares} requires that each cohort, each within-cohort cell, and each stack-level cell retain non-vanishing population shares. The cohort-superscripted notation $\pi^{(g)}_{s,q}$ makes explicit that the cell share is a \emph{within-cohort} conditional probability, not a marginal—different cohorts may have different $(S,Q)$-cell distributions, and the stack-level share $\pi^{(g,g_c)}_{s,q}$ derived in (iii) inherits this cohort-indexed structure. The primary sampling frame is i.i.d.\ at the unit level (\Cref{as:sampling}), hence cell shares are limits of sample proportions. Under the i.i.d. sampling frame, \Cref{as:shares}(i)--(iii) follow from the weak law of large numbers.

\begin{prop}[Within-Stack Consistency]\label{prop:within_stack_consistency}
	Under \Cref{as:sampling}--\Cref{as:admissibility} and \Cref{as:shares}, the sample triple difference is consistent for the cohort-time average treatment effect on the treated:
	\begin{equation}\label{eq:within_consistency}
		\widehat{\ATT}_{\Stack_g, g_c}(g,t) \pto \ATT(g,t) \quad \text{as } n \to \infty~.
	\end{equation}
\end{prop}

\Cref{prop:within_stack_consistency} extends to estimated weights satisfying $\widehat{\omega}_g(e) \pto \omega_g(e)$, which holds in particular for cohort-size weights $\widehat{\omega}_g(e) = n_{g,1}/\sum_{g'\in\calGtrg(e)} n_{g',1}$ under \Cref{as:shares}.

\subsubsection{Aggregation Across Stacks}\label{sec:aggregation}

The within-stack estimators $\widehat{\ATT}_{\Stack_g, g_c(g)}(g,t)$ recover cohort-time treatment effects separately for each $g$. To construct event-study parameters, I aggregate across cohorts at each relative event-time $e$. Define the stacked event-study estimator
\begin{equation}\label{eq:es-stack}
	\ESstack(e) = \sum_{\substack{g \in \calGtrg :\\ g + e  \in  \{g-L, \ldots, g+K\}}} \widehat{\omega}_g(e)  \cdot  \widehat{\ATT}_{\Stack_g, g_c(g)}(g, g+e)~,
\end{equation}
where $L$ and $K$ are the pre- and post-treatment window lengths and $\widehat{\omega}_g(e) \geq 0$ with $\sum_g \widehat{\omega}_g(e) = 1$ for each $e$. The summation is over all cohorts $g$ for which event-time $e$ falls within the stack window. Following my earlier discussion, I consider two choices of weights.

\medskip
\noindent\textbf{1.\ Cohort-size weights.}
Set
\begin{equation}\label{eq:weights-cohort}
	\widehat{\omega}_g(e) = \frac{n_{g,1}}{\displaystyle\sum_{g' \in \calGtrg:  g'+e  \in  \{g'-L, \ldots, g'+K\}} n_{g',1}},
\end{equation}
so that each cohort's contribution is proportional to the number of treated-eligible units it contains. This weighting targets the population event-study parameter $\ES(e) = \E[\ATT(G_i, G_i + e) \mid G_i + e \in \{G_i - L, \ldots, G_i + K\}]$, which weights cohorts by their population shares. It is the natural choice when the goal is to estimate an average effect across the treated population.

\medskip
\noindent\textbf{2.\ Equal weights.}
Set
\begin{equation}\label{eq:weights-equal}
	\widehat{\omega}_g(e) = \frac{1}{\big|\{g' \in \calGtrg : g'+e \in \{g'-L, \ldots, g'+K\}\}\big|},
\end{equation}
which assigns each cohort equal influence regardless of its sample size or estimation precision. Equal weighting is the most transparent aggregation scheme and places equal value on detecting treatment effects in small and large cohorts alike. However, it has no clear population interpretation and may be inefficient when cohort sizes vary substantially.

\begin{remark}
	The choice among these two weights involves a familiar trade-off. Cohort-size weights have the most transparent causal interpretation, targeting the ``per unit'' (of treatment) ATT. Equal weights are the most transparent and are robust to concerns about one or two large cohorts dominating the aggregation, but they may sacrifice both efficiency and interpretability when cohort sizes are heterogeneous. In practice, I recommend reporting results under cohort-size weights as the primary specification, with equal weights as a robustness check. I note that precision weights are statistically efficient but may assign opaque weights that change across event-times, making causal interpretation difficult. 
\end{remark}

\subsection{Inference}

The stacked event-study regression \eqref{eq:stacked_es_reg} delivers point estimates $\widehat{\tau}_e$ that are convex combinations of the within-stack triple differences (\Cref{prop:stacked_es_estimands}). However, standard OLS standard errors from this regression---the heteroskedasticity-robust  standard errors---do not correctly account for the cross-stack dependence induced by shared comparison units. It turns out the the cluster-robust variance estimators (CRVE) clustered at the original unit level {do} correctly recover the asymptotic variance.
In the results that follow, all limits are taken as $n \to \infty$ with the number of periods $T$ and the number of cohorts $|\calGtrg|$ held fixed.

\begin{assumption}[Finite Second Moments]\label{as:moments}
	For each stack $\Stack_g$ with comparison group $g_c$, each cell $(s,q)$, and each $t$ in the stack window, $\E[Y_{i,t}^2 \mid S_i = s, Q_i = q] < \infty$.
\end{assumption}

\Cref{as:moments} requires that outcomes have finite variance in each cell.

\begin{theorem}[Within-Stack Asymptotic Normality]\label{thm:within_stack_clt}
	Under \Cref{as:sampling}--\Cref{as:admissibility}, \Cref{as:shares}, and \Cref{as:moments}, for each stack $\Stack_g$ with comparison group $g_c$ and for $t \geq g$, {writing $\DeltaY_{i, t} \equiv Y_{i, t} - Y_{i, g - 1}$ for the long difference from the stack's baseline period $g - 1$,}
	\begin{equation}\label{eq:within_clt}
		\sqrt{n_g}\Big(\widehat{\ATT}_{\Stack_g, g_c}(g,t) - \ATT(g,t)\Big) \dto \mathcal{N}\big(0, \Sigma_{g,t,g_c}\big)~,
	\end{equation}
	where the asymptotic variance is
	\begin{equation}\label{eq:within_var}
		\Sigma_{g,t,g_c} = \sum_{(s,q)} \frac{1}{\pi_{s,q}} \Var\big(\DeltaY_{i,t} \mid S_i = s, Q_i = q\big)~,
	\end{equation}
	{and $\pi_{s, q} \equiv \lim_{n \to \infty} n_{s, q} / n_g$ is the \emph{within-stack} conditional share of cell $(s, q)$.}
	The influence function for unit $i \in \Stack_g$ is
	\begin{equation}\label{eq:if_decomp}
		\psi_{\Stack_g}(W_i; g, t, g_c) = \sum_{(s,q)} \frac{\one\{S_i = s,  Q_i = q\}}{\pi_{s,q}} c_{s,q} \big(\DeltaY_{i,t} - \E[\DeltaY_{i,t} \mid S_i = s, Q_i = q]\big)~,
	\end{equation}
	where $c_{g,1} = +1$, $c_{g,0} = -1$, $c_{g_c,1} = -1$, $c_{g_c,0} = +1$ are the signs of the triple difference.
\end{theorem}

The influence function \eqref{eq:if_decomp} has a transparent structure. Each unit contributes through its cell-specific deviation from the cell mean, weighted by the inverse of its cell's population share and signed according to the triple-difference pattern. The variance \eqref{eq:within_var} is the sum of cell-specific variance contributions, each scaled by the inverse cell share---larger cells contribute less variance per observation, as expected.

I now turn to the asymptotic distribution of the aggregated event-study estimator $\ESstack(e)$, which combines within-stack treatment effect estimates across cohorts using weights $\omega_g(e)$.

\begin{theorem}[Stacked Event-Study CLT]\label{thm:agg_clt}
	Under the conditions of \Cref{thm:within_stack_clt}, for deterministic weights $\omega_g(e) \geq 0$ satisfying $\sum_{g \in \calGtrg} \omega_g(e) = 1$ and for event-time $e \in \{0,\ldots,K\}$:
	\begin{equation}\label{eq:agg_clt}
		\sqrt{n}\Big(\ESstack(e) - \ES(e)\Big) \dto \mathcal{N}\big(0, \Vstack(e)\big),
	\end{equation}
	where $n = \sum_{g \in \calGtrg} n_g$ and the form of $\Vstack(e)$ depends on whether stacks share control units.
\end{theorem}

The key subtlety in the aggregated result is that when stacks share comparison units---the most common case in practice, where every stack uses the never-treated group $S_i = \infty$ as its comparison---the within-stack estimators are \emph{not} independent. The shared comparison units induce cross-stack correlation that must be accounted for in the variance formula.
To see why, consider a unit $i$ with $S_i = \infty$ that belongs to stacks $\Stack_g$ and $\Stack_{g'}$ for two distinct cohorts $g \neq g'$. The influence function contribution of this unit to the two within-stack estimators is generally non-zero for both, creating dependence between $\widehat{\ATT}_{\Stack_g, g_c}(g, g+e)$ and $\widehat{\ATT}_{\Stack_{g'}, g_c}(g', g'+e)$. The following result characterizes the resulting variance structure.

\begin{prop}[Variance with Shared Controls]\label{prop:shared_var}
	When stacks share comparison units (\eg, $g_c = \infty$ for all stacks), the asymptotic variance of the aggregated estimator is:
	\begin{equation}\label{eq:shared_var}
		\Vstack(e) = \sum_{g \in \calGtrg} \sum_{g' \in \calGtrg} \omega_g(e)  \omega_{g'}(e)  \mathrm{Cov}\big(\psi_{\Stack_g}(W_i; g, g+e, g_c), \psi_{\Stack_{g'}}(W_i; g', g'+e, g_c)\big),
	\end{equation}
	where the covariance is non-zero when stacks $\Stack_g$ and $\Stack_{g'}$ share comparison units. When stacks use distinct comparison groups ($g_c(g) \neq g_c(g')$ with no shared units), the cross-terms vanish and:
	\begin{equation}\label{eq:indep_var}
		\Vstack(e) = \sum_{g \in \calGtrg} \omega_g(e)^2  \frac{\Sigma_{g, g+e, g_c(g)}}{n_g} .
	\end{equation}
\end{prop}

\begin{remark}[Magnitude of cross-stack correlation]\label{rmk:cross_corr}
	The cross-stack covariance in \eqref{eq:shared_var} arises exclusively from shared comparison units. In applications where the never-treated pool is large relative to the treated cohorts, these units receive small inverse-probability weights, and the cross-stack correlation is modest. Conversely, when the never-treated pool is small and each comparison unit receives substantial weight, the cross-stack dependence is non-negligible and ignoring it leads to confidence intervals that are too narrow.
\end{remark}

A natural estimator of $\Vstack(e)$ in \eqref{eq:shared_var} is the sample analogue based on estimated influence functions. For each unit $i$ and each stack $\Stack_g$, define $\widehat{\psi}_{\Stack_g}(W_i; g, g+e, g_c)$ as the influence function \eqref{eq:if_decomp} evaluated at sample cell proportions and sample cell means, with $\widehat{\psi}_{\Stack_g}(W_i) = 0$ for units $i \notin \Stack_g$.\footnote{To verify that this is the correct influence function, note that $\widehat{\tau}_{g,g+e}^{\mathrm{sat}} = n_g^{-1}\sum_{i \in \Stack_g} \widehat{\psi}_{\Stack_g}(W_i) + \widehat{\tau}_{g,g+e}^{\mathrm{sat}}$, since $\sum_{i \in \Stack_g} \widehat{\psi}_{\Stack_g}(W_i) = 0$ by the orthogonality of cell-mean residuals. The influence function captures the first-order effect of perturbing unit $i$'s outcome on the triple difference---adding a unit to cell $(g,1)$ with outcome $\DeltaY_{i,g+e}$ shifts $\overline{\DeltaY}_{g,1,g+e}$ by approximately $(\DeltaY_{i,g+e} - \overline{\DeltaY}_{g,1,g+e})/n_{g,1}$, which changes $\widehat{\tau}_{g,g+e}^{\mathrm{sat}}$ by approximately $(\DeltaY_{i,g+e} - \overline{\DeltaY}_{g,1,g+e})/n_{g,1} = \widehat{\psi}_{\Stack_g}(W_i)/n_g$.} The plug-in variance estimator is
\begin{equation}\label{eq:var_hat}
	\widehat{V}_{\mathrm{stack}}(e) = \frac{1}{n}\sum_{i=1}^{n}\left(\sum_{g \in \calGtrg} \widehat{\omega}_g(e) \widehat{\psi}_{\Stack_g}(W_i; g, g+e, g_c(g))\right)^{\!2}.
\end{equation}
This estimator has two important features. First, by summing the weighted influence function contributions at the unit level before squaring, it automatically captures the cross-stack covariance for shared control units. When unit $i$ belongs to multiple stacks, its contributions from different stacks are summed inside the square, producing the correct cross-terms. Second, to establish consistency of this variance estimator, I require a strengthening of the moment condition in \Cref{as:moments}.

\begin{assumption}[Finite Fourth Moments]\label{as:moments4}
	For each stack $\Stack_g$ with comparison group $g_c$, each cell $(s,q)$, and each $t$ in the stack window, $\E[Y_{i,t}^4 \mid S_i = s, Q_i = q] < \infty$.
\end{assumption}

\Cref{as:moments4} strengthens the second-moment condition in \Cref{as:moments} to fourth moments. This is needed because the variance estimator \eqref{eq:var_hat} is a sample average of \emph{squared} influence functions: the law of large numbers requires that these squared terms have finite variance, which imply finite fourth moments of the outcome differences by the $c_r$ inequality.

\begin{prop}[Variance Estimator Consistency]\label{prop:var_consistency}
	Under \Cref{as:sampling}--\Cref{as:admissibility}, \Cref{as:shares}, and \Cref{as:moments4}, for deterministic weights $\omega_g(e)$ satisfying $\sum_g \omega_g(e) = 1$, and for each event-time $e \in \{0,\ldots,K\}$~
	\begin{equation}\label{eq:var_consistency}
		\widehat{V}_{\mathrm{stack}}(e) \pto \Vstack(e)~.
	\end{equation}
\end{prop}

\Cref{prop:var_consistency} extends to estimated weights satisfying $\widehat{\omega}_g(e) \pto \omega_g(e)$, which holds in particular for cohort-size weights $\widehat{\omega}_g(e) = n_{g,1}/\sum_{g'\in\calGtrg(e)} n_{g',1}$ under \Cref{as:shares}.
At a high-level, the proof involves showing that the plug-in variance estimator is a sample average of squared estimated influence functions that converges by the law of large numbers, with the estimation error in the influence functions controlled by all the regularity conditions above. As an immediate consequence, pointwise confidence intervals of the form
\begin{equation}\label{eq:pointwise_ci}
	\ESstack(e) \pm z_{\alpha/2}\sqrt{\widehat{V}_{\mathrm{stack}}(e)/n}
\end{equation}
have asymptotically correct coverage. More precisely, \Cref{prop:var_consistency} combined with \Cref{thm:agg_clt} and Slutsky's theorem gives
\begin{equation}\label{eq:pointwise_coverage}
	\frac{\ESstack(e) - \ES(e)}{\sqrt{\widehat{V}_{\mathrm{stack}}(e)/n}} \dto \mathcal{N}(0,1)~,
\end{equation}
so the interval \eqref{eq:pointwise_ci} has asymptotic coverage $1-\alpha$ for each $e$.

In practice, valid inference under the fully saturated OLS regression specification applied to stacked dataset can be done via the cluster-robust standard error, where the clustering is done at the unit level. This validity is ensured under the same assumptions maintained in \Cref{prop:var_consistency}; I prove this in Appendix \ref{app:proof-CRVE-validity}.
Hence, in the case of estimating fully saturated OLS regressions on the stacked dataset, the usual cluster-robust standard errors---clustered at the level of treatment---can be used without implementing the estimator proposed in \Cref{prop:var_consistency}. It is possible to construct uniform confidence bands in this setting; I discuss this in \Cref{app:technical-results}.

\section{Empirical Illustrations}\label{sec:empirical}

I illustrate the stacked DDD framework with two applications drawn from recent empirical work that exploits staggered rollouts over a country $\times$ eligibility-group panel: \citet{hansen_national_2023}'s study of the national impacts of genetically modified crops, and \citet{shastry_vaccine_2025}'s evaluation of Gavi's vaccine program. For both illustrations, I use the never-treated groups as my clean comparison group in each stack, and present the OLS-implied weights for each stack.

\subsection{National and Global Impacts of Genetically Modified Crops}\label{sub:empirical_crops}

\cite{hansen_national_2023} (HW) study the national and global impacts of genetically modified (GM) crops. In particular, they estimate the impact of genetically modified (GM) crops on countrywide yields, harvested area, and trade ``using a triple-differences rollout design that exploits variation in the availability of GM seeds across crops, countries, and time''. In other words, there are staggered national approvals for GM cultivation across countries, a clear eligibility dimension defined at the crop level, and a multi-country panel dataset encompassing both treated and never-treated crops and countries. 
They find ``positive impacts on yields, especially in poor countries,'' and emphasize that ``without GM crops, the world would have needed 3.4 percent additional cropland to keep global agricultural output at its 2019 level.''

Mapping their setting to my notations, the \emph{treatment-enabling group} $S_i$ records the year a country $i$ first approved GM cultivation, with never-adopting countries assigned $S_i = \infty$. \emph{Eligibility} $Q_c$ indicates whether crop $c$ has a commercially viable GM variety available globally (specifically: cotton, maize, rapeseed, and soybean, so $Q_c = 1$), whereas other field crops (e.g., rice, wheat) serve as ineligible within-country controls ($Q_c = 0$). The \emph{treatment indicator} $D_{ict} = \one\{t \geq S_i\} Q_c$ takes the value one if country $i$ has approved GM cultivation by year $t$ and crop $c$ is a GM-eligible crop.
Finally, I examine four primary outcomes at the country-crop-year level: the GM adoption rate, log yield, log harvested area, and net-export share. 

I construct a stack for each GM approval cohort $g \in \calGtrg$. Within each stack, the never-treated countries ($S_i = \infty$) serve as the clean comparison group ($g_c = \infty$). 
For a given cohort $g$, the stack contains four cells: \emph{treated-eligible} (GM crops in cohort $g$ countries), \emph{within-group ineligible} (non-GM crops in cohort $g$ countries), \emph{across-group eligible} (GM crops in never-adopting countries), and \emph{across-group ineligible} (non-GM crops in never-adopting countries).

I now compare the results of the conventional 3WFE estimator (``H\&W (pooled)'', \eqref{eq:hw_original}) against a stacked OLS estimator (``Sat. OLS (stacked)'') and the triple-differences estimator of \cite{ortiz-santanna_triple_2025} (``triplediff (DR)''). For aggregating cohort-level ATTs into event-study, for my proposed method, I used cohort-size as weights, so the event-study estimates should be interpreted as ``per unit'' ATT. The event-study plots for all four outcomes are presented in Figure \ref{fig:sec11_comparison} and Figure \ref{fig:sec11_att}.

\subsubsection*{Event-Study Estimates}
Across all outcomes of interest, the stacked DDD estimator display more modest, but still statistically significant, treatment effect estimates. In particular, my results confirm most of the findings in \cite{hansen_national_2023}. Unsurprisingly, my estimates and those of ``triplediff (DR)'' are qualitatively similar. The pooled specification of HW likely inflates the treatment effect estimates by using early-adopting countries as controls for later-adopting countries in the data.

\subsubsection*{Aggregated ATT}
All estimated effects under the stacked design are smaller in magnitude compared to the results from HW. In the case of log harvested area, the stacked OLS estimator reveals a statistically insignificant effect of GM adoption on percentage of harvest. In other cases, all the estimated aggregated ATTs are half of the estimated effects of HW.
To illustrate treatment effect heterogeneity, I show cohort-level ATTs in Figure \ref{fig:sec11_att}. 
The vertical dotted lines inside the panels plot the cohort-size weighted aggregated ATTs across these stacks, which map exactly to the corresponding pooled point estimates generated by the ``Sat. OLS (stacked)'' estimator presented in Figure \ref{fig:sec11_comparison}. 

\medskip

Overall, the alternative estimation strategy I employ here mostly confirms the positive impact of the technology on outcomes of interest, with the exception of harvested area. The main advantage of the stacked design is it deals with identification transparently, allowing practitioners to design credible sub-experiments for each cohort and better argue that identification assumption holds.

\subsection{Effective Health Aid through Gavi's Vaccine Program}\label{sub:empirical_vaccine}

I next revisit \citeauthor{shastry_vaccine_2025}'s \citeyear{shastry_vaccine_2025} evaluation of Gavi, the Vaccine Alliance, which has distributed over US\$16 billion in vaccine aid to low- and middle-income countries since its founding in 1999. \citet{shastry_vaccine_2025} identify Gavi's causal effect from the differential timing of Gavi funding across countries, vaccines, and birth cohorts, comparing Gavi-funded vaccines to non-funded vaccines within the same country and birth cohort, and comparing the same vaccines across Gavi-adopting and non-adopting countries. They estimate a triple-differences specification
\begin{equation}\label{eq:st_ols}
	Y_{ivt} = \beta_1  \mathrm{Post}_{ivt} + \theta_{iv} + \theta_{it} + \theta_{vt} + \varepsilon_{ivt} ~,
\end{equation}
with country $\times$ vaccine, country $\times$ cohort, and vaccine $\times$ cohort fixed effects, clustering standard errors by country. They report that ``Gavi's support for a vaccine increased coverage rates by $2-5$ percentage points across all vaccines, on average, and by $10-20$ percentage points for newer vaccines,'' and that ``Gavi's support for a vaccine reduced child mortality from related causes by 1 child per 1,000 live births,'' implying roughly $1.5$ million lives saved at approximately US\$9,000 per life saved. \citet{shastry_vaccine_2025} acknowledge the sensitivity of staggered 3WFE estimators to treatment-effect heterogeneity and verify robustness using \citet{borusyak_revisiting_2023}'s imputation estimator.

Here is the mapping from the context of \cite{shastry_vaccine_2025} to my notations. The treatment-enabling group $S_i$ is the year country $i$ first received Gavi funding for the vaccine targeting cause $c$, with never-funded countries assigned $S_i = \infty$. Eligibility $Q_c$ is an indicator for the cause of death being the primary target of a Gavi-funded vaccine, so respiratory, diarrheal, measles, and meningitis/encephalitis causes have $Q_c = 1$ and other under-five causes have $Q_c = 0$. The treatment indicator is $D_{ict} = \one\{t \geq S_i\} Q_c$. The cross-sectional unit is the country-cause pair, and the outcome $Y_{ict}$ is the post-neonatal mortality rate (deaths per 1,000 live births) from cause $c$ in country $i$ in year $t$.

For each Gavi-funding cohort $g \in \calGtrg$ and each Gavi-linked cause $c$, I construct a stack with never-funded countries ($S_i = \infty$) as the clean comparison group. I set $L = K = 5$ to match the 11-year window used in \citet{shastry_vaccine_2025}. I analyze two samples in \citet{shastry_vaccine_2025}'s Figure 5: the full GHO sample (``All Countries'') and the vital-registry subset restricted to years from 2005 onwards (``Post-2005 Vital Registry''), which avoids imputed mortality data at the cost of a substantially smaller set of countries.

\subsubsection*{Event-study estimates}
Figures \ref{fig:vaccine_es_all} and \ref{fig:vaccine_es_vr} report event-study coefficients from three estimators (same as those considered earlier in the HW application) applied to the data.
In both samples, all three estimators display essentially flat pre-trends, consistent with \citet{shastry_vaccine_2025}'s identifying assumption that coverage of Gavi-funded vaccines would have trended similarly to coverage of non-funded vaccines absent Gavi. Post-introduction, all three estimators show a downward trend in mortality rate.

However, the three estimators disagree in magnitude. In the full sample (Figure \ref{fig:vaccine_es_all}), the pooled OLS post-treatment ATT is $-0.945$ deaths per 1,000 live births, the saturated stacked OLS estimator gives $-0.633$, and triplediff produces $-0.248$. All aggregated ATTs estimated are statistically significant. In the cleaner vital-registry subset (Figure \ref{fig:vaccine_es_vr}), the three estimates are to $-0.362$, $-0.248$, and $-0.257$ per 1,000 live births for the original specification, saturated stacked OLS estimator, and triplediff, respectively. The large gap between pooled OLS and the two stacked estimators in the full sample liekly reflects the forbidden-comparison problem in \eqref{eq:st_ols}: later-funded country-cause pairs are implicitly compared to already-funded ones whose mortality rates continue to decline in the post-treatment window, and the resulting negative weights inflate the pooled OLS point estimate. In the smaller vital-registry subset, where the treatment timing is more concentrated and the effective control pool is narrower, the three estimators roughly agree at approximately a quarter of a death per 1,000 live births.
Figures \ref{fig:vaccine_decomp_all} and \ref{fig:vaccine_decomp_vr} decompose the saturated stacked OLS estimate into its underlying cause-by-cohort stacks. 
Applying the stacked DDD methodology to Gavi's vaccine rollout reinforces \citet{shastry_vaccine_2025}'s central conclusion that Gavi-funded vaccines meaningfully reduced child mortality. The stacked decomposition additionally shows transparently which country-vaccine-cohort stack drive the aggregate treatment effect estimate.

\begin{landscape}
	\begin{figure}[ht]
		\centering
		\includegraphics[width=1.4\textheight]{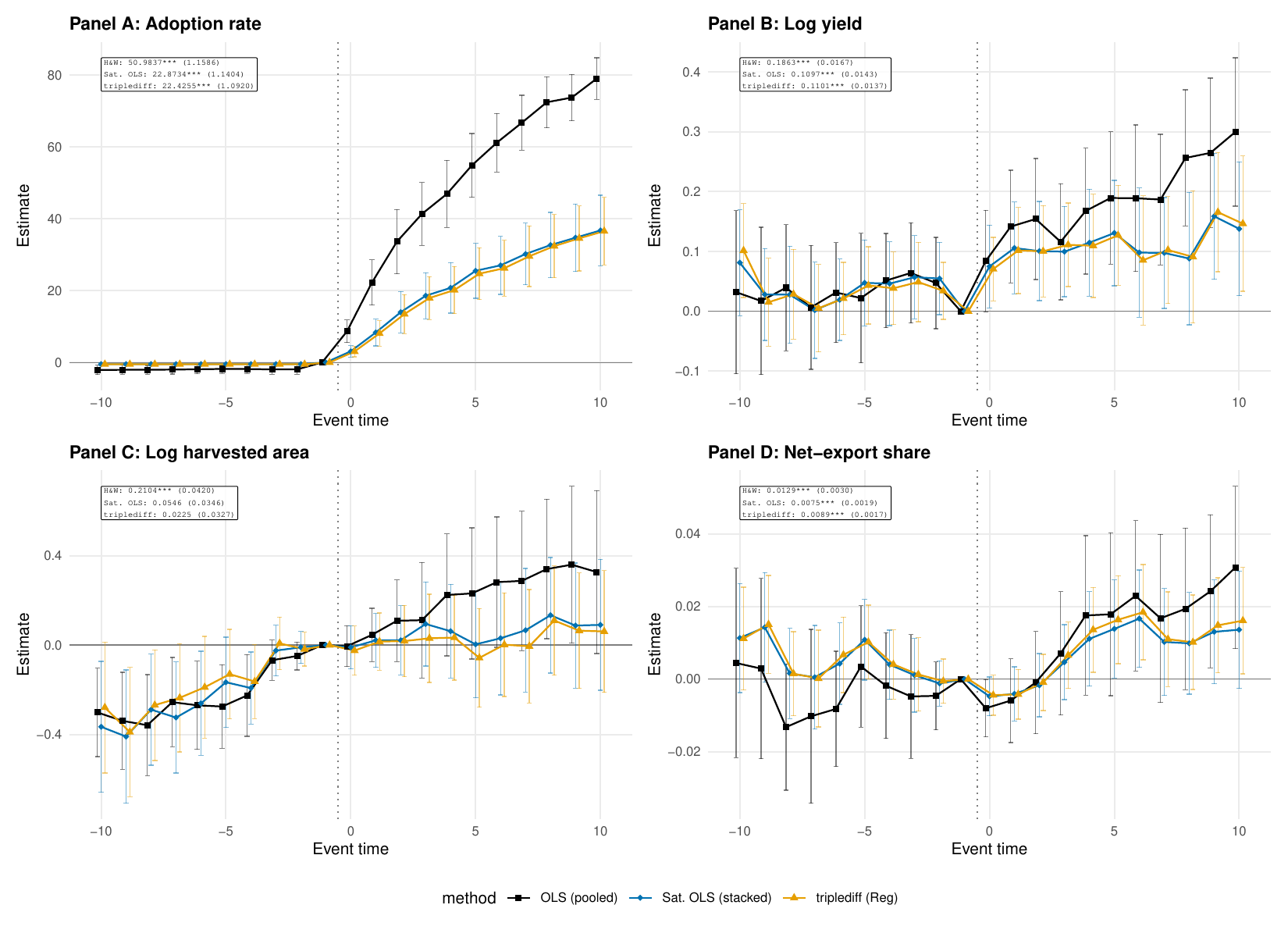}
		\caption{Event-Study Estimates by Outcome in \cite{hansen_national_2023}}
		\label{fig:sec11_comparison}
	\end{figure}
	
	\begin{figure}[ht]
		\centering
		\includegraphics[trim={0 0 0 0.7cm},clip,width=1.25\textheight]{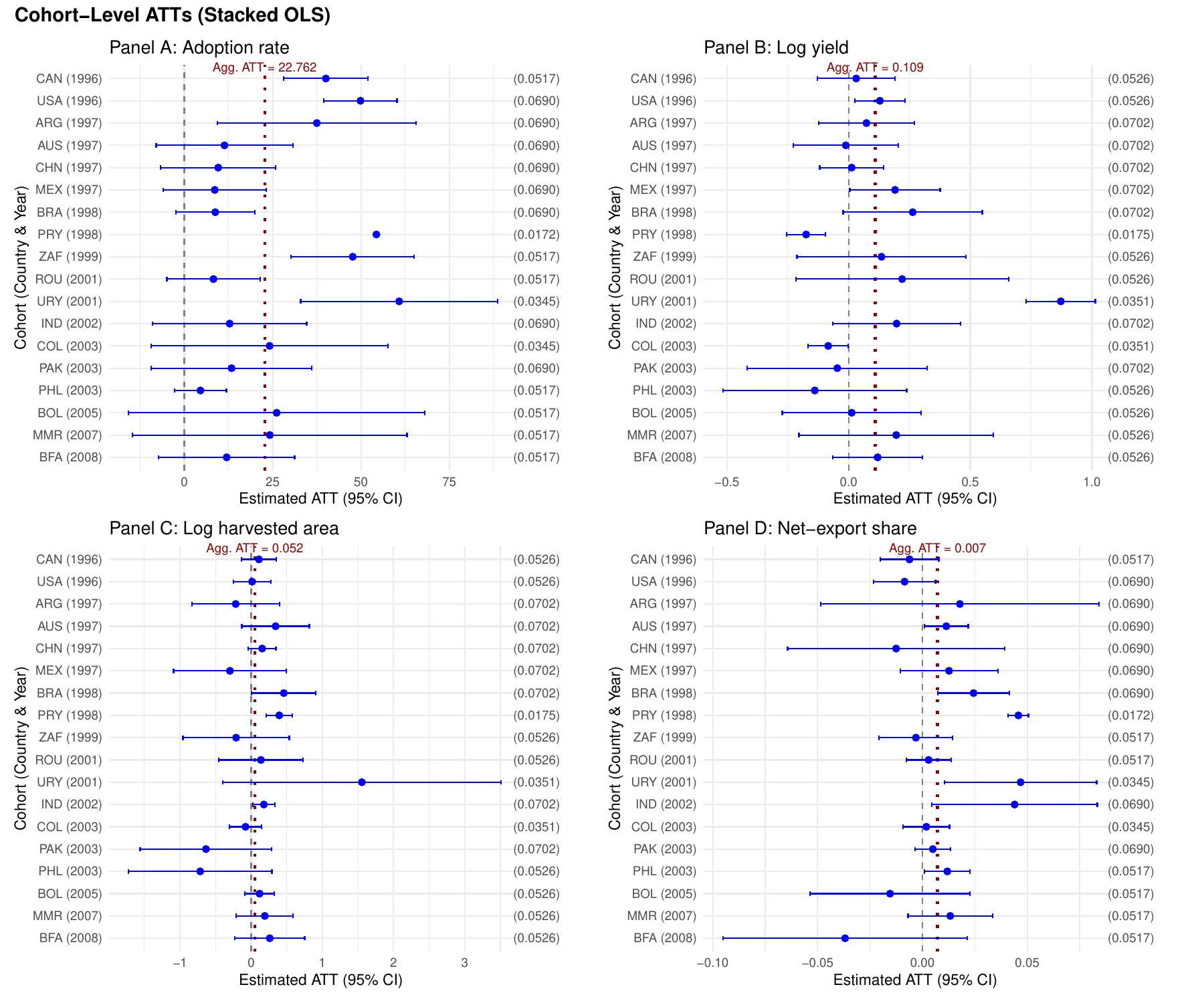}
		\caption{Cohort-Level ATT by Outcome in \cite{hansen_national_2023}}
		\label{fig:sec11_att}
	\end{figure}
	
	\begin{figure}[htbp]
		\centering
		\includegraphics[trim={0 0 0 0.7cm},clip,width=1.3\textheight]{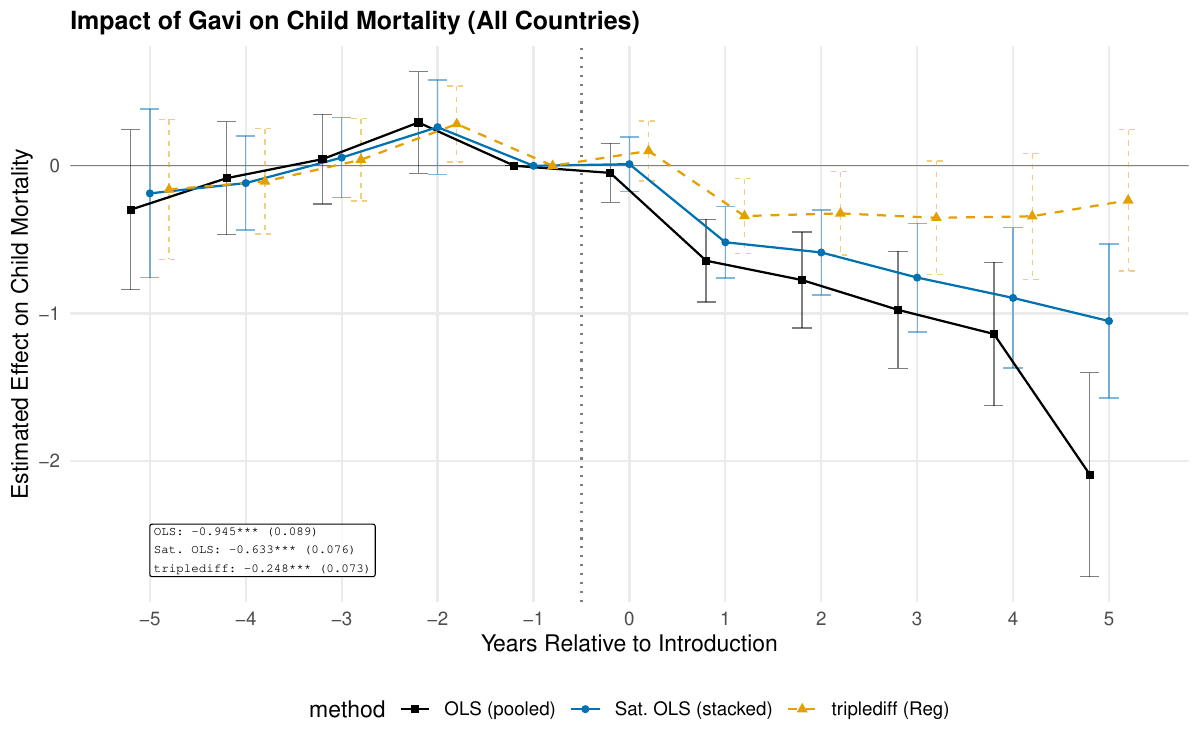}
		\caption{Event-Study Estimates: Impact of Gavi on Child Mortality, All Countries \citep{shastry_vaccine_2025}}
		\label{fig:vaccine_es_all}
	\end{figure}
	
	\begin{figure}[htbp]
		\centering
		\includegraphics[trim={0 0 0 0.7cm},clip,width=1.3\textheight]{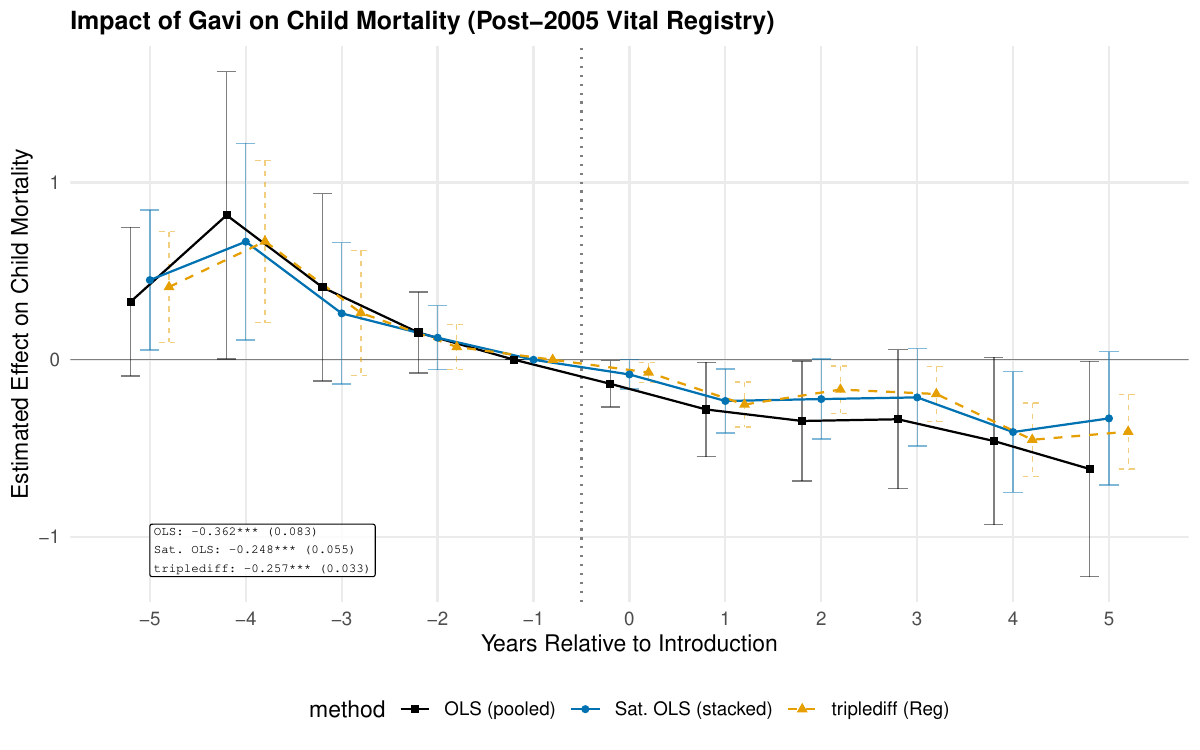}
		\caption{Event-Study Estimates: Impact of Gavi on Child Mortality, Post-2005 Vital Registry Subset \citep{shastry_vaccine_2025}}
		\label{fig:vaccine_es_vr}
	\end{figure}

	\begin{figure}[htbp]
		\centering
		\includegraphics[width=1.3\textheight]{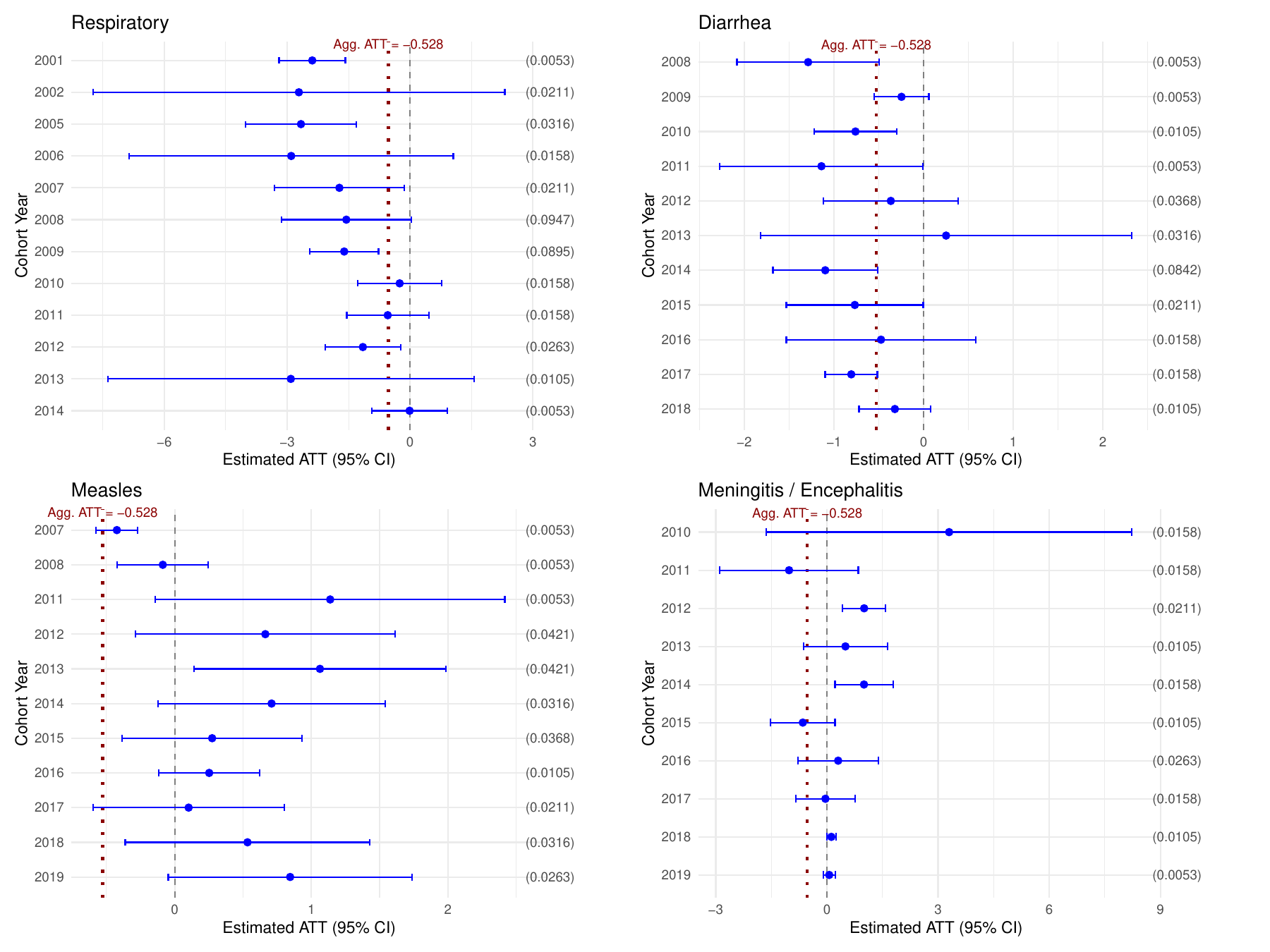}
		\caption{Stack-Level ATT Decomposition: All Countries \citep{shastry_vaccine_2025}}
		\label{fig:vaccine_decomp_all}
	\end{figure}

	\begin{figure}[htbp]
		\centering
		\includegraphics[trim={0 0 0 0.7cm},clip,width=1.25\textheight]{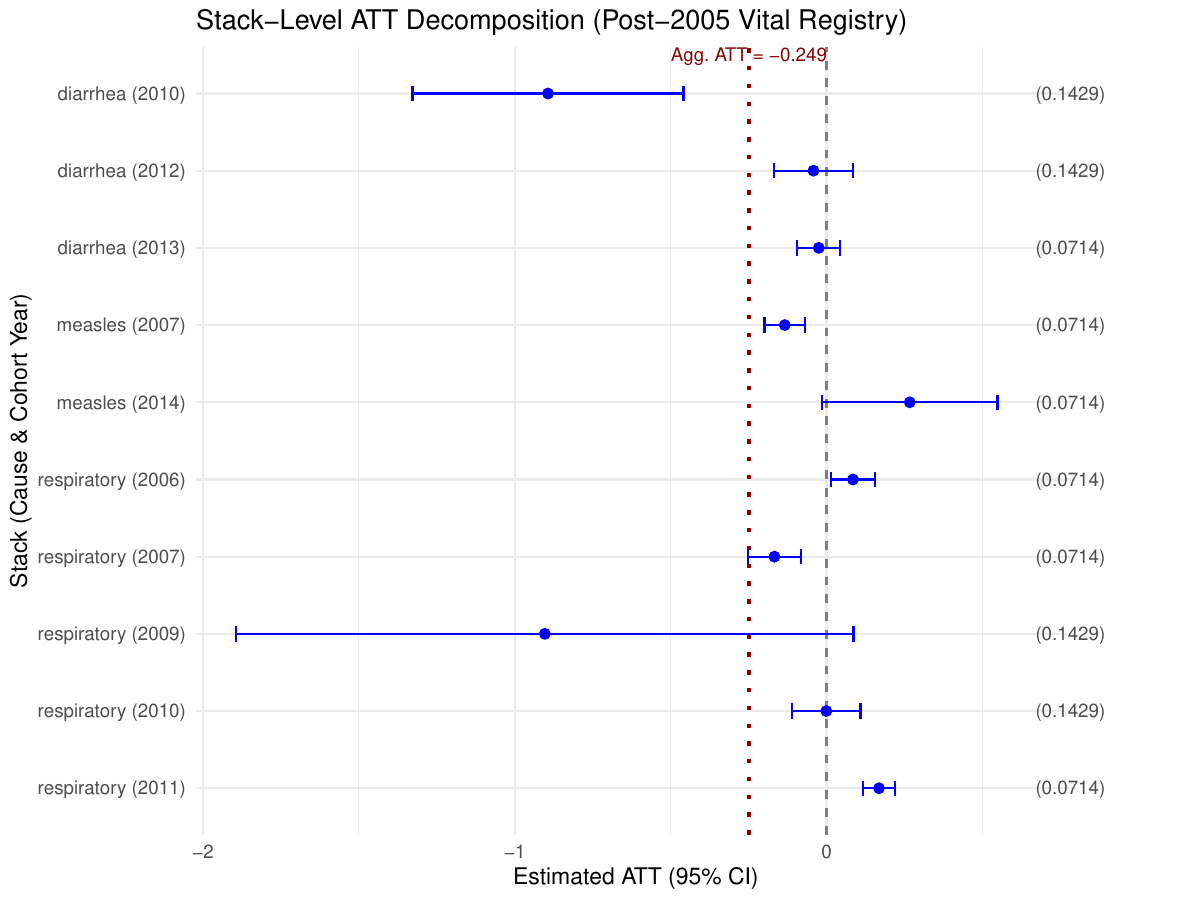}
		\caption{Stack-Level ATT Decomposition: Post-2005 Vital Registry Subset \citep{shastry_vaccine_2025}}
		\label{fig:vaccine_decomp_vr}
	\end{figure}
	
\end{landscape}

\clearpage

\bibliography{stacked}

@article{shastry_vaccine_2025,
Author = {Shastry, Gauri Kartini and Tortorice, Daniel L.},
Title = {Effective Health Aid: Evidence from Gavi's Vaccine Program},
Journal = {American Economic Journal: Economic Policy},
Volume = {17},
Number = {1},
Year = {2025},
Month = {February},
Pages = {540–74},
DOI = {10.1257/pol.20230036},
URL = {https://www.aeaweb.org/articles?id=10.1257/pol.20230036}}

@article{courtemanche_early_2017,
	title = {Early {Impacts} of the {Affordable} {Care} {Act} on {Health} {Insurance} {Coverage} in {Medicaid} {Expansion} and {Non}-{Expansion} {States}},
	author = {Courtemanche, Charles and Marton, James and Ukert, Benjamin and Yelowitz, Aaron and Zapata, Daniela},
	journal = {Journal of Policy Analysis and Management},
	volume = {36},
	number = {1},
	pages = {178--210},
	year = {2017},
	doi = {10.1002/pam.21961},
}

@article{kaestner_effects_2017,
	title = {Effects of {ACA} {Medicaid} {Expansions} on {Health} {Insurance} {Coverage} and {Labor} {Supply}},
	author = {Kaestner, Robert and Garrett, Bowen and Chen, Jiajia and Gangopadhyaya, Anuj and Fleming, Caitlyn},
	journal = {Journal of Policy Analysis and Management},
	volume = {36},
	number = {3},
	pages = {608--642},
	year = {2017},
	doi = {10.1002/pam.21993},
}

@article{greenstone_impacts_2002,
	title = {The {Impacts} of {Environmental} {Regulations} on {Industrial} {Activity}: {Evidence} from the 1970 and 1977 {Clean} {Air} {Act} {Amendments} and the {Census} of {Manufactures}},
	author = {Greenstone, Michael},
	journal = {Journal of Political Economy},
	volume = {110},
	number = {6},
	pages = {1175--1219},
	year = {2002},
	doi = {10.1086/342808},
}

@article{walker_transitioning_2013,
	title = {The {Transitional} {Costs} of {Sectoral} {Reallocation}: {Evidence} from the {Clean} {Air} {Act} and the {Workforce}},
	author = {Walker, W. Reed},
	journal = {Quarterly Journal of Economics},
	volume = {128},
	number = {4},
	pages = {1787--1835},
	year = {2013},
	doi = {10.1093/qje/qjt022},
}

@article{sant2020doubly,
	title = {Doubly Robust Difference-in-Differences Estimators},
	author = {Sant'Anna, Pedro H. C. and Zhao, Jun},
	journal = {Journal of Econometrics},
	volume = {219},
	number = {1},
	pages = {101--122},
	year = {2020},
	publisher = {Elsevier}
}

@techreport{ortiz-santanna_triple_2025,
	title = {Better {Understanding} {Triple} {Differences} {Estimators}},
	author = {Ortiz-Villavicencio, Marcelo and Sant'Anna, Pedro H. C.},
	year = {2025},
	institution = {Emory University},
	type = {Working Paper},
	url = {https://arxiv.org/abs/2505.09942},
}

@techreport{leventer_triple_2025,
	title = {Triple {Differences}: {Identification}, {Estimation}, and {Inference}},
	author = {Leventer, Avi},
	year = {2025},
	institution = {Hebrew University of Jerusalem},
	type = {Working Paper}
}

@techreport{caron_triple_2025,
	title = {Triple {Differences} with {Staggered} {Adoption} and {Heterogeneous} {Treatment} {Effects}},
	author = {Caron, Laura},
	year = {2025},
	type = {Working Paper}
}

@article{robins_new_1986,
	title = {A {New} {Approach} to {Causal} {Inference} in {Mortality} {Studies} with a {Sustained} {Exposure} {Period}---{Application} to {Control} of the {Healthy} {Worker} {Survivor} {Effect}},
	author = {Robins, James},
	journal = {Mathematical Modelling},
	volume = {7},
	pages = {1393--1512},
	year = {1986}
}

@article{butters_how_2022,
	title = {How {Do} {National} {Firms} {Respond} to {Local} {Cost} {Shocks}?},
	volume = {112},
	issn = {0002-8282},
	url = {https://pubs.aeaweb.org/doi/10.1257/aer.20201524},
	doi = {10.1257/aer.20201524},
	language = {en},
	number = {5},
	urldate = {2025-03-21},
	journal = {American Economic Review},
	author = {Butters, R. Andrew and Sacks, Daniel W. and Seo, Boyoung},
	month = may,
	year = {2022},
	pages = {1737--1772},
}

@article{deshpande_who_2019,
	title = {Who {Is} {Screened} {Out}? {Application} {Costs} and the {Targeting} of {Disability} {Programs}},
	volume = {11},
	issn = {1945-7731, 1945-774X},
	shorttitle = {Who {Is} {Screened} {Out}?},
	url = {https://pubs.aeaweb.org/doi/10.1257/pol.20180076},
	doi = {10.1257/pol.20180076},
	language = {en},
	number = {4},
	urldate = {2025-03-17},
	journal = {American Economic Journal: Economic Policy},
	author = {Deshpande, Manasi and Li, Yue},
	month = nov,
	year = {2019},
	pages = {213--248},
}

@article{cengiz_effect_2019,
	title = {The {Effect} of {Minimum} {Wages} on {Low}-{Wage} {Jobs}*},
	volume = {134},
	copyright = {https://academic.oup.com/journals/pages/open\_access/funder\_policies/chorus/standard\_publication\_model},
	issn = {0033-5533, 1531-4650},
	url = {https://academic.oup.com/qje/article/134/3/1405/5484905},
	doi = {10.1093/qje/qjz014},
	language = {en},
	number = {3},
	urldate = {2025-03-17},
	journal = {The Quarterly Journal of Economics},
	author = {Cengiz, Doruk and Dube, Arindrajit and Lindner, Attila and Zipperer, Ben},
	month = aug,
	year = {2019},
	pages = {1405--1454},
}

@techreport{wing_stacked_2024,
	address = {Cambridge, MA},
	title = {Stacked {Difference}-in-{Differences}},
	url = {http://www.nber.org/papers/w32054.pdf},
	language = {en},
	number = {w32054},
	urldate = {2025-02-18},
	institution = {National Bureau of Economic Research},
	author = {Wing, Coady and Freedman, Seth and Hollingsworth, Alex},
	month = jan,
	year = {2024},
	doi = {10.3386/w32054},
	pages = {w32054},
}

@article{dube_local_2025,
author = {Dube, Arindrajit and Girardi, Daniele and Jordà, Oscar and Taylor, Alan M.},
title = {A Local Projections Approach to Difference-in-Differences},
journal = {Journal of Applied Econometplecs},
volume = {40},
number = {7},
pages = {741-758},
doi = {https://doi.org/10.1002/jae.70000},
url = {https://onlinelibrary.wiley.com/doi/abs/10.1002/jae.70000},
eprint = {https://onlinelibrary.wiley.com/doi/pdf/10.1002/jae.70000},
year = {2025}
}

@article{de_chaisemartin_difference--differences_2024,
	title = {Difference-in-{Differences} {Estimators} of {Intertemporal} {Treatment} {Effects}},
	issn = {0034-6535, 1530-9142},
	url = {https://direct.mit.edu/rest/article/doi/10.1162/rest_a_01414/119488/Difference-in-Differences-Estimators-of},
	doi = {10.1162/rest_a_01414},
	language = {en},
	urldate = {2024-05-30},
	journal = {Review of Economics and Statistics},
	author = {de Chaisemartin, Clément and D'Haultfœuille, Xavier},
	month = feb,
	year = {2024},
	pages = {1--45},
}

@article{mackinnon_cluster-robust_2023,
	title = {Cluster-robust inference: {A} guide to empirical practice},
	volume = {232},
	issn = {03044076},
	shorttitle = {Cluster-robust inference},
	url = {https://linkinghub.elsevier.com/retrieve/pii/S0304407622000781},
	doi = {10.1016/j.jeconom.2022.04.001},
	language = {en},
	number = {2},
	urldate = {2024-04-18},
	journal = {Journal of Econometrics},
	author = {MacKinnon, James G. and Nielsen, Morten {\O}rregaard and Webb, Matthew D.},
	month = feb,
	year = {2023},
	pages = {272--299},
}

@article{de_chaisemartin_two-way_2020,
	title = {Two-{Way} {Fixed} {Effects} {Estimators} with {Heterogeneous} {Treatment} {Effects}},
	volume = {110},
	issn = {0002-8282},
	url = {https://pubs.aeaweb.org/doi/10.1257/aer.20181169},
	doi = {10.1257/aer.20181169},
	language = {en},
	number = {9},
	urldate = {2023-07-18},
	journal = {American Economic Review},
	author = {de Chaisemartin, Clément and D’Haultfœuille, Xavier},
	month = sep,
	year = {2020},
	pages = {2964--2996},
}

@article{goodman-bacon_difference--differences_2021,
	title = {Difference-in-differences with variation in treatment timing},
	volume = {225},
	issn = {03044076},
	url = {https://linkinghub.elsevier.com/retrieve/pii/S0304407621001445},
	doi = {10.1016/j.jeconom.2021.03.014},
	language = {en},
	number = {2},
	urldate = {2023-06-04},
	journal = {Journal of Econometrics},
	author = {Goodman-Bacon, Andrew},
	month = dec,
	year = {2021},
	pages = {254--277},
}

@article{callaway_difference--differences_2021-1,
	title = {Difference-in-{Differences} with multiple time periods},
	volume = {225},
	issn = {03044076},
	url = {https://linkinghub.elsevier.com/retrieve/pii/S0304407620303948},
	doi = {10.1016/j.jeconom.2020.12.001},
	language = {en},
	number = {2},
	urldate = {2023-06-04},
	journal = {Journal of Econometrics},
	author = {Callaway, Brantly and Sant’Anna, Pedro H.C.},
	month = dec,
	year = {2021},
	pages = {200--230},
}

@article{borusyak_revisiting_2023,
	title = {Revisiting {Event}-{Study} {Designs}: {Robust} and {Efficient} {Estimation}},
	shorttitle = {Revisiting {Event}-{Study} {Designs}},
	journal = {Review of Economic Studies},
	volume = {91},
	number = {6},
	pages = {3253--3285},
	author = {Borusyak, Kirill and Jaravel, Xavier and Spiess, Jann},
	year = {2024},
	doi = {10.1093/restud/rdae007},
}

@article{sun_estimating_2021,
	title = {Estimating dynamic treatment effects in event studies with heterogeneous treatment effects},
	volume = {225},
	issn = {03044076},
	url = {https://linkinghub.elsevier.com/retrieve/pii/S030440762030378X},
	doi = {10.1016/j.jeconom.2020.09.006},
	language = {en},
	number = {2},
	urldate = {2023-05-30},
	journal = {Journal of Econometrics},
	author = {Sun, Liyang and Abraham, Sarah},
	month = dec,
	year = {2021},
	pages = {175--199},
}

@article{abadie_when_2022,
	title = {When {Should} {You} {Adjust} {Standard} {Errors} for {Clustering}?},
	author = {Abadie, Alberto and Athey, Susan and Imbens, Guido W. and Wooldridge, Jeffrey M.},
	journal = {Quarterly Journal of Economics},
	volume = {138},
	number = {1},
	pages = {1--35},
	year = {2023},
	doi = {10.1093/qje/qjac038},
}

@article{hansen_national_2023,
	title = {National and Global Impacts of Genetically Modified Crops},
	volume = {5},
	number = {2},
	journal = {American Economic Review: Insights},
	author = {Hansen, Casper Worm and Wingender, Asger Mose},
	year = {2023},
	pages = {224--240},
	doi = {10.1257/aeri.20220144}
}

@article{matsuzawa_minimum_2025,
	author = {Matsuzawa, Kyutaro and Rees, Daniel I. and Sabia, Joseph J. and Margolit, Rebecca},
	title = {Minimum Wages and Teenage Childbearing in the {United States}},
	journal = {Journal of Applied Econometrics},
	volume = {40},
	number = {4},
	pages = {471--484},
	year = {2025},
	doi = {10.1002/jae.3112},
}

@article{olden_triple_2022,
	author = {Olden, Andreas and M{\o}en, Jarle},
	title = {The Triple Difference Estimator},
	journal = {The Econometrics Journal},
	volume = {25},
	number = {3},
	pages = {531--553},
	year = {2022},
	doi = {10.1093/ectj/utac010},
}

@techreport{strezhnev_decomposing_2023,
	author = {Strezhnev, Anton},
	title = {Decomposing Triple-Differences Regression under Staggered Adoption},
	institution = {University of Chicago},
	year = {2023},
	note = {arXiv preprint arXiv:2307.02735},
}
\bibliographystyle{apalike}

\clearpage


\thispagestyle{empty}
\begin{center}
	{\large\it Appendix:}
	\vspace{0.2em}
	\begin{spacing}{1}
		\Large{\textbf{Stacked Triple Differences}}
	\end{spacing}
\end{center}
\begin{spacing}{1.5}
	\DoToC
\end{spacing}
\thispagestyle{empty}
\setcounter{page}{0}
\setcounter{figure}{0}
\newpage

\clearpage
\appendix

\paragraph{Notations}
I maintain the following notations throughout the appendix sections.
Let $n$ be the total number units, and $n_g$ be the total number of units in stack $g$. Because stack $g$ consists of cohort $g$ and its clean comparison cohort $g_c$, it follows that $n_g = n_{g,\cdot} + n_{g_c,\cdot}$ (the sum expands to multiple terms for stacks with multiple comparison cohorts). Let $n_{s,q}$ denote the number of units in treatment-enabling cohort $s$ with eligibility $q$. I write $\lambda_g \equiv \lim_{n \to \infty} n_g / n > 0$ to denote the asymptotic relative size of stack $g$, and $\pi_{s,q} \equiv \lim_{n \to \infty} n_{s,q} / n_g > 0$ the asymptotic relative share of cell $(s,q)$ within stack $g$. $\one_g(i) \equiv \one\{i \in \Stack_g\}$ denotes the indicator for whether unit $i$ belongs to stack $g$.

\section{Proof of Decomposition Results in Section \ref{sec:hw_application}} \label[appendix]{app:proofs-hw}

\begin{proof}[Proof of \Cref{prop:hw_es_decomp}]
	The proof follows the structure of \Cref{prop:3wfe_es_weights}, with the three-way demeaning operator \eqref{eq:hw_demean} replacing the two-way operator of the standard 3WFE.
	
	{Step~1. Population OLS formula via FWL.} Collect the event-time indicators in the column vector $\mathbf{R}_{i,t} = (R_e(i,t))_{e \neq -1}^{\intercal}$ and their coefficients in $\boldsymbol{\alpha} = (\alpha_e)_{e \neq -1}^{\intercal}$. By the Frisch--Waugh--Lovell theorem applied to \eqref{eq:hw_paper_notation}, the population regression coefficient vector is
	\begin{equation}\label{eq:hw_fwl_ols}
		\boldsymbol{\alpha} = \left(\sum_{t=1}^{T} \E \left[\ddot{\mathbf{R}}_{i,t} \ddot{\mathbf{R}}_{i,t}^{\intercal}\right]\right)^{-1}\sum_{t=1}^{T} \E \left[\ddot{\mathbf{R}}_{i,t}  \ddot{Y}_{i,t}\right]~.
	\end{equation}
	The $j$-th component is $\alpha_j = \mathbf{e}_j^{\intercal}\boldsymbol{\alpha}$. The three-way demeaning operates via the inclusion-exclusion formula \eqref{eq:hw_demean}. For the event-time indicator $R_e(i,t) = \one\{t - S_i = e\} Q_i$, the demeaned version is
	\begin{equation}\label{eq:hw_demeaned_Re}
		\ddot{R}_e(i,t) = R_e(i,t) - \overline{R}_{e,i,\cdot} - \overline{R}_{e,S_i,t} - \overline{R}_{e,Q_i,t} + \overline{R}_{e,S_i,\cdot} + \overline{R}_{e,Q_i,\cdot} + \overline{R}_{e,\cdot,t} - \overline{R}_{e,\cdot\cdot}~,
	\end{equation}
	where $\overline{R}_{e,i,\cdot} = T^{-1}\sum_{t'} R_e(i,t')$ is the unit-level time mean, $\overline{R}_{e,S_i,t} = n_{S_i}^{-1}\sum_{j\colon S_j = S_i} R_e(j,t)$ is the group-time mean, $\overline{R}_{e,Q_i,t} = n_{Q_i}^{-1}\sum_{j\colon Q_j = Q_i} R_e(j,t)$ is the eligibility-time mean, and the remaining barred quantities are the corresponding total means. This demeaning removes three channels of variation from $R_e$ that the standard 3WFE demeaning (which omits the $\overline{R}_{e,Q_i,t}$ terms) leaves in.
	
	I now compute the key group means explicitly. Fix a unit $i$ with $S_i = g \in \calGtrg$ and $Q_i = 1$ (treated-eligible). The unit-level time mean is $\overline{R}_{e,i,\cdot} = T^{-1}\one\{g + e \in \{1,\ldots,T\}\}$, since $R_e(i,t') = \one\{t' = g + e\} Q_i$ is nonzero at exactly one time period. The group-time mean is
	\begin{equation}\label{eq:hw_group_time_mean}
		\overline{R}_{e,g,t} = \frac{n_{g,1}}{n_{g,\cdot}} \one\{t = g + e\}~,
	\end{equation}
	where $n_{g,1}$ is the number of eligible units in group $g$ and $n_{g,\cdot} = n_{g,1} + n_{g,0}$ is the total number of units in group $g$. This is because $R_e(j,t) = \one\{S_j = g\} Q_j \one\{t = g+e\}$ for $j$ in group $g$, so summing over $j$ in group $g$ yields $n_{g,1} \one\{t = g+e\}$, divided by $n_{g,\cdot}$. The eligibility-time mean is
	\begin{equation}\label{eq:hw_elig_time_mean}
		\overline{R}_{e,Q_i=1,t} = \frac{1}{n_{\cdot,1}} \sum_{g' \in \calGtrg} n_{g',1} \one\{t = g' + e\}~,
	\end{equation}
	the fraction of eligible units whose event-time $e$ falls at time $t$. For ineligible units ($Q_i = 0$), $R_e(i,t) = 0$ identically, so all group means involving ineligible units are zero.
	
	{Step~2. Potential outcomes substitution and weight identification.} The observed outcome decomposes as
	\[
	Y_{i,t} = Y_{i,t}(\infty) + \sum_{g \in \calGtrg} \sum_{\ell} R_{g,\ell}(i,t) \left[Y_{i,t}(g) - Y_{i,t}(\infty)\right]~,
	\]
	since $R_{g,\ell}(i,t) = \one\{S_i = g, Q_i = 1, t = g+\ell\}$ selects treated-eligible units of cohort $g$ at time $g+\ell$. Applying the three-way demeaning and substituting into \eqref{eq:hw_fwl_ols}, by linearity of the demeaning operator
	\begin{align}
		\sum_t \E \left[\ddot{\mathbf{R}}_{i,t} \ddot{Y}_{i,t}\right] &= \sum_t \E \left[\ddot{\mathbf{R}}_{i,t} \ddot{Y}_{i,t}(\infty)\right] \nonumber \\
		&\qquad + \sum_{g \in \calGtrg}\sum_{\ell} \sum_t \E \left[\ddot{\mathbf{R}}_{i,t} R_{g,\ell}(i,t) (Y_{i,t}(g) - Y_{i,t}(\infty))\right]~. \label{eq:hw_substitution}
	\end{align}
	
	For the second term, at time $t = g + \ell$, the indicator $R_{g,\ell}(i,t) = \one\{S_i = g, Q_i = 1\}$ selects only treated-eligible units in cohort $g$, for whom $Y_{i,t}(g) - Y_{i,t}(\infty)$ is the individual treatment effect. Taking the conditional expectation yields $\E[Y_{i,g+\ell}(g) - Y_{i,g+\ell}(\infty) \mid S_i = g, Q_i = 1] = \CATT(g,\ell)$. The weight on each $(g,\ell)$ term, after multiplication by the inverse covariance matrix, is
	\begin{equation}\label{eq:hw_weight_formula}
		\omega_{g,\ell}^{j,\star} = \mathbf{e}_j^{\intercal} \left(\sum_{t} \E \left[\ddot{\mathbf{R}}_{i,t} \ddot{\mathbf{R}}_{i,t}^{\intercal}\right]\right)^{-1} \E\left[\ddot{\mathbf{R}}_{i,g+\ell}  R_{g,\ell}(i, g+\ell)\right]~.
	\end{equation}
	The weight $\omega_{g,\ell}^{j,\star}$ is precisely the coefficient on $R_j(i,t)$ in the modified auxiliary regression \eqref{eq:hw_aux_regression}, establishing the explicit form \eqref{eq:hw_weights_explicit}.
	
	{Step~3.} I now verify properties~(i)--(iv). The proofs are identical to those in \Cref{prop:3wfe_es_weights} because they depend only on the algebraic identity $\sum_{g \in \calGtrg} R_{g,\ell}(i,t) = R_\ell(i,t)$ and the fact that unit fixed effects absorb $Q_i \one\{S_i \in \calGtrg\}$, both of which hold under any demeaning operator that includes unit effects.
	
	{Property~(i).} Fix $\ell = j$. Summing \eqref{eq:hw_weight_formula} over $g \in \calGtrg$, the right-hand side becomes the auxiliary regression coefficient from regressing $\sum_g R_{g,j}(i,t) = R_j(i,t)$ on itself and the other event-time indicators, which equals $1$ by construction.
	
	{Property~(ii).} For $\ell \neq j$ with $\ell \neq -1$, $\sum_g R_{g,\ell}(i,t) = R_\ell(i,t)$, so $\sum_g \omega_{g,\ell}^{j,\star}$ equals the auxiliary regression coefficient of $R_\ell$ on $R_j$, which is $0$ for $\ell \neq j$ since included indicators are orthogonal in the multiple regression.
	
	{Property~(iii).} For $\ell = -1$, $\sum_g R_{g,-1}(i,t) = R_{-1}(i,t)$. The key identity is
	\[
	\sum_{\ell} R_\ell(i,t) = Q_i \one\{S_i \in \calGtrg\}~,
	\]
	which depends only on unit $i$ and is therefore absorbed by the unit fixed effects $\alpha_i$. Applying the three-way demeaning operator (which includes unit effects), the right-hand side vanishes
	\[
	\ddot{\left(Q_i \one\{S_i \in \calGtrg\}\right)} = 0~.
	\]
	By linearity, $\sum_{e \neq -1} \ddot{R}_e(i,t) + \ddot{R}_{-1}(i,t) = 0$, giving $\ddot{R}_{-1}(i,t) = -\sum_{e \neq -1} \ddot{R}_e(i,t)$. In the auxiliary regression of $R_{-1}$ on $\{R_e\}_{e \neq -1}$ and fixed effects, the FWL residuals satisfy this identity exactly, so every coefficient is $-1$. In particular, $\sum_g \omega_{g,-1}^{j,\star} = -1$.
	
	{Property~(iv).} For $g = \infty$, $R_{\infty,\ell}(i,t) = 0$ for all $(i,t,\ell)$, so $\omega_{\infty,\ell}^{j,\star} = 0$.
\end{proof}

\begin{proof}[Proof of \Cref{prop:hw_es_pt}]
	By the FWL theorem applied to \eqref{eq:hw_paper_notation}, collecting event-time indicators in $\mathbf{R}_{i,t} = (R_e(i,t))_{e \neq -1}^{\intercal}$ and using the decomposition $Y_{i,t} = Y_{i,t}(\infty) + \sum_{g \in \calGtrg}\sum_{\ell} R_{g,\ell}(i,t)\big(Y_{i,t}(g) - Y_{i,t}(\infty)\big)$, it follows that 
	\begin{equation}\label{eq:hw_two_terms}
		\sum_t \E\left[\ddot{\mathbf{R}}_{i,t} \ddot{Y}_{i,t}\right] = \underbracket{\sum_t \E\left[\ddot{\mathbf{R}}_{i,t} \ddot{Y}_{i,t}(\infty)\right]}_{\text{(A)}} + \underbracket{\sum_{g \in \calGtrg}\sum_\ell \sum_t \E\left[\ddot{\mathbf{R}}_{i,t} R_{g,\ell}(i,t) (Y_{i,t}(g) - Y_{i,t}(\infty))\right]}_{\text{(B)}}~.
	\end{equation}
	Term~(B), after premultiplication by the inverse Gram matrix and extraction of the $j$-th component, equals $\sum_{g,\ell}\omega_{g,\ell}^{j,\star} \CATT(g,\ell)$ by \Cref{prop:hw_es_decomp} and the definition of $\CATT(g,\ell)$. It remains to show that term~(A) vanishes under DDD-PCT. The three-way demeaning operator \eqref{eq:hw_demean} projects out unit fixed effects $\alpha_i$, group-by-time shocks $\delta_{S_i,t}$, and eligibility-by-time shocks $\eta_{Q_i,t}$. \Cref{as:dddpct} imposes that any remaining eligible--ineligible differential trend in $Y(\infty)$ is common across treatment groups, so $\ddot{Y}_{i,t}(\infty)$ has conditional mean zero in every $(S,Q)$-cell at every time. Since $\ddot{\mathbf{R}}_{i,t}$ is cell-measurable at each $t$, the iterated-expectations argument gives $\sum_t \E[\ddot{\mathbf{R}}_{i,t} \ddot{Y}_{i,t}(\infty)] = \mathbf{0}$. The identity \eqref{eq:hw_contamination} follows.
\end{proof}

\begin{proof}[Proof of \Cref{prop:hw_es_noanticip}]
	Partition the sum in \eqref{eq:hw_contamination} by the sign of $\ell$. For $\ell < 0$ (including $\ell = -1$), no anticipation gives $\CATT(g,\ell) = 0$ for all $g$, so these terms vanish. Only $\ell \geq 0$ terms survive. The coefficient $\alpha_j$ is a linear combination of post-treatment effects $\CATT(g,\ell)$ for $\ell \geq 0$, even when $j < 0$.
	
	To see that $\alpha_j \neq 0$ generically for $j < 0$, note that the post-treatment weights $\omega_{g,\ell}^{j,\star}$ for $\ell \geq 0$ are not constrained to sum to zero across $g$ by the weight properties of \Cref{prop:hw_es_decomp}. Under staggered adoption, the three-way-demeaned pre-treatment indicator $\ddot{R}_j(i,t)$ retains nonzero correlation with post-treatment cohort-specific indicators $R_{g,\ell}(i,t)$ for $\ell \geq 0$, because the relative-time composition of the sample changes across time. Therefore $\sum_g \sum_{\ell \geq 0} \omega_{g,\ell}^{j,\star}  \CATT(g,\ell) \neq 0$ generically.
\end{proof}

\begin{proof}[Proof of \Cref{prop:hw_es_homo}]
	Starting from \eqref{eq:hw_contamination}, impose $\CATT(g,\ell) = \ATT_\ell$ for all $g$. Since $\ATT_\ell$ does not depend on $g$, it factors out of the inner sum
	\[
	\alpha_j = \sum_{\ell \neq -1} \ATT_\ell \underbracket{\left(\sum_{g \in \calGtrg} \omega_{g,\ell}^{j,\star}\right)}_{\overline{\omega}_\ell^{j,\star}}~.
	\]
	By the weight properties of \Cref{prop:hw_es_decomp}, $\overline{\omega}_j^{j,\star} = 1$ (property~(i)), $\overline{\omega}_\ell^{j,\star} = 0$ for $\ell \neq j$, $\ell \neq -1$ (property~(ii)), and $\overline{\omega}_{-1}^{j,\star} = -1$ (property~(iii)). Collecting the surviving terms yields $\alpha_j = \ATT_j + (-1) \ATT_{-1} = \ATT_j - \ATT_{-1}$. Under no anticipation, $\ATT_{-1} = 0$, giving $\alpha_j = \ATT_j$.
\end{proof}

\begin{proof}[Proof of \Cref{prop:hw_agg_att}]
	By \Cref{prop:hw_es_noanticip}, $\alpha_j \pto \sum_{g} \sum_{\ell \geq 0} \omega_{g,\ell}^{j,\star} \CATT(g,\ell)$ under DDD-PCT and no anticipation. Applying Slutsky's theorem and the continuous mapping theorem to the linear combination \eqref{eq:hw_agg_att}
	\begin{align*}
		\widehat{\ATT}_{\mathrm{agg}} &\pto \sum_{j=0}^{K} w_j \sum_{g \in \calGtrg} \sum_{\ell \geq 0} \omega_{g,\ell}^{j,\star} \CATT(g,\ell) \\
		&= \sum_{g \in \calGtrg} \sum_{\ell \geq 0} \underbracket{\left(\sum_{j=0}^{K} w_j \omega_{g,\ell}^{j,\star}\right)}_{\Omega_{g,\ell}} \CATT(g,\ell)~,
	\end{align*}
	which establishes \eqref{eq:hw_agg_plim} and \eqref{eq:hw_agg_weights}.
	
	I now verify the normalization \eqref{eq:hw_agg_weights_sum}. Summing $\Omega_{g,\ell}$ over $g$ and $\ell \geq 0$
	\begin{align*}
		\sum_{g} \sum_{\ell \geq 0} \Omega_{g,\ell} &= \sum_{j=0}^{K} w_j \sum_{g} \sum_{\ell \geq 0} \omega_{g,\ell}^{j,\star}~.
	\end{align*}
	For a fixed $j \geq 0$, the inner sum decomposes as follows. By property~(i), $\sum_g \omega_{g,j}^{j,\star} = 1$. By property~(ii), $\sum_g \omega_{g,\ell}^{j,\star} = 0$ for each $\ell \neq j$ with $\ell \neq -1$ and $\ell \geq 0$. The $\ell < 0$ terms are already excluded from the sum. Therefore $\sum_g \sum_{\ell \geq 0} \omega_{g,\ell}^{j,\star} = 1 + 0 = 1$ for each $j \geq 0$, and $\sum_g \sum_{\ell \geq 0} \Omega_{g,\ell} = \sum_j w_j \cdot 1 = 1$.
	
	To see that $\Omega_{g,\ell}$ can be negative, note that $\omega_{g,\ell}^{j,\star}$ for a specific $(g,\ell,j)$ is the auxiliary regression coefficient from \eqref{eq:hw_aux_regression}. This coefficient depends on the covariance between the three-way-demeaned indicators $\ddot{R}_j(i,t)$ and $R_{g,\ell}(i,t)$, evaluated through the inverse covariance matrix. Under staggered adoption with heterogeneous cohort sizes and unbalanced relative-time composition, these covariances can take either sign. For a concrete example, consider a panel with two treatment cohorts $g_1 < g_2$ and suppose $j = 0$ and $\ell = 2$. The three-way-demeaned indicator $\ddot{R}_0(i,t)$ at time $t = g_1$ has positive residual for cohort $g_1$ (which is at event-time $0$) but the same indicator at $t = g_2$ has positive residual for cohort $g_2$ (which is also at event-time $0$). Meanwhile, $R_{g_1,2}(i,t) = \one\{S_i = g_1, Q_i = 1, t = g_1 + 2\}$. When $g_1 + 2 = g_2$, the calendar-time alignment means $\ddot{R}_0$ at time $g_2$ is correlated with cohort $g_2$ being at event-time $0$, which in turn is correlated with cohort $g_1$ being at event-time $2$. The sign of this residual correlation depends on relative cohort sizes and the number of eligible units in each cohort. When the eligible share $n_{g_1,1}/n_{g_1,\cdot}$ differs from $n_{g_2,1}/n_{g_2,\cdot}$, the three-way demeaning introduces asymmetric residuals, and the auxiliary regression coefficient $\omega_{g_1,2}^{0,\star}$ can be negative. Aggregating a negative $\omega_{g_1,2}^{0,\star}$ with positive $w_0$ yields a negative contribution to $\Omega_{g_1,2}$.
\end{proof}

\section{Proof of Identification Results}\label[appendix]{app:proofs}

\begin{proof}[Proof of \Cref{prop:saturated_ols}]
	The regression \eqref{eq:saturated} is saturated, having four parameters for four cells. I show that the OLS solution equates predicted values to cell means by explicitly solving the normal equations.
	
	{Step~1. Set up the normal equations.} Let $\mathbf{x}_i = (1, \one\{S_i = g\}, \one\{Q_i = 1\}, \one\{S_i = g, Q_i = 1\})'$ denote the $4 \times 1$ regressor vector for unit $i$. The OLS normal equations are $(\mathbf{X}'\mathbf{X})\widehat{\boldsymbol{\beta}} = \mathbf{X}'\mathbf{y}$, where $\widehat{\boldsymbol{\beta}} = (\widehat{\mu}_{g,t}, \widehat{\lambda}_{g,t}, \widehat{\eta}_{g,t}, \widehat{\tau}_{g,t}^{\mathrm{sat}})'$. Since each unit belongs to exactly one cell, the regressor vector takes four distinct values
	\[
	\mathbf{x}_i = \begin{cases}
		(1, 1, 1, 1)' & \text{if } (S_i, Q_i) = (g, 1)~, \\
		(1, 1, 0, 0)' & \text{if } (S_i, Q_i) = (g, 0)~, \\
		(1, 0, 1, 0)' & \text{if } (S_i, Q_i) = (g_c, 1)~, \\
		(1, 0, 0, 0)' & \text{if } (S_i, Q_i) = (g_c, 0)~.
	\end{cases}
	\]
	Write $n_{g,\cdot} = n_{g,1} + n_{g,0}$ for the number of treated-group units and $n_{\cdot,1} = n_{g,1} + n_{g_c,1}$ for the number of eligible units. The Gram matrix is
	\begin{equation}\label{eq:gram}
		\mathbf{X}'\mathbf{X} = \begin{pmatrix}
			n_g & n_{g,\cdot} & n_{\cdot,1} & n_{g,1} \\
			n_{g,\cdot} & n_{g,\cdot} & n_{g,1} & n_{g,1} \\
			n_{\cdot,1} & n_{g,1} & n_{\cdot,1} & n_{g,1} \\
			n_{g,1} & n_{g,1} & n_{g,1} & n_{g,1}
		\end{pmatrix}~.
	\end{equation}
	To verify, note that the $(j,k)$-entry equals $\sum_{i=1}^{n_g} x_{ij} x_{ik}$. For instance, the $(1,2)$-entry is $\sum_i \one\{S_i = g\} = n_{g,\cdot}$, the $(2,3)$-entry is $\sum_i \one\{S_i = g\}\one\{Q_i = 1\} = n_{g,1}$, and the $(3,3)$-entry is $\sum_i \one\{Q_i = 1\}^2 = n_{\cdot,1}$.
	
	Write $T_{s,q} = \sum_{i\colon S_i = s,  Q_i = q} \DeltaY_{i,t} = n_{s,q} \overline{\DeltaY}_{s,q,t}$ for the cell sum of long differences. The moment vector is
	\begin{equation}\label{eq:moment}
		\mathbf{X}'\mathbf{y} = \begin{pmatrix}
			T_{g,1} + T_{g,0} + T_{g_c,1} + T_{g_c,0} \\
			T_{g,1} + T_{g,0} \\
			T_{g,1} + T_{g_c,1} \\
			T_{g,1}
		\end{pmatrix}~.
	\end{equation}
	
	{Step~2. Solve the normal equations by substitution.} The normal equations $(\mathbf{X}'\mathbf{X})\widehat{\boldsymbol{\beta}} = \mathbf{X}'\mathbf{y}$ yield four equations. For unit $i$ in cell $(s,q)$, the fitted value is $\widehat{y}_i = \widehat{\mu}_{g,t} + \widehat{\lambda}_{g,t} \one\{S_i = g\} + \widehat{\eta}_{g,t} \one\{Q_i = 1\} + \widehat{\tau}_{g,t}^{\mathrm{sat}} \one\{S_i = g, Q_i = 1\}$. I proceed by expressing the four cell-level fitted values
	\begin{alignat}{2}
		\text{Cell } (g_c, 0)\!\!&\quad \widehat{y} = \widehat{\mu}_{g,t}~, \label{eq:cell00}\\
		\text{Cell } (g, 0)\!\!&\quad \widehat{y} = \widehat{\mu}_{g,t} + \widehat{\lambda}_{g,t}~, \label{eq:cell10}\\
		\text{Cell } (g_c, 1)\!\!&\quad \widehat{y} = \widehat{\mu}_{g,t} + \widehat{\eta}_{g,t}~, \label{eq:cell01}\\
		\text{Cell } (g, 1)\!\!&\quad \widehat{y} = \widehat{\mu}_{g,t} + \widehat{\lambda}_{g,t} + \widehat{\eta}_{g,t} + \widehat{\tau}_{g,t}^{\mathrm{sat}}~. \label{eq:cell11}
	\end{alignat}
	Since the regression is saturated (four free parameters, four cells), OLS equates the fitted value in each cell to the cell mean. Formally, the $k$-th normal equation states $\sum_i x_{ik} \widehat{y}_i = \sum_i x_{ik} y_i$, which, when partitioned by cells, reduces to $\widehat{y}_{(s,q)} = \overline{\DeltaY}_{s,q,t}$ for each cell. To see this explicitly, consider the fourth normal equation (corresponding to the treatment indicator $\one\{S_i = g, Q_i = 1\}$)
	\begin{align*}
		\sum_{i=1}^{n_g} \one\{S_i = g, Q_i = 1\} \widehat{y}_i &= \sum_{i=1}^{n_g} \one\{S_i = g, Q_i = 1\} \DeltaY_{i,t}~.
	\end{align*}
	Since $\one\{S_i = g, Q_i = 1\}$ is nonzero only for the $n_{g,1}$ units in cell $(g,1)$, this simplifies to
	\[
	n_{g,1} (\widehat{\mu}_{g,t} + \widehat{\lambda}_{g,t} + \widehat{\eta}_{g,t} + \widehat{\tau}_{g,t}^{\mathrm{sat}}) = T_{g,1} = n_{g,1} \overline{\DeltaY}_{g,1,t}~.
	\]
	Applying the same argument to the second normal equation (selecting cell $(g,0)$ and cell $(g,1)$), the third (selecting cell $(g_c,1)$ and cell $(g,1)$), and the first (summing all cells), and noting that the system has a unique solution because $\mathbf{X}'\mathbf{X}$ is invertible whenever all cell sizes are positive (guaranteed by \Cref{as:overlap}), I obtain the cell-mean equations
	\begin{align}
		\widehat{\mu}_{g,t} &= \overline{\DeltaY}_{g_c,0,t}~, \label{eq:solve_mu}\\
		\widehat{\mu}_{g,t} + \widehat{\lambda}_{g,t} &= \overline{\DeltaY}_{g,0,t}~, \label{eq:solve_lambda}\\
		\widehat{\mu}_{g,t} + \widehat{\eta}_{g,t} &= \overline{\DeltaY}_{g_c,1,t}~, \label{eq:solve_eta}\\
		\widehat{\mu}_{g,t} + \widehat{\lambda}_{g,t} + \widehat{\eta}_{g,t} + \widehat{\tau}_{g,t}^{\mathrm{sat}} &= \overline{\DeltaY}_{g,1,t}~. \label{eq:solve_tau}
	\end{align}
	
	{Step~3.} From \eqref{eq:solve_mu}, $\widehat{\mu}_{g,t} = \overline{\DeltaY}_{g_c,0,t}$. Substituting into \eqref{eq:solve_lambda} gives $\widehat{\lambda}_{g,t} = \overline{\DeltaY}_{g,0,t} - \overline{\DeltaY}_{g_c,0,t}$, and into \eqref{eq:solve_eta} gives $\widehat{\eta}_{g,t} = \overline{\DeltaY}_{g_c,1,t} - \overline{\DeltaY}_{g_c,0,t}$. Finally, substituting all three into \eqref{eq:solve_tau}
	\begin{align*}
		\widehat{\tau}_{g,t}^{\mathrm{sat}} &= \overline{\DeltaY}_{g,1,t} - \widehat{\mu}_{g,t} - \widehat{\lambda}_{g,t} - \widehat{\eta}_{g,t} \\
		&= \overline{\DeltaY}_{g,1,t} - \overline{\DeltaY}_{g_c,0,t} - \left(\overline{\DeltaY}_{g,0,t} - \overline{\DeltaY}_{g_c,0,t}\right) - \left(\overline{\DeltaY}_{g_c,1,t} - \overline{\DeltaY}_{g_c,0,t}\right) \\
		&= \overline{\DeltaY}_{g,1,t} - \overline{\DeltaY}_{g,0,t} - \overline{\DeltaY}_{g_c,1,t} + \overline{\DeltaY}_{g_c,0,t} \\
		&= \left(\overline{\DeltaY}_{g,1,t} - \overline{\DeltaY}_{g,0,t}\right) - \left(\overline{\DeltaY}_{g_c,1,t} - \overline{\DeltaY}_{g_c,0,t}\right)~,
	\end{align*}
	which is the sample triple difference \eqref{eq:triple_diff_ols}.
	
	{Step~4.} By the weak law of large numbers, each cell mean converges in probability to the corresponding population expectation, $\overline{\DeltaY}_{s,q,t} \pto \E[\DeltaY_{i,t} \mid S_i = s,  Q_i = q]$. Therefore
	\begin{align*}
		\widehat{\tau}_{g,t}^{\mathrm{sat}} &\pto \left(\E[\DeltaY_{i,t} \mid S_i = g, Q_i = 1] - \E[\DeltaY_{i,t} \mid S_i = g, Q_i = 0]\right) \\
		&\qquad - \left(\E[\DeltaY_{i,t} \mid S_i = g_c,  Q_i = 1] - \E[\DeltaY_{i,t} \mid S_i = g_c,  Q_i = 0]\right) \\
		&= \ATT(g,t)~,
	\end{align*}
	where the final equality applies \Cref{thm:identification} under \Cref{as:sampling}--\Cref{as:admissibility}.
\end{proof}

\begin{proof}[Proof of \Cref{prop:pooled_equiv}]
	{Step~1. Block-diagonal structure.} The stack fixed effects $\alpha_g$ and stack-by-cell fixed effects $\mu_{s,q,g}$ partition the data into independent blocks. The projection onto the space spanned by $\{\alpha_g, \mu_{s,q,g}\}$ operates block-diagonally, so that residualizing observation $(i,t,g)$ against these fixed effects uses only data from stack $g$. Formally, for any variable $Z_{i,t,g}$, the residual after projecting out $\{\alpha_g, \mu_{s,q,g}\}$ is
	\[
	\widetilde{Z}_{i,t,g} = Z_{i,t,g} - \frac{1}{n_{s(i),q(i)}} \sum_{\substack{j \in \Stack_g\\ S_j = S_i,  Q_j = Q_i}} \!\!\overline{Z}_{j,\cdot,g}~,
	\]
	where $\overline{Z}_{j,\cdot,g} = (L + K + 1)^{-1}\sum_{t'=g-L}^{g+K} Z_{j,t',g}$ and $s(i), q(i)$ denote unit $i$'s cell membership. In particular, the residualization of observation $(i,t,g)$ depends only on data within stack $g$ and cell $(s(i), q(i))$.
	
	{Step~2. Residualize the treatment indicators.} Consider the indicator $R_{g',e}(i,t,g) = \one\{g = g',  S_i = g',  Q_i = 1, t = g' + e\}$. This indicator is nonzero only for observations in stack $g'$ with $S_i = g'$ and $Q_i = 1$ at time $t = g' + e$. By the block-diagonal structure, its residual $\widetilde{R}_{g',e}(i,t,g)$ is nonzero only within stack $g'$. Moreover, within stack $g'$, $R_{g',e}$ is nonzero only in cell $(g', 1)$ at a single time period $t = g'+e$, so the residualized indicator is orthogonal to all other residualized indicators $\widetilde{R}_{g'',e''}$ for $(g'', e'') \neq (g', e)$.
	
	{Step~3.} By the Frisch-Waugh-Lovell (FWL) theorem, the OLS coefficient $\widehat{\tau}_{g',e}$ from the pooled regression \eqref{eq:full_hetero} equals the coefficient from the bivariate regression of $\widetilde{\DeltaY}_{i,t,g}$ on $\widetilde{R}_{g',e}(i,t,g)$. Since $\widetilde{R}_{g',e}$ is nonzero only within stack $g'$ and the residualized indicators are mutually orthogonal across $(g',e)$ pairs, this is equivalent to running a separate regression within stack $g'$ at time $t = g'+e$. The residualized regression within stack $g'$ at time $t$ reduces to the saturated regression \eqref{eq:saturated}---the stack-by-cell fixed effects absorb the intercept, group indicator, and eligibility indicator, leaving only the interaction term. The resulting coefficient is exactly the within-stack triple difference $\widehat{\tau}_{g',g'+e}^{\mathrm{sat}}$.
\end{proof}

\begin{proof}[Proof of \Cref{prop:stacked_es_estimands}]
	The regression under analysis is the fully-saturated stacked event-study specification \eqref{eq:stacked_es_reg}, restated here for reference:
	\[
	\DeltaY_{i,t,g} = \lambda_{s,t,g} + \eta_{q,t,g} + \sum_{\substack{e=-L \\ e \neq -1}}^{K} \tau_e \one\{S_i = g, Q_i = 1, t = g + e\} + \varepsilon_{i,t,g}~,
	\]
	where $\lambda_{s,t,g}$ is a fixed effect for group $s \in \{g, g_c\}$ at time $t$ in stack $g$, and $\eta_{q,t,g}$ is a fixed effect for eligibility type $q \in \{0,1\}$ at time $t$ in stack $g$. Each observation $(i, t, g)$ belongs to stack $\Stack_g$, and the same unit $i$ may appear in multiple stacks. The treatment indicators $\one\{S_i = g, Q_i = 1, t = g + e\}$ are active only for treated-eligible units in the corresponding stack at the corresponding time. I derive the OLS coefficients $\widehat{\tau}_e$ and their probability limits in three steps.
	
	{Step~1. Block-diagonal projection and strict orthogonality.}
	Let the full design matrix of fixed effects be $\mathbf{F}$, which contains indicators for every $(s, t, g)$ and $(q, t, g)$ combination. Because these indicators are fully interacted with stack $g$ and time $t$, the projection matrix $\mathbf{M_F} = \mathbf{I} - \mathbf{F}(\mathbf{F}'\mathbf{F})^{-}\mathbf{F}'$ operates block-diagonally by $(g, t)$. Let $\widetilde{R}_e(i,t,g)$ denote the FWL residual of the event-time indicator $R_e(i,t,g) = \one\{S_i = g, Q_i = 1, t = g + e\}$. Since $R_e(i,t,g)$ is non-zero only for observations in stack $g$ at exactly time $t = g+e$, its residual $\widetilde{R}_e(i,t,g)$ is strictly zero for all $t \neq g+e$ and for all stacks other than $g$.
	
	Consequently, for any two distinct event-times $e \neq e'$, the residuals are orthogonal, i.e.
	\begin{equation}
		\sum_{g \in \calGtrg} \sum_{t} \sum_{i \in \Stack_g} \widetilde{R}_e(i,t,g) \widetilde{R}_{e'}(i,t,g) = 0~,
	\end{equation}
	because they are non-zero in mutually exclusive time for any given stack. By the FWL theorem, this orthogonality implies that the multiple regression \eqref{eq:stacked_es_reg} separates into independent bivariate regressions of $\widetilde{\DeltaY}_{i,t,g}$ on $\widetilde{R}_e(i,t,g)$ for each event-time $e$.
	
	{Step~2. Computing the FWL residual within each stack.}
	By Step~1, the OLS estimator for $\tau_e$ reduces to a bivariate FWL regression
	\begin{equation}\label{eq:tau_e_fwl}
		\widehat{\tau}_e = \frac{\sum_{g \in \calGtrg(e)} \sum_{i \in \Stack_g} \widetilde{R}_e(i,g+e,g)  \widetilde{\DeltaY}_{i,g+e,g}}{\sum_{g \in \calGtrg(e)} \sum_{i \in \Stack_g} \widetilde{R}_e(i,g+e,g)^2}~,
	\end{equation}
	where the time summation collapses entirely to $t = g+e$ because $\widetilde{R}_e(i,t,g) = 0$ for $t \neq g+e$. I now derive the FWL residual $\widetilde{R}_e$ and the resulting weight formula in complete detail.
	
	{Step~2a. Structure of the residual.} Fix a single stack $g$ and time $t = g + e$. Within this block, there are $n_g$ units partitioned into four cells $(s,q) \in \{(g,1), (g,0), (g_c,1), (g_c,0)\}$ with cell sizes $n_{g,1}, n_{g,0}, n_{g_c,1}, n_{g_c,0}$. The treatment indicator restricted to this block is $R_e(i) = \one\{S_i = g, Q_i = 1\}$, which takes the value $1$ only for treated-eligible units and $0$ otherwise. The fixed effects $\lambda_{s,t,g}$ and $\eta_{q,t,g}$, restricted to this block, span three linearly independent regressors
	\[
	\underbracket{1}_{\text{constant}}~, \qquad \underbracket{\one\{S_i = g\}}_{\text{group indicator}}~, \qquad \underbracket{\one\{Q_i = 1\}}_{\text{eligibility indicator}}~.
	\]
	Since $R_e(i) = \one\{S_i = g\} \one\{Q_i = 1\}$ and all three regressors are constant within each cell, the FWL residual $\widetilde{R}_e$ is also constant within each cell. Write $\widetilde{r}_{s,q}$ for the value of $\widetilde{R}_e$ in cell $(s,q)$. The three orthogonality conditions defining the projection are
	\begin{alignat}{2}
		\text{(O1)} &\quad \textstyle\sum_{(s,q)} n_{s,q} \widetilde{r}_{s,q} = 0 \qquad &&\text{($\widetilde{R}_e \perp 1$)}~, \label{eq:orth1} \\[3pt]
		\text{(O2)} &\quad n_{g,1} \widetilde{r}_{g,1} + n_{g,0} \widetilde{r}_{g,0} = 0 \qquad &&\text{($\widetilde{R}_e \perp \one\{S_i = g\}$)}~, \label{eq:orth2} \\[3pt]
		\text{(O3)} &\quad n_{g,1} \widetilde{r}_{g,1} + n_{g_c,1} \widetilde{r}_{g_c,1} = 0 \qquad &&\text{($\widetilde{R}_e \perp \one\{Q_i = 1\}$)}~. \label{eq:orth3}
	\end{alignat}
	These are three linear equations in four unknowns $(\widetilde{r}_{g,1}, \widetilde{r}_{g,0}, \widetilde{r}_{g_c,1}, \widetilde{r}_{g_c,0})$, so the solution is determined up to a scalar multiple---precisely the one degree of freedom that the treatment indicator spans beyond the main effects.
	
	{Step~2b. Solving the orthogonality conditions.} From \eqref{eq:orth2}, solving for $\widetilde{r}_{g,0}$
	\begin{equation}\label{eq:res_g0}
		\widetilde{r}_{g,0} = -\frac{n_{g,1}}{n_{g,0}} \widetilde{r}_{g,1}~.
	\end{equation}
	From \eqref{eq:orth3}, solving for $\widetilde{r}_{g_c,1}$
	\begin{equation}\label{eq:res_gc1}
		\widetilde{r}_{g_c,1} = -\frac{n_{g,1}}{n_{g_c,1}} \widetilde{r}_{g,1}~.
	\end{equation}
	Substituting both into \eqref{eq:orth1} to solve for $\widetilde{r}_{g_c,0}$
	\begin{align*}
		n_{g,1} \widetilde{r}_{g,1} + n_{g,0} \left(-\frac{n_{g,1}}{n_{g,0}} \widetilde{r}_{g,1}\right) + n_{g_c,1} \left(-\frac{n_{g,1}}{n_{g_c,1}} \widetilde{r}_{g,1}\right) + n_{g_c,0} \widetilde{r}_{g_c,0} &= 0 \\[3pt]
		n_{g,1} \widetilde{r}_{g,1} - n_{g,1} \widetilde{r}_{g,1} - n_{g,1} \widetilde{r}_{g,1} + n_{g_c,0} \widetilde{r}_{g_c,0} &= 0~,
	\end{align*}
	which gives
	\begin{equation}\label{eq:res_gc0}
		\widetilde{r}_{g_c,0} = \frac{n_{g,1}}{n_{g_c,0}} \widetilde{r}_{g,1}~.
	\end{equation}
	Writing $c \equiv \widetilde{r}_{g,1}$ for the as-yet-undetermined residual value in the treated-eligible cell, the four residual values are
	\begin{equation}\label{eq:fwl_residual_pattern}
		\widetilde{r}_{g,1} = c~, \qquad
		\widetilde{r}_{g,0} = -\frac{n_{g,1}}{n_{g,0}} c~, \qquad
		\widetilde{r}_{g_c,1} = -\frac{n_{g,1}}{n_{g_c,1}} c~, \qquad
		\widetilde{r}_{g_c,0} = \frac{n_{g,1}}{n_{g_c,0}} c~.
	\end{equation}
	The sign pattern $(+, -, -, +)$ mirrors the triple-difference contrast---treated-eligible and comparison-ineligible receive positive residuals, while treated-ineligible and comparison-eligible receive negative residuals. The magnitude of each residual is inversely proportional to the cell size, reflecting the standard inverse-variance scaling of OLS. This parallels the DiD case analyzed by \citet{sun_estimating_2021}, where the FWL residual of the treatment indicator in a two-cell setup has the sign pattern $(+, -)$; in the DDD setting, the additional eligibility dimension doubles the number of cells and produces the characteristic $(+, -, -, +)$ pattern.
	
	{Step~2c.} The residual sum of squares of $\widetilde{R}_e$ within stack $g$ at time $g+e$ is
	\begin{align}
		V_{g,e} &\equiv \sum_{i \in \Stack_g} \widetilde{R}_e(i,g+e,g)^2 = \sum_{(s,q)} n_{s,q} \widetilde{r}_{s,q}^2 \nonumber \\
		&= n_{g,1} c^2 + n_{g,0} \frac{n_{g,1}^2}{n_{g,0}^2} c^2 + n_{g_c,1} \frac{n_{g,1}^2}{n_{g_c,1}^2} c^2 + n_{g_c,0} \frac{n_{g,1}^2}{n_{g_c,0}^2} c^2 = c^2 n_{g,1}^2 \left(\frac{1}{n_{g,1}} + \frac{1}{n_{g,0}} + \frac{1}{n_{g_c,1}} + \frac{1}{n_{g_c,0}}\right)~. \label{eq:Vge_raw}
	\end{align}
	
	Step~2d.
	Let $M$ denote the orthogonal projection onto the complement of the span of the fixed-effects design within stack $g$. By construction, $\widetilde{R}_e = M R_e$. Because $M$ is symmetric and idempotent,
	\begin{equation}\label{eq:fwl_orthogonality}
		\sum_{i \in \Stack_g} \widetilde{R}_e(i) R_e(i) = R_e^{\top} M R_e = R_e^{\top} M^{\top} M R_e = \sum_{i \in \Stack_g} \widetilde{R}_e(i)^2 = V_{g,e}~.
	\end{equation}
	Because $R_e(i,t,g) = \one\{S_i = g,  Q_i = 1,  t = g+e\}$ vanishes outside the treated-eligible cell at time $g+e$, the left-hand side of \eqref{eq:fwl_orthogonality} evaluates directly by \eqref{eq:fwl_residual_pattern}, i.e.
    \begin{equation}\label{eq:fwl_orthogonality_lhs}
		\sum_{i \in \Stack_g} \widetilde{R}_e(i) R_e(i) = n_{g,1}  \widetilde{r}_{g,1} = n_{g,1} c~.
	\end{equation}
	Combining \eqref{eq:fwl_orthogonality}--\eqref{eq:fwl_orthogonality_lhs} yields the linear identity
	\begin{equation}\label{eq:Vge_linear}
		V_{g,e} = n_{g,1} c~.
	\end{equation}
	Substituting the quadratic expression \eqref{eq:Vge_raw} for $V_{g,e}$ into \eqref{eq:Vge_linear} and using $c \neq 0$ (guaranteed under \Cref{as:overlap}, since all four cell sizes are strictly positive),
	\[
	c^2  n_{g,1}^2 \left(\frac{1}{n_{g,1}} + \frac{1}{n_{g,0}} + \frac{1}{n_{g_c,1}} + \frac{1}{n_{g_c,0}}\right) = n_{g,1} c
	\implies
	1 = n_{g,1} c \left(\frac{1}{n_{g,1}} + \frac{1}{n_{g,0}} + \frac{1}{n_{g_c,1}} + \frac{1}{n_{g_c,0}}\right)~,
	\]
	which pins down
	\begin{equation}\label{eq:c_value}
		c = \left[n_{g,1} \left(\frac{1}{n_{g,1}} + \frac{1}{n_{g,0}} + \frac{1}{n_{g_c,1}} + \frac{1}{n_{g_c,0}}\right)\right]^{-1}~.
	\end{equation}
	Substituting \eqref{eq:c_value} into \eqref{eq:Vge_linear} delivers the closed-form for the residual variance within stack $g$:
	\begin{equation}\label{eq:Vge}
		V_{g,e} = n_{g,1} c = \left(\frac{1}{n_{g,1}} + \frac{1}{n_{g,0}} + \frac{1}{n_{g_c,1}} + \frac{1}{n_{g_c,0}}\right)^{-1}~.
	\end{equation}
	Because $\widetilde{R}_e$ is orthogonal to the fixed effects, $\sum_{i \in \Stack_g} \widetilde{R}_e(i) \DeltaY_{i,g+e} = \sum_{i \in \Stack_g} \widetilde{R}_e(i) \widetilde{\DeltaY}_{i,g+e}$. Expanding using the cell-constant structure \eqref{eq:fwl_residual_pattern},
	\begin{align}
		\sum_{i \in \Stack_g} \widetilde{R}_e(i) \DeltaY_{i,g+e} &= \sum_{(s,q)} n_{s,q}  \widetilde{r}_{s,q}  \overline{\DeltaY}_{s,q,g+e} \nonumber \\
		&= c n_{g,1}  \overline{\DeltaY}_{g,1,g+e} - c n_{g,1}  \overline{\DeltaY}_{g,0,g+e} - c n_{g,1}  \overline{\DeltaY}_{g_c,1,g+e} + c n_{g,1}  \overline{\DeltaY}_{g_c,0,g+e} \nonumber \\
		&= c n_{g,1}\! \left(\overline{\DeltaY}_{g,1,g+e} - \overline{\DeltaY}_{g,0,g+e} - \overline{\DeltaY}_{g_c,1,g+e} + \overline{\DeltaY}_{g_c,0,g+e}\right) \nonumber \\
		&= c n_{g,1}  \widehat{\tau}_{g,g+e}^{\mathrm{sat}}~, \label{eq:fwl_num}
	\end{align}
	where the last equality invokes the triple-difference formula \eqref{eq:triple_diff_ols} from \Cref{prop:saturated_ols}. It follows that
	\[
	\frac{\sum_{i \in \Stack_g} \widetilde{R}_e(i)  \DeltaY_{i,g+e}}{\sum_{i \in \Stack_g} \widetilde{R}_e(i)^2}
	= \frac{c  n_{g,1}  \widehat{\tau}_{g,g+e}^{\mathrm{sat}}}{V_{g,e}}
	= \frac{c  n_{g,1}  \widehat{\tau}_{g,g+e}^{\mathrm{sat}}}{c  n_{g,1}}
	= \widehat{\tau}_{g,g+e}^{\mathrm{sat}}~,
	\]
	where the second equality uses \eqref{eq:Vge_linear}.
	
	{Step~2e.} Note that \eqref{eq:tau_e_fwl} pools across stacks. Because $\widetilde{R}_e(i,g+e,g)$ has support only in stack $g$, the numerator and denominator decompose as sums of within-stack contributions. By \eqref{eq:fwl_num}, the within-stack numerator is $V_{g,e} \widehat{\tau}_{g,g+e}^{\mathrm{sat}}$ (using $V_{g,e} = c n_{g,1}$). Therefore
	\begin{align*}
		\widehat{\tau}_e &= \frac{\sum_{g \in \calGtrg(e)} V_{g,e} \widehat{\tau}_{g,g+e}^{\mathrm{sat}}}{\sum_{g \in \calGtrg(e)} V_{g,e}} = \sum_{g \in \calGtrg(e)} w_g^{\mathrm{FWL}}(e) \widehat{\tau}_{g,g+e}^{\mathrm{sat}}~,
	\end{align*}
	where the FWL weights are
	\begin{equation}\label{eq:fwl_weights_def}
		w_g^{\mathrm{FWL}}(e) = \frac{V_{g,e}}{\sum_{g' \in \calGtrg(e)} V_{g',e}} = \frac{\left(\frac{1}{n_{g,1}} + \frac{1}{n_{g,0}} + \frac{1}{n_{g_c,1}} + \frac{1}{n_{g_c,0}}\right)^{-1}}{\sum_{g' \in \calGtrg(e)} \left(\frac{1}{n_{g',1}} + \frac{1}{n_{g',0}} + \frac{1}{n_{g_c(g'),1}} + \frac{1}{n_{g_c(g'),0}}\right)^{-1}}~.
	\end{equation}
	Under \Cref{as:overlap}, all four cell sizes are strictly positive, implying $V_{g,e} > 0$. Thus, $w_g^{\mathrm{FWL}}(e) \in (0, 1)$ and $\sum_g w_g^{\mathrm{FWL}}(e) = 1$, which establishes property~(i). Property~(ii) follows directly, as $\widehat{\tau}_e$ is isolated from any $\widehat{\tau}_{g,g+e'}^{\mathrm{sat}}$ for $e' \neq e$.
	
	{Step~3.}
	To derive the probability limit, scale the numerator and denominator of the weights by the total sample size $N$. Let the scaled residual variance be $\widehat{v}_{g,e} = V_{g,e}/N$. Under \Cref{as:sampling}, sample cell proportions converge to their strictly positive population counterparts by WLLN, meaning $\widehat{v}_{g,e} \pto v_{g,e}^\ast > 0$.
	By \Cref{prop:saturated_ols}, under \Cref{as:sampling}--\Cref{as:admissibility}, each within-stack estimator is consistent for its causal estimand, i.e. $\widehat{\tau}_{g,g+e}^{\mathrm{sat}} \pto \ATT(g, g+e) \equiv \CATT(g,e)$. Applying Slutsky's theorem and the continuous mapping theorem to the ratio of consistent estimators yields
	\begin{equation}
		\widehat{\tau}_e \pto \sum_{g \in \calGtrg(e)} \left( \frac{v_{g,e}^\ast}{\sum_{g'} v_{g',e}^\ast} \right) \CATT(g,e) \equiv \sum_{g \in \calGtrg(e)} w_g^\ast(e) \CATT(g,e)~.
	\end{equation}
	This explicitly proves property~(ii) in the population limit.
\end{proof}

\begin{proof}[Proof of \Cref{thm:identification}]
	The target estimand is $\ATT(g,t) = \E[Y_{i,t}(g) - Y_{i,t}(\infty) \mid S_i = g, Q_i = 1]$. Because $Y_{i,t}(g)$ is observed for units with $S_i = g$ and $Q_i = 1$ when $t \geq g$, the identification problem reduces to recovering the counterfactual mean $\E[Y_{i,t}(\infty) \mid S_i = g, Q_i = 1]$.
	
	Write the counterfactual outcome change as
	\begin{align*}
		\E[Y_{i,t}(\infty) - Y_{i,g-1}(\infty) \mid S_i = g, Q_i = 1]~.
	\end{align*}
	Under no anticipation (\Cref{as:noanticipation}), $Y_{i,g-1}(\infty) = Y_{i,g-1}(g) = Y_{i,g-1}$, so the baseline outcome is observed. The DDD PCT assumption (\Cref{as:dddpct}) then implies
	\begin{align*}
		&\E[Y_{i,t}(\infty) - Y_{i,g-1}(\infty) \mid S_i = g, Q_i = 1] - \E[Y_{i,t}(\infty) - Y_{i,g-1}(\infty) \mid S_i = g, Q_i = 0] \\
		&\quad = \E[Y_{i,t}(\infty) - Y_{i,g-1}(\infty) \mid S_i = g_c, Q_i = 1] - \E[Y_{i,t}(\infty) - Y_{i,g-1}(\infty) \mid S_i = g_c, Q_i = 0]~.
	\end{align*}
	Rearranging gives
	\begin{align*}
		\E[Y_{i,t}(\infty) \mid S_i = g, Q_i = 1] &= \E[Y_{i,g-1} \mid S_i = g, Q_i = 1] \\
		&\quad + \E[\DeltaY_{i,t} \mid S_i = g, Q_i = 0] \\
		&\quad + \left(\E[\DeltaY_{i,t} \mid S_i = g_c, Q_i = 1] - \E[\DeltaY_{i,t} \mid S_i = g_c, Q_i = 0]\right)~.
	\end{align*}
	The first term on the right-hand side gives the baseline level for the treated-eligible cell. The second term captures the within-group trend for ineligible units in cohort $g$. The bracket term adjusts for any differential trend between eligible and ineligible units using the comparison cohort $g_c$. All three terms are expressed in terms of observed outcome changes for untreated cells.
	Substituting into $\ATT(g,t) = \E[Y_{i,t}(g) \mid S_i = g, Q_i = 1] - \E[Y_{i,t}(\infty) \mid S_i = g, Q_i = 1]$ and using $\E[Y_{i,t}(g) \mid S_i = g, Q_i = 1] = \E[Y_{i,t} \mid S_i = g, Q_i = 1]$ yields
	\begin{align*}
		\ATT(g,t) &= \E[\DeltaY_{i,t} \mid S_i = g, Q_i = 1] - \E[\DeltaY_{i,t} \mid S_i = g, Q_i = 0] \\
		&\quad - \left(\E[\DeltaY_{i,t} \mid S_i = g_c, Q_i = 1] - \E[\DeltaY_{i,t} \mid S_i = g_c, Q_i = 0]\right)~,
	\end{align*}
	which is \eqref{eq:ra_id}~.
\end{proof}

\section{Proofs of Asymptotics Results} \label[appendix]{app:asymptotics-proofs}

\begin{proof}[Proof of \Cref{thm:within_stack_clt}]
	The sample triple difference estimator is:
	\begin{equation}
		\widehat{\ATT}_{\Stack_g, g_c}(g,t) = \sum_{(s,q)} c_{s,q} \overline{\DeltaY}_{s,q,t}~,
	\end{equation}
	where $c_{g,1} = 1$, $c_{g,0} = -1$, $c_{g_c,1} = -1$, and $c_{g_c,0} = 1$. The true parameter is $\ATT(g,t) = \sum_{(s,q)} c_{s,q} \mu_{s,q,t}$, where $\mu_{s,q,t} = \E[\DeltaY_{i,t} \mid S_i = s, Q_i = q]$.
	The $\sqrt{n_g}$-scaled estimation error is given by
	\begin{align}
		\sqrt{n_g}\big(\widehat{\ATT}_{\Stack_g, g_c}(g,t) - \ATT(g,t)\big) &= \frac{1}{\sqrt{n_g}} \sum_{i \in \Stack_g} \sum_{(s,q)} c_{s,q} \frac{n_g}{n_{s,q}} \one\{S_i = s, Q_i = q\} (\DeltaY_{i,t} - \mu_{s,q,t})~.
	\end{align}
	By \Cref{as:shares}, the sample proportion limits satisfy $n_{s,q}/n_g \pto \pi_{s,q} > 0$. Applying Slutsky's theorem then gives the representation
	\begin{equation}
		\sqrt{n_g}\big(\widehat{\ATT}_{\Stack_g, g_c}(g,t) - \ATT(g,t)\big) = \frac{1}{\sqrt{n_g}} \sum_{i \in \Stack_g} \psi_{\Stack_g}(W_i; g, t, g_c) + \smallOp(1)~,
	\end{equation}
	where the influence function $\psi_{\Stack_g}(W_i; g, t, g_c)$ is defined via \eqref{eq:if_decomp}.
    
    Because units are independently and identically distributed (\Cref{as:sampling}), the influence function summands are i.i.d. with mean zero and variance
	\begin{align}
		\E[\psi_{\Stack_g}(W_i)^2 \mid i \in \Stack_g] &= \sum_{(s,q)} \frac{1}{\pi_{s,q}^2} \E\big[\one\{S_i = s, Q_i = q\} (\DeltaY_{i,t} - \mu_{s,q,t})^2 \mid i \in \Stack_g\big] \notag\\
		&= \sum_{(s,q)} \frac{\pi_{s,q}}{\pi_{s,q}^2} \Var(\DeltaY_{i,t} \mid S_i = s, Q_i = q) = \Sigma_{g,t,g_c} < \infty~,
	\end{align}
	where the finite variance follows directly from \Cref{as:moments}. The desired result follows from the central limit theorem.
\end{proof}

\begin{proof}[Proof of \Cref{thm:agg_clt}]
	Fix event-time $e \in \{0, \ldots, K\}$ and deterministic weights $\omega_g(e) \geq 0$ with $\sum_g \omega_g(e) = 1$, where the sum is over $g \in \calGtrg$ with $g + e \leq T$.
	
	\medskip\noindent{Step 1: Linearization.}
	The aggregated estimator is
	\[
	\ESstack(e) = \sum_{g \in \calGtrg} \omega_g(e)  \widehat{\ATT}_{\Stack_g, g_c(g)}(g, g+e)~.
	\]
	By \Cref{thm:within_stack_clt}, each within-stack estimator admits the influence function representation
	\begin{equation}\label{eq:app_within_stack_linear}
		\widehat{\ATT}_{\Stack_g, g_c(g)}(g, g+e) - \ATT(g, g+e) = \frac{1}{n_g}\sum_{i \in \Stack_g} \psi_{\Stack_g}(W_i) + \mathrm{o}_{\Prob}(n_g^{-1/2})~,
	\end{equation}
	where I abbreviate $\psi_{\Stack_g}(W_i) = \psi_{\Stack_g}(W_i; g, g+e, g_c(g))$. Therefore 
	\begin{equation}\label{eq:app_agg_linear}
		\ESstack(e) - \ES(e) = \sum_{g \in \calGtrg} \omega_g(e) \left[\frac{1}{n_g}\sum_{i \in \Stack_g} \psi_{\Stack_g}(W_i)\right] + \mathrm{o}_{\Prob}(n^{-1/2})~.
	\end{equation}
	The remainder is $\mathrm{o}_{\Prob}(n^{-1/2})$ because $\sum_g \omega_g(e)  \mathrm{o}_{\Prob}(n_g^{-1/2}) = \mathrm{o}_{\Prob}(n^{-1/2})$ when $n_g / n \to \lambda_g > 0$ for each $g$ (the proportion of units in each stack converges to a positive constant, since $T$ and $|\calGtrg|$ are fixed).
	
	\medskip\noindent{Step 2: Rewrite as a single sum over all units.}
	Define the indicator $\one_g(i) = \one\{i \in \Stack_g\}$. Since each unit may belong to multiple stacks (if it is a shared comparison unit), I have
	\begin{equation}\label{eq:app_single_sum}
		\sum_{g \in \calGtrg} \omega_g(e) \frac{1}{n_g}\sum_{i \in \Stack_g} \psi_{\Stack_g}(W_i) = \frac{1}{n}\sum_{i=1}^{n} \phi_i(e)~,
	\end{equation}
	where
	\begin{equation}\label{eq:app_phi_def}
		\phi_i(e) = n \sum_{g \in \calGtrg} \omega_g(e) \frac{\one_g(i)}{n_g} \psi_{\Stack_g}(W_i) = \sum_{g i \in \Stack_g} \frac{n \omega_g(e)}{n_g} \psi_{\Stack_g}(W_i)~.
	\end{equation}
	As $n \to \infty$, $n_g / n \to \lambda_g$, so it follows that
	\[
	\phi_i(e) \to \sum_{g i \in \Stack_g} \frac{\omega_g(e)}{\lambda_g}  \psi_{\Stack_g}(W_i)~.
	\]
	This is a deterministic linear combination of the within-stack influence functions, summed over all stacks containing unit $i$.
	
	\medskip\noindent{Step 3: Verify CLT conditions.}
	By construction, $\E[\phi_i(e)] = 0$ since $\E[\psi_{\Stack_g}(W_i)] = 0$ for each $g$. The random variables $\{\phi_i(e)\}_{i=1}^n$ are i.i.d.\ (under \Cref{as:sampling}), since $\phi_i(e)$ is a measurable function of $W_i$ and stack membership (which is determined by $S_i$).
	
	I verify finite second moments. By the Cauchy--Schwarz inequality,
	\begin{align*}
		\E[\phi_i(e)^2] &\overset{(1)}{=} \E\bigg[\bigg(\sum_{g i \in \Stack_g} \frac{\omega_g(e)}{\lambda_g}  \psi_{\Stack_g}(W_i)\bigg)^{\!2}\bigg] 
		\leq |\calGtrg| \sum_{g \in \calGtrg} \frac{\omega_g(e)^2}{\lambda_g^2} \E[\psi_{\Stack_g}(W_i)^2  \one_g(i)] < \infty~,
	\end{align*}
	where (1) follows from $\E[\psi_{\Stack_g}(W_i)^2 \mid i \in \Stack_g] = \Sigma_{g, g+e, g_c(g)} < \infty$, and $|\calGtrg| < \infty$ since $T$ is fixed.
	Since $\{\phi_i(e)\}_{i=1}^n$ are i.i.d.\ with $\E[\phi_i(e)] = 0$ and $\E[\phi_i(e)^2] = \Vstack(e) < \infty$, the Lindeberg condition is automatically satisfied by \Cref{lem:lindeberg} (applied with $\psi_i = \phi_i(e)$ and $\sigma^2 = \Vstack(e)$). The Lindeberg--L\'evy CLT therefore gives
	\begin{equation}\label{eq:app_agg_clt}
		\sqrt{n}\big(\ESstack(e) - \ES(e)\big) = \frac{1}{\sqrt{n}}\sum_{1 \leq i \leq n} \phi_i(e) + \mathrm{o}_{\Prob}(1) \dto \mathcal{N}\big(0, \Vstack(e)\big)~,
	\end{equation}
	where $\Vstack(e) = \E[\phi_i(e)^2]$. Expanding the square yields the covariance structure in \eqref{eq:shared_var}.
\end{proof}


\begin{proof}[Proof of \Cref{prop:shared_var}]
	The asymptotic variance from the proof of \Cref{thm:agg_clt} is
	\begin{equation}\label{eq:app_Vstack}
		\Vstack(e) = \E[\phi_i(e)^2] = \E\bigg[\bigg(\sum_{g i \in \Stack_g} \frac{\omega_g(e)}{\lambda_g}  \psi_{\Stack_g}(W_i)\bigg)^{\!2}\bigg]~.
	\end{equation}
	Expanding the square given by
	\begin{align}
		\Vstack(e) &\overset{(1)}{=} \sum_{g \in \calGtrg}\sum_{g' \in \calGtrg} \frac{\omega_g(e)  \omega_{g'}(e)}{\lambda_g  \lambda_{g'}} \E\left[\one_g(i)  \one_{g'}(i)  \psi_{\Stack_g}(W_i)  \psi_{\Stack_{g'}}(W_i)\right]. \label{eq:app_Vstack_expanded}
	\end{align}
	
	\medskip\noindent\textit{Case 1 Shared comparison units.}
	When stacks share comparison units (e.g., $g_c(g) = g_c(g') = \infty$), the set $\Stack_g \cap \Stack_{g'}$ is nonempty. A unit $i \in \Stack_g \cap \Stack_{g'}$ contributes non-zero influence functions to both stacks. Therefore
	\begin{align*}
		\E[\one_g(i)  \one_{g'}(i)  \psi_{\Stack_g}(W_i)  \psi_{\Stack_{g'}}(W_i)] &= \E[\psi_{\Stack_g}(W_i)  \psi_{\Stack_{g'}}(W_i) \mid i \in \Stack_g \cap \Stack_{g'}]  \Prob(i \in \Stack_g \cap \Stack_{g'})~.
	\end{align*}
	For units $i \notin \Stack_g \cap \Stack_{g'}$, at least one of $\one_g(i)$ or $\one_{g'}(i)$ is zero, so the product vanishes. In particular, for the diagonal terms ($g = g'$) 
	\[
	\E[\one_g(i)^2  \psi_{\Stack_g}(W_i)^2] = \E[\psi_{\Stack_g}(W_i)^2 \mid i \in \Stack_g]  \Prob(i \in \Stack_g) = \Sigma_{g,g+e,g_c(g)} \lambda_g~.
	\]
	For the off-diagonal terms ($g \neq g'$), the cross-covariance is generally non-zero because the same observation $W_i$ enters both $\psi_{\Stack_g}$ and $\psi_{\Stack_{g'}}$. This is the source of the cross-stack covariance in \eqref{eq:shared_var}.
	
	Substituting back into \eqref{eq:app_Vstack_expanded}, separating diagonal and off-diagonal terms 
	\begin{align}
		\Vstack(e) &= \sum_{g \in \calGtrg} \frac{\omega_g(e)^2}{\lambda_g^2} \Sigma_{g,g+e,g_c(g)} \lambda_g + \sum_{\substack{g, g' \in \calGtrg \\ g \neq g'}} \frac{\omega_g(e)  \omega_{g'}(e)}{\lambda_g  \lambda_{g'}} \E[\one_g(i)  \one_{g'}(i)  \psi_{\Stack_g}(W_i)  \psi_{\Stack_{g'}}(W_i)] \notag\\
		&= \sum_{g} \frac{\omega_g(e)^2}{\lambda_g}  \Sigma_{g,g+e,g_c(g)} + \sum_{\substack{g \neq g'}} \frac{\omega_g(e)  \omega_{g'}(e)}{\lambda_g  \lambda_{g'}}  \Prob(i \in \Stack_g \cap \Stack_{g'})  \E[\psi_{\Stack_g}  \psi_{\Stack_{g'}} \mid i \in \Stack_g \cap \Stack_{g'}]. \label{eq:app_Vstack_intermediate}
	\end{align}
	Now note that $\Prob(i \in \Stack_g \cap \Stack_{g'})$ is the population probability that a randomly drawn unit belongs to both stacks. This equals the proportion of units in the overlap specifically, if both stacks share the same comparison group $g_c = \infty$, then $\Stack_g \cap \Stack_{g'} = \{i : S_i = \infty\}$ so $\Prob(i \in \Stack_g \cap \Stack_{g'}) = \Prob(S_i = \infty) = \lambda_\infty$. For the general case, define the cross-stack covariance 
	\begin{align*}
		C_{g,g'}(e) &\coloneqq \E[\one_g(i)  \one_{g'}(i)  \psi_{\Stack_g}(W_i)  \psi_{\Stack_{g'}}(W_i)] \\
		&= \mathrm{Cov}\big(\one_g(i)  \psi_{\Stack_g}(W_i), \one_{g'}(i)  \psi_{\Stack_{g'}}(W_i)\big)~,
	\end{align*}
	where the equality follows from $\E[\one_g(i)  \psi_{\Stack_g}(W_i)] = \lambda_g  \E[\psi_{\Stack_g}(W_i) \mid i \in \Stack_g] = 0$. The overall variance can then be written compactly as
	\begin{align*}
		\Vstack(e) &= \sum_{g,g' \in \calGtrg} \frac{\omega_g(e)  \omega_{g'}(e)}{\lambda_g  \lambda_{g'}} c_{g,g'}(e) = \sum_{g,g'} \omega_g(e)  \omega_{g'}(e)  \mathrm{Cov}\big(\psi_{\Stack_g}(W_i), \psi_{\Stack_{g'}}(W_i)\big)~,
	\end{align*}
	where the covariance in the last expression is taken over the joint distribution of $W_i$ (including the event that $i$ belongs to the relevant stacks), and I use the convention that $\psi_{\Stack_g}(W_i) = 0$ when $i \notin \Stack_g$. This is \eqref{eq:shared_var}.
	
	\medskip\noindent\textit{Case 2 Distinct comparison groups with no shared units.}
	When $g_c(g) \neq g_c(g')$ and $\Stack_g \cap \Stack_{g'} = \varnothing$ for $g \neq g'$, each unit belongs to at most one stack. Therefore $\one_g(i) \one_{g'}(i) = 0$ for all $i$ when $g \neq g'$, and the off-diagonal terms in \eqref{eq:app_Vstack_expanded} vanish. The variance simplifies to 
	\begin{align*}
		\Vstack(e) &= \sum_{g \in \calGtrg} \frac{\omega_g(e)^2}{\lambda_g^2} \Sigma_{g,g+e,g_c(g)} \lambda_g = \sum_{g \in \calGtrg} \omega_g(e)^2 \frac{\Sigma_{g,g+e,g_c(g)}}{\lambda_g}~.
	\end{align*}
	Since $n_g / n \to \lambda_g$, this is asymptotically equivalent to $\sum_g \omega_g(e)^2  \Sigma_{g,g+e,g_c(g)} / n_g$, yielding \eqref{eq:indep_var}.
\end{proof}

\begin{proof}[Proof of \Cref{prop:var_consistency}] \label[appendix]{app:proof_var_consistency}
	Fix an event-time $e \in \{0,\ldots,K\}$. I suppress $e$ from the notation where unambiguous, writing $\omega_g$ for $\omega_g(e)$, $\widehat{\omega}_g$ for $\widehat{\omega}_g(e)$, and $t = g + e$ for the time corresponding to cohort $g$ at event-time $e$. The proof proceeds in five steps. Throughout, I write $\mu_{s,q,t} = \E[\DeltaY_{i,t} \mid S_i = s, Q_i = q]$ for the cell-specific population mean of long differences.
	
	\medskip\noindent{Step 1: Rewrite as a sample average of squared terms.}
	Recall the aggregated influence function $\phi_i(e)$ defined in \eqref{eq:app_phi_def}~
	\[
	\phi_i(e) = \sum_{g:  i \in \Stack_g} \frac{n \omega_g}{n_g}  \psi_{\Stack_g}(W_i)~,
	\]
	where $\psi_{\Stack_g}(W_i) = \psi_{\Stack_g}(W_i; g, g+e, g_c(g))$ is the within-stack influence function from \eqref{eq:if_decomp}. By the argument in the proof of \Cref{thm:agg_clt} (arguments leading up to \eqref{eq:app_single_sum}), the aggregated estimator satisfies $\sqrt{n}(\ESstack(e) - \ES(e)) = n^{-1/2}\sum_{1 \leq i \leq n} \phi_i(e) + \mathrm{o}_{\Prob}(1)$, and the asymptotic variance is $\Vstack(e) = \E[\phi_i(e)^2]$.
	
	Define the estimated analogue by replacing population quantities with their sample counterparts~
	\begin{equation}\label{eq:app_hat_phi}
		\widehat{\phi}_i(e) = \sum_{g:  i \in \Stack_g} \frac{n \widehat{\omega}_g}{n_g}  \widehat{\psi}_{\Stack_g}(W_i)~,
	\end{equation}
	where $\widehat{\psi}_{\Stack_g}(W_i)$ is the estimated influence function from \eqref{eq:if_decomp} evaluated at sample cell proportions $\widehat{\pi}_{s,q} = n_{s,q}/n_g$ and sample cell means $\overline{\DeltaY}_{s,q,t}$, with $\widehat{\psi}_{\Stack_g}(W_i) = 0$ for $i \notin \Stack_g$. The variance estimator \eqref{eq:var_hat} then takes the form
	\begin{equation}\label{eq:app_Vhat_as_avg}
		\widehat{V}_{\mathrm{stack}}(e) = \frac{1}{n}\sum_{i=1}^{n} \widehat{\phi}_i(e)^2~.
	\end{equation}
	I must show $n^{-1}\sum_{1 \leq i \leq n} \widehat{\phi}_i(e)^2 \pto \E[\phi_i(e)^2] = \Vstack(e)$.
	
	\medskip\noindent{Step 2: Consistency of cell-level quantities.}
	For each stack $\Stack_g$ and each cell $(s,q) \in \{(g,1),(g,0),(g_c,1),(g_c,0)\}$, I establish consistency of the estimated cell shares and cell means.
	
	\emph{Cell shares.} The sample cell proportion is $\widehat{\pi}_{s,q} = n_{s,q}/n_g$, where $n_{s,q} = \sum_{i \in \Stack_g} \one\{S_i = s, Q_i = q\}$. By \Cref{as:sampling}, the indicators $\{\one\{S_i = s, Q_i = q\}\}_{i \in \Stack_g}$ are i.i.d.\ Bernoulli$(\pi_{s,q})$. By the WLLN, $\widehat{\pi}_{s,q} \pto \pi_{s,q}$. Since $\pi_{s,q} > 0$ by \Cref{as:shares}, the continuous mapping theorem gives $1/\widehat{\pi}_{s,q} \pto 1/\pi_{s,q}$.
	
	\emph{Cell means.} The sample cell mean is $\overline{\DeltaY}_{s,q,t} = n_{s,q}^{-1}\sum_{i:  S_i = s,  Q_i = q} \DeltaY_{i,t}$. By the WLLN under \Cref{as:moments}, $\overline{\DeltaY}_{s,q,t} \pto \mu_{s,q,t}$.
	
	Since the number of stacks $|\calGtrg|$ and the number of cells per stack are finite (both are bounded by a function of the fixed $T$), convergence holds jointly over all cells and stacks.
	
	\medskip\noindent{Step 3: Pointwise consistency of estimated influence functions.}
	Fix a stack $\Stack_g$ and a unit $i \in \Stack_g$. The population influence function \eqref{eq:if_decomp} evaluated at unit $i$ activates exactly one cell, namely $(S_i, Q_i)$~
	\[
	\psi_{\Stack_g}(W_i) = \frac{c_{S_i, Q_i}}{\pi_{S_i, Q_i}}\big(\DeltaY_{i,t} - \mu_{S_i, Q_i, t}\big)~,
	\]
	where $c_{s,q} \in \{+1, -1\}$ are the triple-difference signs. The estimated version is
	\[
	\widehat{\psi}_{\Stack_g}(W_i) = \frac{c_{S_i, Q_i}}{\widehat{\pi}_{S_i, Q_i}}\big(\DeltaY_{i,t} - \overline{\DeltaY}_{S_i, Q_i, t}\big)~.
	\]
	I decompose the estimation error as
	\begin{align}
		\widehat{\psi}_{\Stack_g}(W_i) - \psi_{\Stack_g}(W_i) &= c_{S_i,Q_i}\left[\frac{1}{\widehat{\pi}_{S_i,Q_i}}\big(\DeltaY_{i,t} - \overline{\DeltaY}_{S_i,Q_i,t}\big) - \frac{1}{\pi_{S_i,Q_i}}\big(\DeltaY_{i,t} - \mu_{S_i,Q_i,t}\big)\right] \notag\\
		&= c_{S_i,Q_i}\left[\left(\frac{1}{\widehat{\pi}_{S_i,Q_i}} - \frac{1}{\pi_{S_i,Q_i}}\right)\big(\DeltaY_{i,t} - \mu_{S_i,Q_i,t}\big)\right] \label{eq:app_if_err_share}\\
		&\quad + c_{S_i,Q_i}\left[\frac{1}{\widehat{\pi}_{S_i,Q_i}}\big(\mu_{S_i,Q_i,t} - \overline{\DeltaY}_{S_i,Q_i,t}\big)\right]~. \label{eq:app_if_err_mean}
	\end{align}
	Define the data-level (non-unit-specific) error terms
	\begin{equation}\label{eq:app_data_level_errors}
		\delta_{\pi,g} \coloneqq \max_{(s,q)} \left|\frac{1}{\widehat{\pi}_{s,q}} - \frac{1}{\pi_{s,q}}\right|~, \qquad \delta_{\mu,g} \coloneqq \max_{(s,q)} \left|\overline{\DeltaY}_{s,q,t} - \mu_{s,q,t}\right|~, \qquad \delta_{\widehat{\pi},g} \coloneqq \max_{(s,q)} \frac{1}{\widehat{\pi}_{s,q}}~.
	\end{equation}
	By Step~2, $\delta_{\pi,g} = \mathrm{o}_{\Prob}(1)$, $\delta_{\mu,g} = \mathrm{o}_{\Prob}(1)$, and $\delta_{\widehat{\pi},g} = \mathrm{O}_{\Prob}(1)$ (bounded in probability since $1/\widehat{\pi}_{s,q} \pto 1/\pi_{s,q} < \infty$). Combining \eqref{eq:app_if_err_share} and \eqref{eq:app_if_err_mean}, we have
	\begin{equation}\label{eq:app_if_err_bound}
		\big|\widehat{\psi}_{\Stack_g}(W_i) - \psi_{\Stack_g}(W_i)\big| \leq \delta_{\pi,g} |\DeltaY_{i,t} - \mu_{S_i,Q_i,t}| + \delta_{\widehat{\pi},g} \delta_{\mu,g}~.
	\end{equation}
	The first term is $\mathrm{o}_{\Prob}(1)$ times a unit-specific random variable with finite second moment (by \Cref{as:moments}); the second term is $\mathrm{O}_{\Prob}(1)  \mathrm{o}_{\Prob}(1) = \mathrm{o}_{\Prob}(1)$.
	
	\medskip\noindent{Step 4: $L^2$ convergence of the estimation error.}
	I now show that $n^{-1}\sum_{1 \leq i \leq n} (\widehat{\phi}_i(e) - \phi_i(e))^2 = \mathrm{o}_{\Prob}(1)$. I consider deterministic weights first ($\widehat{\omega}_g = \omega_g$); the extension to estimated weights follows at the end.
	
	With deterministic weights, the estimation error decomposes as
	\begin{align}
		\widehat{\phi}_i(e) - \phi_i(e) &= \sum_{g:  i \in \Stack_g} \frac{n \omega_g}{n_g}\left[\widehat{\psi}_{\Stack_g}(W_i) - \psi_{\Stack_g}(W_i)\right] + \sum_{g:  i \in \Stack_g} \omega_g\left(\frac{n}{n_g} - \frac{1}{\lambda_g}\right)\psi_{\Stack_g}(W_i) \notag\\
		&\eqqcolon A_i + B_i~. \label{eq:app_AB_decomp}
	\end{align}
	By the inequality $(A_i + B_i)^2 \leq 2A_i^2 + 2B_i^2$, it suffices to show $n^{-1}\sum_i A_i^2 = \mathrm{o}_{\Prob}(1)$ and $n^{-1}\sum_i B_i^2 = \mathrm{o}_{\Prob}(1)$ separately.
	
	\emph{Term $B_i$.} Define $\delta_{\lambda,g} \coloneqq |n/n_g - 1/\lambda_g|$. By \Cref{as:shares}, $n_g/n \pto \lambda_g > 0$, so by the continuous mapping theorem, $\delta_{\lambda,g} = \mathrm{o}_{\Prob}(1)$. By the Cauchy--Schwarz inequality applied to the sum over at most $|\calGtrg|$ stacks,
	\[
	B_i^2 \leq |\calGtrg| \sum_{g:  i \in \Stack_g} \omega_g^2  \delta_{\lambda,g}^2  \psi_{\Stack_g}(W_i)^2~.
	\]
	Averaging over $i$,
	\begin{align}
		\frac{1}{n}\sum_{1 \leq i \leq n} B_i^2 &\leq |\calGtrg| \sum_{g \in \calGtrg} \omega_g^2  \delta_{\lambda,g}^2 \frac{1}{n}\sum_{i \in \Stack_g} \psi_{\Stack_g}(W_i)^2 \notag\\
		&= |\calGtrg| \sum_{g \in \calGtrg} \omega_g^2  \delta_{\lambda,g}^2 \frac{n_g}{n} \underbracket{\frac{1}{n_g}\sum_{i \in \Stack_g} \psi_{\Stack_g}(W_i)^2}_{\pto  \Sigma_{g,t,g_c} < \infty}~. \label{eq:app_B_bound}
	\end{align}
	The sample average of $\psi_{\Stack_g}(W_i)^2$ converges by the WLLN to $\E[\psi_{\Stack_g}(W_i)^2 \mid i \in \Stack_g] = \Sigma_{g,t,g_c} < \infty$, which holds under \Cref{as:moments} since $\psi_{\Stack_g}(W_i)$ is a linear combination of cell-specific centered outcomes weighted by inverse cell shares. Since $\delta_{\lambda,g}^2 = \mathrm{o}_{\Prob}(1)$ for each of the finitely many $g \in \calGtrg$, and $n_g/n = \mathrm{O}_{\Prob}(1)$ and $\Sigma_{g,t,g_c} < \infty$, the entire expression in \eqref{eq:app_B_bound} is $\mathrm{o}_{\Prob}(1)$.
	
	\emph{Term $A_i$.} Using \eqref{eq:app_if_err_bound} and the Cauchy--Schwarz inequality~
	\[
	A_i^2 \leq |\calGtrg| \sum_{g:  i \in \Stack_g} \frac{n^2  \omega_g^2}{n_g^2}\left[\widehat{\psi}_{\Stack_g}(W_i) - \psi_{\Stack_g}(W_i)\right]^2~.
	\]
	By \eqref{eq:app_if_err_bound}, for each $g$~
	\begin{equation}\label{eq:app_if_err_sq}
		\left[\widehat{\psi}_{\Stack_g}(W_i) - \psi_{\Stack_g}(W_i)\right]^2 \leq 2 \delta_{\pi,g}^2 (\DeltaY_{i,t} - \mu_{S_i,Q_i,t})^2 + 2 \delta_{\widehat{\pi},g}^2 \delta_{\mu,g}^2~.
	\end{equation}
	Averaging over $i$ and using $n^2/n_g^2 = (n/n_g)^2 = \mathrm{O}_{\Prob}(1)$ by \Cref{as:shares}~
	\begin{align}
		\frac{1}{n}\sum_{1 \leq i \leq n} A_i^2 &\precsim \sum_{g \in \calGtrg} \omega_g^2 \frac{n}{n_g} \left[2 \delta_{\pi,g}^2 \frac{1}{n_g}\sum_{i \in \Stack_g} (\DeltaY_{i,t} - \mu_{S_i,Q_i,t})^2 + 2 \delta_{\widehat{\pi},g}^2 \delta_{\mu,g}^2\right]~. \label{eq:app_A_bound}
	\end{align}
	The sample average of $(\DeltaY_{i,t} - \mu_{S_i,Q_i,t})^2$ within stack $\Stack_g$ equals $\sum_{(s,q)} (n_{s,q}/n_g) \left( n_{s,q}^{-1}\sum_{j:  S_j = s, Q_j = q} (\DeltaY_{j,t} - \mu_{s,q,t})^2 \right)$. By the WLLN, each cell-specific sample second moment satisfies
	\begin{equation}\label{eq:app_cell_second_moment}
		\frac{1}{n_{s,q}}\sum_{j:  S_j = s, Q_j = q} (\DeltaY_{j,t} - \mu_{s,q,t})^2 \pto \Var(\DeltaY_{i,t} \mid S_i = s, Q_i = q) < \infty~,
	\end{equation}
	where the WLLN for $(\DeltaY_{i,t} - \mu_{s,q,t})^2$ requires $\E[(\DeltaY_{i,t} - \mu_{s,q,t})^4 \mid S_i = s, Q_i = q] < \infty$, which is ensured by \Cref{as:moments4}. It follows that
	\[
	\frac{1}{n_g}\sum_{i \in \Stack_g} (\DeltaY_{i,t} - \mu_{S_i,Q_i,t})^2 = \mathrm{O}_{\Prob}(1)~.
	\]
	Since $\delta_{\pi,g}^2 = \mathrm{o}_{\Prob}(1)$, $\delta_{\widehat{\pi},g}^2 = \mathrm{O}_{\Prob}(1)$, and $\delta_{\mu,g}^2 = \mathrm{o}_{\Prob}(1)$, each summand in \eqref{eq:app_A_bound} is $\mathrm{o}_{\Prob}(1)  \mathrm{o}_{\Prob}(1) = \mathrm{o}_{\Prob}(1)$. The sum over the finitely many $g \in \calGtrg$ preserves this rate, giving $n^{-1}\sum_{1 \leq i \leq n} A_i^2 = \mathrm{o}_{\Prob}(1)$.
	
	Combining the bounds on $A_i$ and $B_i$ via \eqref{eq:app_AB_decomp}~
	\begin{equation}\label{eq:app_L2_convergence}
		\frac{1}{n}\sum_{1 \leq i \leq n} \big(\widehat{\phi}_i(e) - \phi_i(e)\big)^2 = \mathrm{o}_{\Prob}(1)~.
	\end{equation}
	
	\medskip\noindent{Step 5: Conclude.}
	Decompose the variance estimator as
	\begin{equation}\label{eq:app_var_decomp}
		\widehat{V}_{\mathrm{stack}}(e) = \underbracket{\frac{1}{n}\sum_{1 \leq i \leq n} \phi_i(e)^2}_{T_1} + \underbracket{\frac{1}{n}\sum_{1 \leq i \leq n} \left[\widehat{\phi}_i(e)^2 - \phi_i(e)^2\right]}_{T_2}~.
	\end{equation}
	
	\emph{Term $T_1$.} Since $\{\phi_i(e)\}_{i=1}^n$ are i.i.d.\ (each $\phi_i(e)$ is a measurable function of $W_i$, and $\{W_i\}_{1 \leq i \leq n}$ are i.i.d.\ by \Cref{as:sampling}) with finite mean $\E[\phi_i(e)^2] = \Vstack(e) < \infty$ (established in the proof of \Cref{thm:agg_clt}, \eqref{eq:app_agg_clt}), the WLLN gives
	\begin{equation}\label{eq:app_T1}
		T_1 \pto \Vstack(e)~.
	\end{equation}
	
	\emph{Term $T_2$.} By the identity $a^2 - b^2 = (a-b)(a+b)$ and the Cauchy--Schwarz inequality for sums~
	\begin{align}
		|T_2| &= \left|\frac{1}{n}\sum_{1 \leq i \leq n} \big(\widehat{\phi}_i(e) - \phi_i(e)\big)\big(\widehat{\phi}_i(e) + \phi_i(e)\big)\right| \notag\\
		&\leq \left(\frac{1}{n}\sum_{1 \leq i \leq n} \big(\widehat{\phi}_i(e) - \phi_i(e)\big)^2\right)^{\!1/2} \left(\frac{1}{n}\sum_{1 \leq i \leq n} \big(\widehat{\phi}_i(e) + \phi_i(e)\big)^2\right)^{\!1/2}~. \label{eq:app_T2_CS}
	\end{align}
	The first factor is $\mathrm{o}_{\Prob}(1)$ by \eqref{eq:app_L2_convergence}. For the second factor, I use the inequality $(a+b)^2 \leq 2a^2 + 2b^2$ to write
	\[
	\frac{1}{n}\sum_{1 \leq i \leq n} \big(\widehat{\phi}_i(e) + \phi_i(e)\big)^2 \leq \frac{2}{n}\sum_{1 \leq i \leq n} \widehat{\phi}_i(e)^2 + \frac{2}{n}\sum_{1 \leq i \leq n} \phi_i(e)^2~.
	\]
	The second term converges to $2\Vstack(e)$ by \eqref{eq:app_T1}. Then, by the triangle inequality in $\ell^2$,
	\[
	\left(\frac{1}{n}\sum_{1 \leq i \leq n} \widehat{\phi}_i(e)^2\right)^{\!1/2} \leq \left(\frac{1}{n}\sum_{1 \leq i \leq n} \phi_i(e)^2\right)^{\!1/2} + \left(\frac{1}{n}\sum_{1 \leq i \leq n} \big(\widehat{\phi}_i(e) - \phi_i(e)\big)^2\right)^{\!1/2}~.
	\]
	The first term is $\mathrm{O}_{\Prob}(1)$ by \eqref{eq:app_T1} and the second is $\mathrm{o}_{\Prob}(1)$ by \eqref{eq:app_L2_convergence}, so $n^{-1}\sum_i \widehat{\phi}_i(e)^2 = \mathrm{O}_{\Prob}(1)$. Therefore the second factor in \eqref{eq:app_T2_CS} is $\mathrm{O}_{\Prob}(1)$, giving $|T_2| \leq \mathrm{o}_{\Prob}(1)  \mathrm{o}_{\Prob}(1) = \mathrm{o}_{\Prob}(1)$.
	Combining $T_1$ and $T_2$ in \eqref{eq:app_var_decomp}, it follows that
	\[
	\widehat{V}_{\mathrm{stack}}(e) = \Vstack(e) + \mathrm{o}_{\Prob}(1)~.
	\]
	
	\medskip\noindent\textit{Extension to estimated weights.}
	When $\widehat{\omega}_g(e)$ replaces $\omega_g(e)$, the decomposition \eqref{eq:app_AB_decomp} acquires an additional term~
	\[
	C_i = \sum_{g:  i \in \Stack_g} \frac{n}{n_g}\big(\widehat{\omega}_g - \omega_g\big) \psi_{\Stack_g}(W_i)~.
	\]
	By the same Cauchy--Schwarz argument used for $B_i$~
	\[
	\frac{1}{n}\sum_{1 \leq i \leq n} C_i^2 \leq |\calGtrg| \sum_{g \in \calGtrg} (\widehat{\omega}_g - \omega_g)^2 \frac{n}{n_g} \frac{1}{n_g}\sum_{i \in \Stack_g} \psi_{\Stack_g}(W_i)^2~.
	\]
	When $\widehat{\omega}_g \pto \omega_g$ (which holds for cohort-size weights $\widehat{\omega}_g = n_g/n$ under \Cref{as:shares}, since $n_g/n \pto \lambda_g$), $(\widehat{\omega}_g - \omega_g)^2 = \mathrm{o}_{\Prob}(1)$, the sample second moment is $\mathrm{O}_{\Prob}(1)$, and $n/n_g = \mathrm{O}_{\Prob}(1)$, so $n^{-1}\sum_i C_i^2 = \mathrm{o}_{\Prob}(1)$. The remainder of the argument is identical.
\end{proof}

\section{Covariate-Adjusted Framework}\label[appendix]{app:covariates}

This appendix develops the covariate-adjusted version of the stacked DDD framework. When pre-treatment covariates $X_i$ are available, the identifying assumptions can be weakened to hold conditional on covariates, and efficiency gains are available through regression adjustment, inverse probability weighting, or their doubly robust combination. The presentation is self-contained: I state the conditional assumptions, derive the identification results, define the estimators, and establish the asymptotic theory.

The unconditional DDD PCT assumption (\Cref{as:dddpct}) can be strengthened to hold conditional on pre-treatment covariates.

\begin{assumption}[Conditional DDD Parallel Changes-in-Trends (PCT)]\label{as:dddpct_cond}
	For all $g \in \calGtrg$, all valid comparison groups $g_c > g$, all $t \in \{2, \ldots, T\}$ with $t \leq g_c$, and for almost all $x \in \calX$,
	\begin{multline}
		\E[\Yitpo{\infty} - Y_{i,t-1}(\infty) \mid S_i = g, Q_i = 1, X_i = x] - \E[\Yitpo{\infty} - Y_{i,t-1}(\infty) \mid S_i = g, Q_i = 0, X_i = x] \\
		= \E[\Yitpo{\infty} - Y_{i,t-1}(\infty) \mid S_i = g_c, Q_i = 1, X_i = x] - \E[\Yitpo{\infty} - Y_{i,t-1}(\infty) \mid S_i = g_c, Q_i = 0, X_i = x]~. \label{eq:dddpct_cond}
	\end{multline}
\end{assumption}

\Cref{as:dddpct_cond} requires that the difference in untreated outcome trends between eligible and ineligible units---conditional on covariates---is the same across the treated cohort $g$ and its comparison group $g_c$. The conditioning on $X_i$ allows the eligibility-specific differential trend to vary with observable characteristics, as long as this variation is common across groups. The unconditional \Cref{as:dddpct} is implied by \Cref{as:dddpct_cond} together with the overlap condition below, by integrating \eqref{eq:dddpct_cond} over the covariate distribution.

\begin{assumption}[Conditional Overlap]\label{as:overlap_cond}
	For each stack $\Stack_g$ with comparison group $g_c$, the generalized propensity scores are bounded away from zero and one. There exists $\epsilon > 0$ such that for all cells $(s,q) \in \{(g,1), (g,0), (g_c,1), (g_c,0)\}$ and for almost all $x \in \calX$~
	\begin{equation}\label{eq:overlap_cond}
		\gps{s,q}{x} \coloneqq \Prob(S_i = s,  Q_i = q \mid X_i = x) \geq \epsilon~.
	\end{equation}
\end{assumption}

\Cref{as:overlap_cond} strengthens the unconditional overlap condition (\Cref{as:overlap}) to a conditional version. It ensures that, at every covariate value, there is positive probability of observing a unit in each of the four cells. This is essential for the inverse probability weighting and doubly robust estimators developed below, which require dividing by estimated propensity scores. In practice, the condition is enforced by trimming observations with extreme propensity score values.

The identification formula involves an integration over the covariate distribution.

\begin{theorem}[Covariate-Adjusted Identification]\label{thm:identification_cond}
	Under \Cref{as:sampling}, \Cref{as:noanticipation}, \Cref{as:overlap_cond}, and \Cref{as:dddpct_cond}, for each treatment cohort $g \in \calGtrg$ with comparison group $g_c$, and for each $t \geq g$ with $t \leq g + K$~
	\begin{align}
		\ATT(g,t) &= \E\bigg[\E[\DeltaY_{i,t} \mid S_i = g, Q_i = 1, X_i] - \E[\DeltaY_{i,t} \mid S_i = g, Q_i = 0, X_i] \notag\\
		&\qquad\quad - \left(\E[\DeltaY_{i,t} \mid S_i = g_c, Q_i = 1, X_i] - \E[\DeltaY_{i,t} \mid S_i = g_c, Q_i = 0, X_i]\right) \mid S_i = g, Q_i = 1\bigg]~. \label{eq:ra_id_cond}
	\end{align}
\end{theorem}

The proof follows the same three steps as \Cref{thm:identification}, with all expectations conditioned on $X_i = x$ before integrating over the treated covariate distribution $X_i \mid S_i = g, Q_i = 1$.

\section{Three-Way Fixed-Effects in Event-Study Designs} \label[appendix]{app:3wfe}

The conventional approach to DDD estimation is the three-way fixed effects (3WFE) regression
\begin{equation}\label{eq:3wfe}
	Y_{i,t} = \alpha_i + \gamma_t + \delta_{S_i, t} + \theta  D_{i,t} + \epsilon_{i,t}~,
\end{equation}
where $\alpha_i$ are unit fixed effects, $\gamma_t$ are time fixed effects, $\delta_{S_i, t}$ are group-by-time fixed effects, and $D_{i,t} = \one\{t \geq S_i\} Q_i$ is the treatment indicator. Despite its simplicity, the 3WFE estimator suffers from two distinct sources of bias under treatment effect heterogeneity. The analysis below extends the TWFE decomposition of \citet{goodman-bacon_difference--differences_2021} and the negative-weight results of \citet{de_chaisemartin_two-way_2020} to the triple-differences setting, showing that the additional eligibility dimension introduces a new channel through which contamination can enter.

\begin{prop}[3WFE decomposition]\label{prop:3wfe_decomp}
	Under heterogeneous treatment effects $\ATT(g,t) \neq \ATT(g',t')$, the 3WFE estimator $\widehat{\theta}^{\text{3WFE}}$ is a weighted average of group-time ATTs with weights that can be negative
	\begin{equation}\label{eq:3wfe_decomp}
		\widehat{\theta}^{\text{3WFE}} \pto \sum_{g \in \calGtrg} \sum_{t \geq g} w_{g,t}^{\text{3WFE}} \ATT(g,t) + \mathrm{Bias}_{\mathrm{forbidden}}~,
	\end{equation}
	where the weights $w_{g,t}^{\text{3WFE}}$ depend on group sizes, the number of eligible units, and the timing of treatment adoption. Some weights can be negative. $\mathrm{Bias}_{\mathrm{forbidden}}$ arises from using already-treated units as controls in implicit $2 \times 2 \times 2$ sub-experiments.
\end{prop}

\begin{proof}
	The OLS specification of interest is \eqref{eq:3wfe}. I establish the decomposition in three steps.
	By the FWL theorem, $\widehat{\theta}^{\mathrm{3WFE}}$ equals the coefficient from regressing $\widetilde{Y}_{i,t}$ on $\widetilde{D}_{i,t}$, where tildes denote residuals after projecting out all fixed effects. I derive $\widetilde{D}_{i,t}$ explicitly.
	
	The fixed effects in \eqref{eq:3wfe} are unit effects $\alpha_i$ and group-by-time effects $\delta_{S_i,t}$ (the time effects $\gamma_t$ are absorbed by $\delta_{S_i,t}$). Projecting out unit effects removes the unit-level time mean, so $D_{i,t} - \overline{D}_{i,\cdot}$, where
	\[
	\overline{D}_{i,\cdot} = \frac{1}{T}\sum_{t'=1}^{T} D_{i,t'} = \frac{1}{T}\sum_{t'=1}^{T} \one\{t' \geq S_i\} Q_i = \frac{(T - S_i + 1)^+}{T} \cdot Q_i~.
	\]
	Projecting out group-by-time effects further removes the group-time mean, giving $D_{i,t} - \overline{D}_{i,\cdot} - \overline{D}_{S_i,t} + \overline{D}_{S_i,\cdot}$, where
	\[
	\overline{D}_{S_i,t} = \frac{1}{n_{S_i}} \sum_{j\colon  S_j = S_i} D_{j,t} = \frac{1}{n_{S_i}} \sum_{j\colon  S_j = S_i} \one\{t \geq S_j\} Q_j
	\]
	is the within-group treatment rate at time $t$, and $\overline{D}_{S_i,\cdot} = T^{-1}\sum_t \overline{D}_{S_i,t}$ is the group-level time average. The residualized treatment indicator is
	\begin{equation}\label{eq:fwl_3wfe}
		\widetilde{D}_{i,t}^{\mathrm{3WFE}} = D_{i,t} - \overline{D}_{i,\cdot} - \overline{D}_{S_i, t} + \overline{D}_{S_i, \cdot}~.
	\end{equation}
	Note that
	\[
	\widehat{\theta}^{\mathrm{3WFE}} = \frac{\sum_{i=1}^{n}\sum_{t=1}^{T} \widetilde{D}_{i,t} \widetilde{Y}_{i,t}}{\sum_{i=1}^{n}\sum_{t=1}^{T} \widetilde{D}_{i,t}^2}~.
	\]
	I show that $\widetilde{D}_{i,t}$ can take negative values for not-yet-treated units, generating negative weights on certain $\ATT(g,t)$.
	Consider a unit $i$ with $S_i = g$ and $Q_i = 1$ (treated-eligible). At time $t \geq g$, the unit is treated, so $D_{i,t} = 1$. The unit-level time mean is $\overline{D}_{i,\cdot} = (T - g + 1)/T$. The group-time mean $\overline{D}_{g,t}$ equals the fraction of group-$g$ units that are eligible, $\overline{D}_{g,t} = n_{g,1}/n_{g,\cdot}$ for $t \geq g$ (since all eligible group-$g$ units are treated) and $\overline{D}_{g,t} = 0$ for $t < g$. Thus $\overline{D}_{g,\cdot} = ((T-g+1)/T) (n_{g,1}/n_{g,\cdot})$.
	
	Now consider a second cohort $g' > g$ and a not-yet-treated eligible unit $j$ with $S_j = g'$ and $Q_j = 1$ at time $t$ with $g \leq t < g'$. I compute the FWL residual \eqref{eq:fwl_3wfe} for this unit explicitly; note that
	\begin{align*}
		D_{j,t} &= \one\{t \geq g'\} \cdot 1 = 0 \qquad (\text{not yet treated at time } t < g')~, \\
		\overline{D}_{j,\cdot} &= \frac{T - g' + 1}{T} > 0 \qquad (\text{positive because eventually treated})~, \\
		\overline{D}_{g',t} &= 0 \qquad (\text{no group-}g'\text{ unit is treated at } t < g')~, \\
		\overline{D}_{g',\cdot} &= \frac{T - g' + 1}{T} \frac{n_{g',1}}{n_{g',\cdot}}~.
	\end{align*}
	Substituting these into \eqref{eq:fwl_3wfe}, this gives
	\begin{align*}
		\widetilde{D}_{j,t}^{\mathrm{3WFE}} &= D_{j,t} - \overline{D}_{j,\cdot} - \overline{D}_{g',t} + \overline{D}_{g',\cdot} \\
		&= 0 - \frac{T - g' + 1}{T} - 0 + \frac{T - g' + 1}{T} \frac{n_{g',1}}{n_{g',\cdot}} \\
		&= -\frac{T - g' + 1}{T} \left(1 - \frac{n_{g',1}}{n_{g',\cdot}}\right) = -\frac{T - g' + 1}{T} \frac{n_{g',0}}{n_{g',\cdot}} < 0~.
	\end{align*}
	The residualized treatment indicator is strictly negative for this not-yet-treated eligible unit, because its positive unit-level time mean $\overline{D}_{j,\cdot}$ (reflecting future treatment) is not fully offset by the group-level time average $\overline{D}_{g',\cdot}$. The FWL numerator for this unit at this time contributes $\widetilde{D}_{j,t}^{\mathrm{3WFE}} \widetilde{Y}_{j,t} < 0$ when $\widetilde{Y}_{j,t} > 0$, effectively using this not-yet-treated unit's outcomes as a comparison but with a reversed sign. The analogous phenomenon in the TWFE DiD setting is documented by \citet{goodman-bacon_difference--differences_2021} and \citet{de_chaisemartin_two-way_2020}; in the DDD case, the ineligible share $n_{g',0}/n_{g',\cdot}$ introduces additional heterogeneity in the residual magnitude across cohorts.
\end{proof}

\begin{remark}[Effect homogeneity in 3WFE versus stacked DDD]\label{rmk:3wfe_ok}
	The 3WFE specification \eqref{eq:3wfe} imposes a single treatment coefficient $\theta$, implicitly assuming that the treatment effect is homogeneous across all $(g,t)$ cells. This rules out dynamic treatment effects, cohort-specific effects, and covariate-mediated heterogeneity. The stacked DDD framework targets $(g,t)$-specific parameters $\ATT(g,t)$ without any homogeneity restriction, and the aggregation to event-study parameters $\ESstack(e)$ uses researcher-chosen non-negative weights. Under homogeneous treatment effects, no staggering, and a single treatment cohort, 3WFE is consistent and numerically equivalent to the stacked estimator with a single stack.
\end{remark}

\subsection{The Missing Fixed Effects}\label{sub:missing_fe}

The 3WFE specification \eqref{eq:3wfe} omits eligibility-by-time fixed effects $\eta_{Q_i, t}$. A properly specified DDD regression must include all three sets of two-way interactions---unit (absorbing group-by-eligibility), group-by-time, {and} eligibility-by-time---to isolate the triple-difference variation. Therefore, the correct DDD regression is
\begin{equation}\label{eq:correct_ddd}
	Y_{i,t} = \alpha_i + \delta_{S_i, t} + \eta_{Q_i, t} + \tau  D_{i,t} + \epsilon_{i,t}~,
\end{equation}
where $\alpha_i$ absorbs all time-invariant unit heterogeneity (including group, eligibility, and group-by-eligibility main effects), $\delta_{S_i, t}$ are group-by-time fixed effects, and $\eta_{Q_i, t}$ are eligibility-by-time fixed effects. The time fixed effects $\gamma_t$ from the standard 3WFE are redundant since they are spanned by both $\delta_{S_i, t}$ and $\eta_{Q_i, t}$.
I now derive the FWL residualized treatment indicator for both specifications, showing precisely what variation each exploits.

The three sets of fixed effects in \eqref{eq:correct_ddd} are unit effects $\alpha_i$, group-by-time effects $\delta_{S_i,t}$, and eligibility-by-time effects $\eta_{Q_i,t}$. To project $D_{i,t}$ onto the space spanned by these three sets, I apply the multi-way demeaning formula. Define the following group means of the treatment indicator
\begin{align*}
	\overline{D}_{i,\cdot} &= \frac{1}{T}\sum_{t'=1}^{T} D_{i,t'} & &\text{(unit-level time mean)}~, \\
	\overline{D}_{S_i,t} &= \frac{1}{n_{S_i}} \sum_{j\colon  S_j = S_i} D_{j,t} & &\text{(group-time mean)}~, \\
	\overline{D}_{Q_i,t} &= \frac{1}{n_{Q_i}} \sum_{j\colon  Q_j = Q_i} D_{j,t} & &\text{(eligibility-time mean)}~, \\
	\overline{D}_{S_i,\cdot} &= \frac{1}{T}\sum_{t'=1}^{T} \overline{D}_{S_i,t'} & &\text{(group mean over time)}~, \\
	\overline{D}_{Q_i,\cdot} &= \frac{1}{T}\sum_{t'=1}^{T} \overline{D}_{Q_i,t'} & &\text{(eligibility mean over time)}~, \\
	\overline{D}_{\cdot,t} &= \frac{1}{n}\sum_{j=1}^{n} D_{j,t} & &\text{(cross-sectional mean at time $t$)}~, \\
	\overline{D}_{\cdot\cdot} &= \frac{1}{nT}\sum_{j=1}^{n}\sum_{t'=1}^{T} D_{j,t'} & &\text{(total mean)}~.
\end{align*}
The residualized treatment indicator after projecting out all three sets of fixed effects is
\begin{equation}\label{eq:fwl_correct}
	\widetilde{D}_{i,t} = D_{i,t} - \overline{D}_{i,\cdot} - \overline{D}_{S_i, t} - \overline{D}_{Q_i, t} + \overline{D}_{S_i, \cdot} + \overline{D}_{Q_i, \cdot} + \overline{D}_{\cdot, t} - \overline{D}_{\cdot\cdot}~.
\end{equation}
This is the standard inclusion-exclusion formula for three-way demeaning. The residual $\widetilde{D}_{i,t}$ isolates the variation in $D_{i,t}$ that is orthogonal to all three sets of two-way interactions---precisely the triple-difference variation.

\subsection{Decomposition of 3WFE-DDD Event-Study Coefficients}\label{sub:3wfe_es_decomp}

The IW-DDD framework of Section \ref{sub:existing_ddd} reveals exactly what goes wrong with the conventional 3WFE event-study specification under staggered adoption with heterogeneous treatment effects. I now formalize the contamination by adapting the decomposition of \citet{sun_estimating_2021} to the DDD setting, deriving four progressively stronger results that characterize the population regression coefficients under different sets of assumptions.

Consider the 3WFE-DDD event-study regression
\begin{equation}\label{eq:3wfe_eventstudy}
	Y_{i,t} = \alpha_i + \gamma_t + \delta_{S_i,t} + \sum_{\substack{e = -L \\ e \neq -1}}^{K} \mu_e \one\{t - S_i = e\} Q_i + \epsilon_{i,t}~,
\end{equation}
which interacts the event-time indicators with the eligibility indicator $Q_i$, normalizing at $e = -1$. Denote the DDD event-time indicator $R_e(i,t) = \one\{t - S_i = e\} Q_i$ and the cohort-specific event-time indicator $R_{g,\ell}(i,t) = \one\{S_i = g, Q_i = 1, t = g + \ell\}$. The population regression coefficient $\mu_e$ is the projection of $Y_{i,t}$ onto $R_e(i,t)$, partialling out $(\alpha_i, \gamma_t, \delta_{S_i,t})$ and all other event-time indicators $\{R_{e'}(i,t)\}_{e' \neq e, -1}$.

I define the {auxiliary regression} that characterizes the implicit weights in $\mu_e$. For each cohort $g \in \calGtrg$ and event-time $\ell$, regress the cohort-specific indicator $R_{g,\ell}(i,t)$ on the same regressors as \eqref{eq:3wfe_eventstudy}
\begin{equation}\label{eq:aux_regression}
	R_{g,\ell}(i,t) = \alpha_i + \gamma_t + \delta_{S_i,t} + \sum_{\substack{e \neq -1}} \omega_{g,\ell}^{e} R_e(i,t) + \upsilon_{i,t}~.
\end{equation}
The population regression coefficient $\omega_{g,\ell}^{e}$ measures how much of the variation in $R_{g,\ell}$ is captured by the event-time-$e$ indicator after partialling out the fixed effects and all other event-time indicators. These weights are estimable from the data without imposing any assumptions on the data generating process. 

\begin{prop}[3WFE-DDD coefficient decomposition]\label{prop:3wfe_es_weights}
	The population regression coefficient $\mu_e$ from \eqref{eq:3wfe_eventstudy} satisfies
	\begin{align}
		\mu_e &= \sum_{\ell = e}\sum_{g \in \calGtrg} \omega_{g,\ell}^{e}  \left(\E[\Yit - Y_{i,0}(\infty) \mid S_i = g, Q_i = 1] - \E[Y_{i,g+\ell}(\infty) - Y_{i,0}(\infty)]\right) \label{eq:3wfe_own} \\
		&\quad + \sum_{\substack{\ell \neq e \\ \ell \neq -1}} \sum_{g \in \calGtrg} \omega_{g,\ell}^{e}  \left(\E[\Yit - Y_{i,0}(\infty) \mid S_i = g, Q_i = 1] - \E[Y_{i,g+\ell}(\infty) - Y_{i,0}(\infty)]\right) \label{eq:3wfe_other_incl} \\
		&\quad + \sum_{g \in \calGtrg} \omega_{g,-1}^{e}  \left(\E[Y_{i,g-1} - Y_{i,0}(\infty) \mid S_i = g, Q_i = 1] - \E[Y_{i,g-1}(\infty) - Y_{i,0}(\infty)]\right)~, \label{eq:3wfe_excl}
	\end{align}
	where $\omega_{g,\ell}^{e}$ is the coefficient from the auxiliary regression \eqref{eq:aux_regression}, evaluated at $t = g + \ell$. The weights satisfy
	\begin{enumerate}
		\item[{(i)}] own-period weights sum to one{,} $\sum_{g \in \calGtrg} \omega_{g,e}^{e} = 1${;}
		\item[{(ii)}] other included periods sum to zero{,} $\sum_{g \in \calGtrg} \omega_{g,\ell}^{e} = 0$ for each $\ell \neq e${,} $\ell \neq -1${;}
		\item[{(iii)}] excluded period sums to negative one{,} $\sum_{g \in \calGtrg} \omega_{g,-1}^{e} = -1${;}
		\item[{(iv)}] never-treated units receive zero weight{,} $\omega_{\infty,\ell}^{e} = 0$ for all $e,\ell${.}
	\end{enumerate}
\end{prop}

\Cref{prop:3wfe_es_weights} is the DDD analogue of \citeauthor{sun_estimating_2021}'s (\citeyear{sun_estimating_2021}) Proposition~1. The weight properties carry the same interpretation. The own-period weights (i) sum to one, so $\mu_e$ is a weighted average of the ``own'' terms---but the individual weights $\omega_{g,e}^{e}$ need not be non-negative, so this average may lie outside the convex hull of the individual terms. The other-included weights (ii) sum to zero for each event-time $\ell \neq e$, $\ell \neq -1$, so under homogeneity across cohorts at a given $\ell$ the contamination from included periods cancels. The excluded-period weights (iii) sum to $-1$, reflecting the normalization at $e = -1$. The key DDD-specific feature is that the auxiliary regression \eqref{eq:aux_regression} includes group-by-time fixed effects $\delta_{S_i,t}$, so the FWL residual of $R_e$ removes not just unit and time means but also the group-level treatment rate at each period---an additional demeaning step absent from the standard DiD setting. The proof is omitted here because it resembles that of \Cref{prop:hw_es_decomp}.

I now progressively add identifying assumptions to obtain the decomposition in terms of causal parameters.

\begin{prop}[Under DDD-PCT only]\label{prop:3wfe_es_pt}
	Under {\Cref{as:dddpct}}, the population regression coefficient $\mu_e$ is a linear combination of $\CATT(g,\ell)$ with the same weights as in Proposition~{\ref{prop:3wfe_es_weights}}
	\begin{equation}\label{eq:contamination}
		\mu_e = \sum_{g \in \calGtrg}\sum_{\ell \neq -1} \omega_{g,\ell}^{e} \CATT(g,\ell)~.
	\end{equation}
\end{prop}

Two things of note: first, $\mu_e$ receives nonzero weight on $\CATT(g,\ell)$ for $\ell \neq e$---treatment effects from other event-times contaminate the coefficient nominally associated with event-time $e$; second, the own-period weights $\omega_{g,e}^{e}$ may be negative, so the coefficient may lie outside the convex hull of the cohort-specific $\CATT(g,e)$ values. Proof is omitted here because it resembles that of \Cref{prop:hw_es_pt}.

\begin{prop}[Under DDD-PCT and no anticipation]\label{prop:3wfe_es_noanticip}
	Under {\Cref{as:dddpct}} and~{\Cref{as:noanticipation}}, the population coefficient $\mu_e$ satisfies
	\begin{equation}\label{eq:3wfe_noanticip}
		\mu_e = \sum_{g \in \calGtrg}\sum_{\substack{\ell \geq 0}} \omega_{g,\ell}^{e} \CATT(g,\ell)~.
	\end{equation}
	Pre-treatment $\CATT(g,\ell) = 0$ for $\ell < 0$ drop out. However, for pre-treatment event-times $e < 0$, $\mu_e$ is generally nonzero because it depends on post-treatment $\CATT(g,\ell)$ for $\ell \geq 0$ through the excluded-period weight property {(iii)} of Proposition~{\ref{prop:3wfe_es_weights}}.
\end{prop}

\Cref{prop:3wfe_es_noanticip} has a striking implication. Even under the maintained identifying assumptions of this paper (DDD-PCT and no anticipation), the pre-period coefficients $\mu_e$ for $e < 0$ from the 3WFE event-study regression \eqref{eq:3wfe_eventstudy} are contaminated by post-treatment effects. A researcher who finds $\widehat{\mu}_e \neq 0$ for $e < 0$ may incorrectly conclude that the DDD PCT assumption is violated, when in fact the pre-period coefficient reflects heterogeneous post-treatment effects bleeding through the implicit weights. Conversely, $\widehat{\mu}_e = 0$ for $e < 0$ does not validate DDD-PCT, since the contaminating effects may happen to cancel. I return to this point in the subsection on pretrend testing below. Proof is omitted here because it resembles that of \Cref{prop:hw_es_noanticip}.

\begin{prop}[Under DDD-PCT and treatment effect homogeneity]\label{prop:3wfe_es_homo}
	Under {\Cref{as:dddpct}} and the restriction that $\CATT(g,\ell) = \ATT_\ell$ for all $g$ {(}treatment effects depend on exposure duration but not on the cohort{)}, the population coefficient $\mu_e$ simplifies to
	\begin{equation}\label{eq:3wfe_homo}
		\mu_e = \sum_{\ell \neq -1} \overline{\omega}_\ell^{e}  \ATT_\ell~, \qquad \overline{\omega}_\ell^{e} \equiv \sum_{g \in \calGtrg} \omega_{g,\ell}^{e}~.
	\end{equation}
	By the weight properties of Proposition~{\ref{prop:3wfe_es_weights}}, $\overline{\omega}_e^{e} = 1$, $\overline{\omega}_{\ell}^{e} = 0$ for $\ell \neq e$, $\ell \neq -1$, and $\overline{\omega}_{-1}^{e} = -1$. Therefore
	\begin{equation}\label{eq:3wfe_homo_simplified}
		\mu_e = \ATT_e - \ATT_{-1}~.
	\end{equation}
	Under the additional restriction of no anticipation {(}$\ATT_{-1} = 0${)}, $\mu_e = \ATT_e$---the 3WFE event-study specification recovers the homogeneous event-time treatment effect.
\end{prop}

\Cref{prop:3wfe_es_homo} clarifies the conditions under which the 3WFE event-study regression is valid. The specification \eqref{eq:3wfe_eventstudy} recovers interpretable causal parameters only under the joint restrictions of DDD-PCT, treatment effect homogeneity across cohorts, and no anticipation. In practice, these conditions are rarely satisfied simultaneously. The decomposition in Propositions \ref{prop:3wfe_es_weights}--\ref{prop:3wfe_es_homo} provides a precise diagnosis of what goes wrong when each condition fails. Proof is omitted here because it resembles that of \Cref{prop:hw_es_homo}.

I also note that the aggregate weight $\overline{\omega}_\ell^{e} = \sum_g \omega_{g,\ell}^{e}$ in \Cref{prop:3wfe_es_homo} is the population regression coefficient from regressing $R_\ell(i,t) = \one\{t - S_i = \ell\} Q_i$ on the event-time indicators and fixed effects in \eqref{eq:3wfe_eventstudy}. This follows from summing the auxiliary regression \eqref{eq:aux_regression} over $g$ at fixed $\ell$. The aggregate weight is a function of the distribution of treatment cohorts and can be estimated from the sample analogue of this regression.

\section{Proof of Validity of CRVE in Stacked OLS Regressions} \label[appendix]{app:proof-CRVE-validity}

I provide a justification of the use of cluster-robust variance estimator (CRVE), clustered at the level of treatment, in the stacked OLS regression. I establish that the CRVE is a consistent estimator of the asymptotic variance of the stacked estimator under the sampling and moment conditions described in the main text. Specifically, I maintain \Cref{as:sampling}, \Cref{as:shares}, and \Cref{as:moments4} throughout this section. As before, the number of units in each stack $g$ is denoted $n_g$, and the number of units in each $(s,q)$ cell within that stack is denoted $n_{s,q,g}$, and $n_g/n \to \lambda_g > 0$ and $n_{s,q,g}/n_g \to \pi_{s,q,g} > 0$ under \Cref{as:shares}.

The stacked event-study coefficient $\widehat{\tau}_e$ is obtained by estimating a fully saturated ordinary least squares regression on the pooled stacked dataset. By Frisch-Waugh-Lovell, the multivariate regression can be analytically reduced to a bivariate regression of the outcome on the residualized treatment indicator. Let $\widetilde{R}_{i,g} = \widetilde{R}_e(i, g+e, g)$ denote the FWL residual of the treatment indicator for unit $i$ in stack $g$ at event-time $e$, obtained after projecting out the stack-by-group-by-time and stack-by-eligibility-by-time fixed effects. Because these fixed effects are mutually orthogonal across stacks and times, the residual $\widetilde{R}_{i,g}$ is strictly zero for all periods other than $t = g+e$ and for all stacks other than $g$. The estimated coefficient $\widehat{\tau}_e$ can therefore be written exactly as
\begin{equation}
	\widehat{\tau}_e = \frac{\sum_{1 \leq i \leq n} \sum_{g \in \calGtrg(e)} \widetilde{R}_{i,g} Y_{i,g+e}}{\sum_{1 \leq i \leq n} \sum_{g \in \calGtrg(e)} \widetilde{R}_{i,g}^2} ~.
\end{equation}
The population analogue of this coefficient, $\tau_e$, is defined by replacing the sample sums with their population expectations. The estimation error of the stacked estimator can be isolated by substituting the outcome $Y_{i,g+e}$ with the sum of its conditional expectation and the regression error. The OLS residual for unit $i$ in stack $g$ is defined as $\widehat{\varepsilon}_{i,g+e,g} = Y_{i,g+e} - \widehat{Y}_{i,g+e,g}$, where $\widehat{Y}_{i,g+e,g}$ is the OLS fitted value. Let $\widehat{s}_i(e)$ denote the unit-level estimated score, which aggregates the product of the residualized treatment indicator and the regression error across all stacks in which unit $i$ appears:
\begin{equation}
	\widehat{s}_i(e) = \sum_{g \in \calGtrg(e)} \widetilde{R}_{i,g} \widehat{\varepsilon}_{i,g+e,g} ~.
\end{equation}
Multiplying by $\sqrt{n}$ to study the asymptotic distribution, the scaled estimation error takes the form
\begin{equation} \label{eq:app_est_error}
	\sqrt{n}(\widehat{\tau}_e - \tau_e) = \left( \frac{1}{n} \sum_{1 \leq i \leq n} \sum_{g \in \calGtrg(e)} \widetilde{R}_{i,g}^2 \right)^{-1} \frac{1}{\sqrt{n}} \sum_{1 \leq i \leq n} \widehat{s}_i(e) ~.
\end{equation}

The CRVE for the (square root of the) variance of $\widehat{\tau}_e$, clustered at the level of the original unit $i$, is constructed using the empirical sandwich formula
\begin{equation}
	\widehat{\Var}(\widehat{\tau}_e) = \left( \sum_{1 \leq i \leq n} \sum_{g \in \calGtrg(e)} \widetilde{R}_{i,g}^2 \right)^{-2} \sum_{1 \leq i \leq n} \widehat{s}_i(e)^2 ~.
\end{equation}
To prove that this estimator provides valid inference, I must establish that $n\widehat{\Var}(\widehat{\tau}_e)$ converges in probability to the true asymptotic variance of the scaled estimation error in \eqref{eq:app_est_error}. The convergence of the bread is immediate. By the WLLN, the scaled denominator $n^{-1} \sum_{1 \leq i \leq n} \sum_{g \in \calGtrg(e)} \widetilde{R}_{i,g}^2$ converges in probability to its population expectation $V = \E[\sum_g \widetilde{R}_{i,g}^2]$. The critical step is to demonstrate that the scaled meat of the sandwich, given by $n^{-1} \sum_{1 \leq i \leq n} \widehat{s}_i(e)^2$, converges to the population variance of the score, defined as $\Omega = \E[(s_i^*(e))^2]$. The population score for unit $i$ is $s_i^*(e) = \sum_{g \in \calGtrg(e)} \widetilde{R}_{i,g} \varepsilon_{i,g+e,g}$, where $\varepsilon_{i,g+e,g} = Y_{i,g+e} - \E[Y_{i,g+e} \mid S_i, Q_i]$ is the true conditional error term.

I first characterize the rate of convergence of the estimated regression residuals to the population errors. Because the stacked event-study regression is fully saturated by the inclusion of stack-by-group-by-time and stack-by-eligibility-by-time fixed effects, the OLS fitted value for any observation is exactly the sample mean of the outcome within that specific stack-cell-time partition. Specifically, for a unit $i$ belonging to cell $(s,q)$ within stack $g$, the fitted value is $\widehat{Y}_{i,g+e,g} = \overline{Y}_{s,q,g+e,g}$. The OLS residual can therefore be algebraically decomposed as $\widehat{\varepsilon}_{i,g+e,g} = \varepsilon_{i,g+e,g} - \delta_{s,q,g}$, where $\delta_{s,q,g} = \overline{Y}_{s,q,g+e,g} - \E[Y_{i,g+e} \mid S_i=s, Q_i=q]$ captures the estimation error of the cell mean. By the CLT and \Cref{as:shares}, the sample mean of i.i.d.\ outcomes converges to its population expectation at a $\sqrt{n}$ rate, implying that $\delta_{s,q,g} = \bigOp(n^{-1/2})$. The unit-level estimated score can then be rewritten as the population score minus an estimation error term
\begin{equation}
	\widehat{s}_i(e) = s_i^*(e) - \Delta_i,
\end{equation}
where $\Delta_i \equiv \sum_{g \in \calGtrg(e)} \widetilde{R}_{i,g} \delta_{s(i), q(i), g}$. Since $\widetilde{R}_{i,g}$ is bounded and $\delta_{s,q,g} = \bigOp(n^{-1/2})$, it follows that the aggregate unit-level discrepancy is also $\Delta_i = \bigOp(n^{-1/2})$.

I next analyze the convergence of the empirical score variance by expanding the scaled meat of the sandwich estimator. Expanding the square of the estimated score yields three distinct components:
\begin{equation} \label{eq:meat_decomp}
	\frac{1}{n} \sum_{1 \leq i \leq n} \widehat{s}_i(e)^2 = \underbracket{\frac{1}{n} \sum_{1 \leq i \leq n} (s_i^*(e))^2}_{(i)} - \underbracket{\frac{2}{n} \sum_{1 \leq i \leq n} s_i^*(e) \Delta_i}_{(ii)} + \underbracket{\frac{1}{n} \sum_{1 \leq i \leq n} \Delta_i^2}_{(iii)} ~.
\end{equation}
I evaluate each of these three terms in turn. The first term (i) is the sample average of the squared population score. Under the finite fourth moments condition stated in \Cref{as:moments4}, the squared score has a finite variance, allowing the application of the WLLN to conclude that $n^{-1} \sum_{1 \leq i \leq n} (s_i^*(e))^2 \pto \E[(s_i^*(e))^2] = \Omega$. 

The third term (iii) represents the average of the squared discrepancies. Because $\Delta_i = \bigOp(n^{-1/2})$ for all units, its square is $\Delta_i^2 = \bigOp(n^{-1})$. The sample average of these squared terms, $n^{-1} \sum_{1 \leq i \leq n} \Delta_i^2$, is bounded by the maximum over the finite number of cells of $\bigOp(n^{-1})$. Consequently, this term vanishes asymptotically, satisfying $\bigOp(n^{-1}) = \smallOp(1)$.

The second term (ii) is the cross-product between the population score and the estimation error. Substituting the definition of $\Delta_i$, this term can be rearranged by exchanging the summations over units and stacks
\begin{equation}
	\frac{2}{n} \sum_{1 \leq i \leq n} s_i^*(e) \Delta_i = 2 \sum_{g \in \calGtrg(e)} \sum_{(s,q)} \delta_{s,q,g} \left( \frac{1}{n} \sum_{i \in \text{cell} (s,q,g)} s_i^*(e) \widetilde{R}_{i,g} \right) ~.
\end{equation}
Within each cell partition $(s,q,g)$, the expression in the parentheses is the sample average of the product $s_i^*(e) \widetilde{R}_{i,g}$. The population expectation of this product conditional on the cell is exactly zero, because $\E[s_i^*(e) \mid S_i=s, Q_i=q] = 0$ by the definition of the true error term $\varepsilon_{i,g+e,g}$, and $\widetilde{R}_{i,g}$ is a constant within the cell. By the CLT, the sample average of a zero-mean random variable scaled by $n$ is bounded in probability by $\bigOp(n^{-1/2})$. The cross-term is therefore the product of a finite number of components, each of which multiplies $\delta_{s,q,g} = \bigOp(n^{-1/2})$ by an average that is $\bigOp(n^{-1/2})$. The resulting product is $\bigOp(n^{-1})$, which implies that the entire cross-term converges in probability to zero, or $\smallOp(1)$.

Combining these three limits establishes that the empirical variance of the estimated score converges to the true variance of the population score, $n^{-1} \sum_{1 \leq i \leq n} \widehat{s}_i(e)^2 \pto \Omega$. Finally, applying Slutsky's theorem to the full CRVE sandwich estimator demonstrates that the scaled variance converges to the exact asymptotic variance of the OLS estimator
\begin{equation}
	n \widehat{\Var}(\widehat{\tau}_e) = \left( \frac{1}{n} \sum_{1 \leq i \leq n} \sum_{g \in \calGtrg(e)} \widetilde{R}_{i,g}^2 \right)^{-2} \left( \frac{1}{n} \sum_{1 \leq i \leq n} \widehat{s}_i(e)^2 \right) \pto V^{-2} \Omega ~.
\end{equation}
This rigorously justifies the use of cluster-robust standard errors. By calculating the score at the level of the original unit $i$ and implicitly summing the score components across all stacks in which the unit appears before squaring, the CRVE automatically and correctly accounts for the cross-stack dependence induced by the reuse of comparison units.

\section{Auxiliary Results} \label[appendix]{app:auxiliary}

\subsection{Technical Results} \label[appendix]{app:technical-results}

\begin{lemma}[Lindeberg Condition for Bounded Influence Functions]\label{lem:lindeberg}
	If $\{W_i\}_{i=1}^n$ are i.i.d.\ with $\psi_i = \psi(W_i)$ satisfying $\E[\psi_i] = 0$ and $\sigma^2 = \E[\psi_i^2] \in (0, \infty)$, then the Lindeberg condition holds:
	\begin{equation}\label{eq:lem_lindeberg}
		\forall  \varepsilon > 0, \qquad \frac{1}{\sigma^2 n}\sum_{1 \leq i \leq n} \E\left[\psi_i^2  \one\{|\psi_i| > \varepsilon  \sigma\sqrt{n}\}\right] \to 0.
	\end{equation}
\end{lemma}

\begin{proof}
	Since $\{W_i\}$ are i.i.d., every summand is identical, so
	\[
	\frac{1}{\sigma^2 n}\sum_{1 \leq i \leq n} \E\left[\psi_i^2  \one\{|\psi_i| > \varepsilon\sigma\sqrt{n}\}\right] = \frac{1}{\sigma^2}  \E\left[\psi_1^2  \one\{|\psi_1| > \varepsilon\sigma\sqrt{n}\}\right]~.
	\]
	As $n \to \infty$, $\varepsilon\sigma\sqrt{n} \to \infty$, so $\one\{|\psi_1| > \varepsilon\sigma\sqrt{n}\} \to 0$ almost surely (since $|\psi_1| < \infty$ a.s.). The integrand $\psi_1^2  \one\{|\psi_1| > \varepsilon\sigma\sqrt{n}\} \leq \psi_1^2$ is dominated by the integrable function $\psi_1^2$ (since $\E[\psi_1^2] = \sigma^2 < \infty$). By the dominated convergence theorem, $\E[\psi_1^2  \one\{|\psi_1| > \varepsilon\sigma\sqrt{n}\}] \to 0$.
\end{proof}

\begin{remark}[Multiplier bootstrap] \label{rem:multiplier-bootstrap}
	For uniform confidence bands across all event-times, one can implement the multiplier bootstrap proceeds as follows.
	For each bootstrap replication $b = 1, \ldots, B$:
	\begin{enumerate}
		\item[{(i)}] Draw i.i.d.\ multiplier weights $\{\xi_i^{(b)}\}_{i=1}^{n}$ from a distribution satisfying $\E[\xi_i] = 0$, $\E[\xi_i^2] = 1$, and $\E[\xi_i^4] < \infty$. Standard choices are $\xi_i \sim N(0,1)$ or $\xi_i \in \{-1, +1\}$ with equal probability. Importantly, the \emph{same} weight $\xi_i^{(b)}$ is used for unit $i$ across all stacks containing that unit.
		\item[{(ii)}] Define
		\begin{equation}\label{eq:phi_hat_def}
			\widehat\phi_i(e) \equiv \sum_{g \in \calGtrg(e)} \widehat\omega_g(e)  \widehat\psi_{\Stack_g}\big(W_i;  g,  g + e,  g_c(g)\big),
		\end{equation}
		the aggregated influence function of unit $i$ at event-time $e$, summed across all stacks the unit belongs to. Then, calculate
		\begin{equation}\label{eq:bootstrap_statistic}
			\widehat{\tau}_e^{*,(b)} = \frac{1}{n}\sum_{i=1}^{n} \xi_i^{(b)} \widehat{\phi}_i(e)~.
		\end{equation}
		The bootstrap variance estimator is $\widehat{V}_{\mathrm{boot}}(e) = B^{-1}\sum_{b=1}^{B} [\widehat{\tau}_e^{*,(b)}]^2$, which is asymptotically equivalent to $\widehat{V}(e)/n$.
	\end{enumerate}
	
	To construct confidence bands, let $\calE = \{-L, \ldots, K\} \setminus \{-1\}$ be the set of event-times. Compute the bootstrap critical value $c_\alpha$ as the $(1-\alpha)$-quantile of $\max_{e \in \calE} |\widehat{\tau}_e^{*,(b)}| / \sqrt{\widehat{V}_{\mathrm{boot}}(e)}$ across $b = 1, \ldots, B$. The simultaneous $(1-\alpha)$th confidence band is $\{\widehat{\tau}_e \pm c_\alpha \sqrt{\widehat{V}_{\mathrm{boot}}(e)}\}_{e \in \calE}$.
	
	The essential design feature of this bootstrap is the use of a {common} multiplier weight $\xi_i$ for each unit $i$ across all stacks. When unit $i$ belongs to stacks $\Stack_g$ and $\Stack_{g'}$, the bootstrap covariance between the two within-stack contributions is $\E[\xi_i^2] \widehat{\psi}_{\Stack_g}(W_i) \widehat{\psi}_{\Stack_{g'}}(W_i) = \widehat{\psi}_{\Stack_g}(W_i) \widehat{\psi}_{\Stack_{g'}}(W_i)$, which correctly mimics the population covariance arising from unit $i$'s shared membership. Drawing independent multipliers for each stack--unit pair would break this dependence and yield invalid inference.
\end{remark}

\subsection{Additional Remarks} \label[appendix]{app:addl-remarks}

\begin{remark}[{Why the partial residual can be negative}]
	The partial residual $\widetilde{R}_j(i,t)$ is obtained from the aggregate indicator $R_j(i,t) = \one\{t - S_i = j\} Q_i$ by two successive projections. First, the three-way demeaning \eqref{eq:hw_demean} removes unit effects, group-by-time effects, and eligibility-by-time effects, yielding the three-way-demeaned indicator $\ddot{R}_j(i,t)$. Second, the projection onto the space spanned by the other demeaned event-time indicators $\{\ddot{R}_{e'}\}_{e' \neq j, -1}$ is subtracted, yielding $\widetilde{R}_j(i,t) = \ddot{R}_j(i,t) - \sum_{e' \neq j, -1} \widehat{\gamma}_{j,e'}  \ddot{R}_{e'}(i,t)$, where $\widehat{\gamma}_{j,e'}$ is the population regression coefficient from the projection.
	
	Consider a treated-eligible unit in cohort $g$ at time $t = g + \ell$ with $\ell \neq j$. At this time, the unit is at event-time $\ell$ relative to its own treatment, not at event-time $j$. The aggregate indicator $R_j(i, g+\ell)$ is zero for this unit (since $g + \ell - g = \ell \neq j$). After three-way demeaning, $\ddot{R}_j(i, g+\ell)$ can be nonzero because the demeaning subtracts group-time and eligibility-time means that are contaminated by other cohorts. After further partialling out other event-time indicators, $\widetilde{R}_j(i, g+\ell)$ can take either sign, depending on the correlation structure between $\ddot{R}_j$ and $\ddot{R}_{e'}$ at that time.
	
	The sign of $\widetilde{r}_{j; g,1}^{(g+\ell)}$ is driven by the {calendar-time composition} of the design. Under staggered adoption, at any given time $t$, different cohorts occupy different positions on event-time. The indicator $R_j(i,t)$ takes the value one for the cohort $g_t$ satisfying $g_t + j = t$ (if such a cohort exists). After partialling out, the residual at this time reflects the deviation of cohort $g_t$'s eligible share from the average eligible share across all cohorts that are ``active'' at time $t$. When this deviation is offset by the partialling-out step, the residual at a {different} cohort's cell can flip sign.
\end{remark}

\begin{ex}[{Explicit weights in a toy example}]
	{Consider a balanced panel with $T = 4$ and two treatment cohorts $g_1 = 2$, $g_2 = 3$ (no never-treated units), with equal cohort shares ($\pi_1 = \pi_2 = 1/2$) and equal eligibility shares ($p_1 = p_2 = p$). Both cohorts have a proper pre-period at $t = 1$, so the $e = -1$ normalization is well-defined for each.}
	The event-study specification includes only $e = 0$ and excludes $e = -1$ ($L = 1$, $K = 0$). I show below that
	\begin{equation}\label{eq:hw_example_target}
		\alpha_0 = \frac{1}{2}\CATT(g_1, 0) + \frac{1}{2}\CATT(g_2, 0) - \left( \frac{1}{2}\CATT(g_1, 1) + \frac{1}{2}\CATT(g_2, -1) \right)~.
	\end{equation}
	Under no anticipation ($\CATT(g_2, -1) = 0$), \eqref{eq:hw_example_target} collapses to
	\[
	\alpha_0 = \frac{1}{2}\CATT(g_1, 0) + \frac{1}{2}\CATT(g_2, 0) - \frac{1}{2}\CATT(g_1, 1)~.
	\]
	To see this, notice that since only $R_0$ is included, $\widetilde{R}_0 = \ddot{R}_0$, and \eqref{eq:hw_weight_ratio} reduces to
	\begin{equation}\label{eq:hw_example_setup}
		\omega_{g,\ell}^{0,\star} = \frac{p_{g,1}  \ddot{r}_{0; g,1}^{(g+\ell)}}{\sigma_0^{2,\star}}~,
	\end{equation}
	with $p_{g,1} = \Prob(S_i = g, Q_i = 1) = p/2$ under the parameters stated earlier.
	The indicator $R_0(i,t)$ equals $Q_i$ for cohort $g_1$ at $t = 2$ and for cohort $g_2$ at $t = 3$, and is zero otherwise. The relevant marginal shares of $R_0$ needed to form the three-way-demeaned residual are: the unit-level mean is $1/4$ for treated-eligible cells in either cohort and $0$ elsewhere; the group-by-time mean $\bar R_{0,g,t}$ equals $p$ at $(g_1, 2)$ and $(g_2, 3)$ and zero elsewhere; the eligibility-by-time mean $\bar R_{0,Q=1,t}$ equals $1/2$ at $t \in \{2, 3\}$ and zero elsewhere; the group all-time mean is $p/4$ for each group; the eligibility all-time mean is $1/4$ for $Q = 1$; the time-only mean is $p/2$ at $t \in \{2, 3\}$; and the total mean is $p/4$.
	Computing the three-way-demeaned residual for a treated-eligible unit in cohort $g_1$ at its own event-time $t = g_1 = 2$,
	\begin{align}
		\ddot{R}_0(g_1, 1, 2) &= \underbracket{1}_{R_0} - \underbracket{1/4}_{\overline{R}_{0,(g_1,1),\cdot}} - \underbracket{p}_{\overline{R}_{0,g_1,2}} - \underbracket{1/2}_{\overline{R}_{0,Q=1,2}} + \underbracket{p/4}_{\overline{R}_{0,g_1,\cdot}} + \underbracket{1/4}_{\overline{R}_{0,Q=1,\cdot}} + \underbracket{p/2}_{\overline{R}_{0,\cdot,2}} - \underbracket{p/4}_{\overline{R}_{0,\cdot\cdot}} = \frac{1-p}{2}~, \label{eq:hw_example_residual}
	\end{align}
	and by symmetry $\ddot{r}_{0; g_2,1}^{(3)} = (1-p)/2$ at cohort $g_2$'s own event-time. The own-period weights are therefore equal across cohorts, and property~{(i)} of Proposition \ref{prop:hw_es_decomp} forces each to equal $1/2$. Now turn to the cross-period cell: at $t = 3$, cohort $g_1$ is at event-time $\ell = 1$ (one period after its treatment), and the analogous calculation gives
	\begin{equation}\label{eq:hw_example_cross}
		\ddot{r}_{0; g_1,1}^{(3)} = {0 - \frac{1}{4} - 0 - \frac{1}{2} + \frac{p}{4} + \frac{1}{4} + \frac{p}{2} - \frac{p}{4}} = -\frac{1-p}{2}~,
	\end{equation}
	exactly the negative of the own-period residual.
	{An identical calculation at $t = 2$ for cohort $g_2$ (its pre-period $\ell = -1$) yields $\ddot{r}_{0; g_2,1}^{(2)} = -(1-p)/2$. The partial-residual variance is $\sigma_0^{2,\star} = p(1-p)/2$ (summing over the four nonzero cells in each of the eligible and ineligible strata), so \eqref{eq:hw_example_setup} gives $\omega_{g_1,1}^{0,\star} = \omega_{g_2,-1}^{0,\star} = -1/2$, which yields \eqref{eq:hw_example_target}.}
	
	The cross-period weights have the same magnitude as the own-period weights but opposite sign. Under dynamic effects with $\CATT(g_1, 1) > 0$ and no anticipation, $\alpha_0$ is biased {downward} relative to the average of the two own-period CATTs, with the bias growing in $\CATT(g_1,1)$; only under event-time homogeneity $\CATT(g,\ell) = \ATT_\ell$ does the contamination vanish, as formalized in Proposition \ref{prop:hw_es_homo} below.
\end{ex}

\end{document}